\documentclass[11pt]{amsart}

\usepackage{amsmath, amsfonts, amsthm,amssymb}
\usepackage{verbatim, comment}
\usepackage{natbib, float}
\usepackage{bbm}
\usepackage[colorlinks=TRUE,citecolor=blue]{hyperref}
\usepackage{enumitem,graphicx,psfrag,epsf,placeins,subcaption}
\usepackage{algorithm}
\usepackage{algpseudocode}
\usepackage{multirow}
\usepackage{bm}
\usepackage{tabularx}
\usepackage{lmodern}

\newcommand\updatedText[1]{{\color{black}#1}}
\newcommand{\bias}{\mathbf{d}}
    
\newcommand{\Hv}{\bm{H}_{v}}

\newcommand{\bftab}{\fontseries{b}\selectfont}

\usepackage{booktabs}
\usepackage{longtable}
\usepackage{array}
\usepackage{multirow}
\usepackage{wrapfig}
\usepackage{float}
\usepackage{colortbl}
\usepackage{pdflscape}
\usepackage{tabu}
\usepackage{threeparttable}
\usepackage{threeparttablex}
\usepackage[normalem]{ulem}
\usepackage{makecell}

\usepackage[margin = 3cm]{geometry}

\newtheorem{theorem}{Theorem}
\newtheorem{lemma}{Lemma}
\newtheorem{corollary}{Corollary}

\newtheorem{remark}{Remark}[section]
\newtheorem{example}[remark]{Example}
\newtheorem{assumption}{Assumption}

\numberwithin{equation}{section}

\newcommand{\independent}{\protect\mathpalette{\protect\independenT}{\perp}}
\def\independenT#1#2{\mathrel{\rlap{$#1#2$}\mkern2mu{#1#2}}}

\usepackage{color}

\allowdisplaybreaks[4]

\usepackage{tikz}
\usetikzlibrary{shapes,arrows,positioning}

\newcommand{\E}{{\mathbb E}}       
\newcommand{\pa}{{\rm pa}}         
\newcommand{\pr}{{\rm pr}_\theta}  
\newcommand{\nd}{{\rm nd}}         
\newcommand{\ch}{{\rm ch}}         
\newcommand{\an}{{\rm an}}	   
\newcommand{\de}{{\rm de}}	   
\newcommand{\var}{{\rm var}}       

\newcommand{\vecv}{\mathbf{Y}_v}       
\newcommand{\vecY}{\mathbf{Y}}
\newcommand{\vecu}{\mathbf{Y}_u}
\newcommand{\vecU}{\mathbf{Y}_U}       
\newcommand{\vecUint}{\mathbf{Y}_{U.1}}       
\newcommand{\tildevecv}{ \mathbf{\tilde Y}_v}       
\newcommand{\vecepsv}{ \bm{\varepsilon}_{v}}

\newcommand{\cov}{{\rm cov}}
\newcommand{\vecEta}{\bm{\eta}_{v \setminus U}}
\newcommand{\hatEta}{\bm{ \hat\eta}_{v \setminus U}}
\newcommand{\B}{\mathbf{Z}}

\title{Confidence sets for Causal Orderings}

\author[Wang]{Y. Samuel Wang}
\address{Cornell University}
\email{ysw7@cornell.edu}
\author[Kolar]{Mladen Kolar}
\address{University of Southern California}
\email{mkolar@marshall.usc.edu}
\author[Drton]{Mathias Drton}
\address{TU Munich}
\email{mathias.drton@tum.de}

\begin{document}

\begin{abstract}
Causal discovery procedures aim to deduce causal relationships among variables in a multivariate dataset. While various methods have been proposed for estimating a single causal model or a single equivalence class of models, less attention has been given to quantifying uncertainty in causal discovery in terms of confidence statements. \updatedText{A primary challenge in causal discovery of directed acyclic graphs is determining a causal ordering among the variables, and our work offers a framework for constructing confidence sets of causal orderings that the data do not rule out. Our methodology specifically applies to identifiable structural equation models with additive errors and is based on a residual bootstrap procedure to test the goodness-of-fit of causal orderings.} We demonstrate the asymptotic validity of the confidence set constructed using this goodness-of-fit test and explain how the confidence set may be used to form sub/supersets of ancestral relationships as well as confidence intervals for causal effects that incorporate model uncertainty.
\end{abstract}

\maketitle

\section{Introduction}\label{sec:intro}
Inferring causal relations as opposed to mere associations is a problem that is not only of intrinsic scientific interest but also helps predict how an observed system might change under intervention~\citep{peters2017elements}. When randomized controlled trials are infeasible, methods for \emph{causal discovery}---the problem of estimating a causal model from observational data---become valuable tools for hypothesis generation and acceleration of scientific progress. Examples of applications include systems biology~\citep{sachs2005causal}, neuroscience~\citep{shen2020challenges}, and climate modeling~\citep{nowack2020causal}.

\updatedText{The causal models we consider} may be represented as a \emph{directed acyclic graph} (DAG), and---leveraging this representation---\updatedText{many problems in} causal discovery can be cast as recovery of the appropriate DAG. The first step in causal discovery is \emph{identification}: determining appropriate assumptions under which the causal model can be recovered from population information; see, e.g., \citet{shohei2006lingam, loh2014learning, peters2014additive}. The next step is providing a method to \emph{estimate} the causal graph from data; see, e.g., \citet{buhlmann2014cam, chen2019equal, wang2020high}. Once an estimation procedure is established, it is natural to question the estimation \emph{uncertainty}. Uncertainty quantification and the ability to test identifying assumptions are essential for trustworthy estimation of causal graphs and help to determine whether key modeling assumptions are appropriate. Nonetheless, the literature on frequentist causal discovery, with a few exceptions \cite[e.g.,][]{strobl2019pvalues}, only outputs a point estimate in the form of a DAG or single equivalence class.

\updatedText{In the settings we consider,} given a causal ordering of the variables, causal discovery reduces to variable selection in a sequence of regressions. Thus, the key difficulty lies in inferring the causal ordering; this motivates the issue we address in this paper: developing a procedure that provides a confidence set for causal orderings.

\subsection{Setup}\label{sec:DAG}

We represent a causal model for the random vector $Y=(Y_1,\dots,Y_p)$ with a DAG $G=(V,E)$, where each node $v$ in the vertex set $V=[p]$ indexes a random variable $Y_v$.  An edge $u \rightarrow v \in E$ indicates that $Y_u$ has a direct causal effect on $Y_v$, and we say that $u$ is a \emph{parent} of its \emph{child} $v$. If there exists a directed path in $G$ from $u$ to $v$, then $u$ is an \emph{ancestor} of its \emph{descendant} $v$. We denote the sets of parents, children, ancestors, and descendants of node $v$ by $\pa(v)$, $\ch(v)$, $\an(v)$, and $\de(v)$, respectively. 
The models we consider take the form of a 
recursive structural equation model (SEM) with additive noise:
\begin{equation} \label{eq:SEM}
	Y_{v} = f_v\left((Y_{u})_{u \in \pa(v)}\right) + \varepsilon_{v}, \qquad v\in V,
\end{equation}
where the $f_v$ are unknown 
and the errors $\{\varepsilon_v\}_{v = 1}^p$ are mean zero and mutually independent. 

In a fully general SEM, the DAG may only be identified from observational data up to a Markov equivalence class---a collection of graphs that imply the same set of conditional independence relations~\citep{spirtes2000causation}.  As the different graphs in the equivalence class may have contradicting causal interpretations, it is also of interest to work with restricted SEMs in which the DAG itself becomes identifiable~\citep[Chap.~18.6.3]{handbook}. 
\updatedText{Specifically, for the model in \eqref{eq:SEM},  which assumes additive errors, 
the DAG becomes identifiable when 
$f_v$ are non-linear or the errors $\varepsilon_v$ are non-Gaussian. Our methodology is tailored to these settings. In contrast, the linear Gaussian case, also allowed under \eqref{eq:SEM}, features the same Markov equivalence classes as the general nonparametric model. In this case, our procedure fails gracefully in the sense that the confidence statements we establish in Section~\ref{sec:theory} still hold, but the confidence sets will be uninformative because we have trivial power to reject incorrect orderings.}

We focus on a causal ordering for the variables in the model given by DAG~$G$; i.e., a total ordering of $V$ 
where variables that appear later 
have no causal effect on earlier variables.  We may identify each possible ordering with a permutation $\theta:V\to V$, where $\theta$ yields a causal ordering for $G$ if and only if $\theta(u) < \theta(v)$ implies that $v \not \in \an(u)$. 
In general, a causal ordering is not unique, and, letting $\mathcal{S}_V$ be the set of all permutations of $V$, we denote the set of all causal orderings
$
\Theta(G) = \left\{\theta \in\mathcal{S}_V\, : \, \theta(u) < \theta(v) \text{ only if } v \not \in \an(u) \right\}.
$

\subsection{Contribution}
Let $\vecY$ be a sample drawn from the SEM in~\eqref{eq:SEM}, and let $\alpha\in(0,1)$. We propose a procedure to construct a $1 - \alpha$ \emph{confidence set of causal orderings}, $\hat \Theta(\vecY, \alpha)$, where $\hat \Theta(\vecY, \alpha)\subseteq \mathcal{S}_V$. Specifically, our procedure inverts a goodness-of-fit test for a causal ordering and returns the set of all orderings that are not rejected by the test. Thus, for any $\theta \in \Theta(G)$: 
\begin{equation}\label{eq:topConfSetSingle}\small
	\lim_{n \rightarrow \infty} P\left( \theta \in \hat \Theta(\vecY, \alpha) \right) \geq 1 - \alpha.
\end{equation}
It follows that if $G$ has a unique causal ordering (i.e., $\vert \Theta(G) \vert = 1$), then $\hat \Theta(\vecY, \alpha) $ contains that causal ordering with asymptotic probability at least $1-\alpha$. \updatedText{In the oracle setting where each hypothesis test is decided correctly, Alg.~\ref{alg:branchAndBound}---described below---will produce a confidence set that is exactly $\Theta(G)$. However, with sample data, when $\vert \Theta(G) \vert > 1$,
guaranteeing that $\Theta(G) \subseteq \hat \Theta(Y, \alpha)$ with probability at least $1 - \alpha$ is not straightforward, since it would require controlling the family-wise error rate for all $\theta \in \Theta(G)$ and the size of $\Theta(G)$ is not known a priori. Nonetheless, \eqref{eq:topConfSetSingle} implies that $ \Theta(G) \cap \hat \Theta(\vecY, \alpha) \neq \emptyset $ with an asymptotic probability at least $1 - \alpha$. Furthermore, for a fixed data generating procedure, we can expect $\hat \Theta(Y, \alpha)$ to contain $1-\alpha$ of the orderings in $\Theta(G)$; i.e., 
\begin{equation} \label{eq:expectIntersect}
	\lim_{n \rightarrow \infty} \E \left( \frac{\vert \Theta(G) \cap \hat \Theta(Y, \alpha) \vert}{\vert \Theta(G)  \vert} \right) \geq 1-\alpha.
\end{equation}

}

The confidence set $\hat \Theta(\mathbf{Y}, \alpha)$ provides a set of orderings that are not excluded by the data. Different elements of $\hat \Theta(\mathbf{Y}, \alpha)$ suggest different causal orderings which may, but do not have to, lead to different causal conclusions; we elaborate on this point in Section~\ref{subsec:ci_for_causal_effects}. The set $\hat \Theta(\mathbf{Y}, \alpha)$ being large cautions the analyst against overconfidence in a specific estimated ordering.  In contrast, if $\hat \Theta(\mathbf{Y}, \alpha)$ is small, few causal orderings are compatible with the data under the considered model class. This latter aspect is crucial because $\hat \Theta(\mathbf{Y}, \alpha)$ may also be empty, indicating 
that the model class does not capture the data-generating process.

Furthermore, $\hat \Theta(\vecY, \alpha)$ can be post-processed to form other useful objects. Most importantly, similar to the problem studied by \cite{strieder2023equal, strieder2024dual}, we may form confidence intervals for causal effects that also incorporate model uncertainty. In addition, $\hat \Theta(\vecY, \alpha)$ produces a sub/superset of the true ancestral relationships with some user-defined probability. 

Our framework takes a straightforward approach based on goodness-of-fit tests. However, realizing this idea presents significant challenges, and we construct our procedure with careful attention to both statistical and computational aspects. Specifically, our methodology is built using computationally attractive tests for regression models with asymptotic validity \updatedText{in the linear setting when $p\log^{11/2}(n)/n \rightarrow 0$}, where $p$ is the number of variables and $n$ is the sample size.  \updatedText{These tests are then applied across all $p$ variables. Computationally, we do this by devising} the statistical decisions so that we can use a branch-and-bound type procedure to handle problems at a moderate but challenging scale. Despite prioritizing computational tractability, the procedure is asymptotically valid when allowing $p$ to grow with $n$, and we establish the asymptotic validity of the confidence set when \updatedText{$p^2\log^{11/2}(n)/n \rightarrow 0$}.

To motivate $\hat \Theta(\mathbf{Y}, \alpha)$ as an object of interest, we preview the analysis in Section~\ref{sec:numerics-data-analysis} of daily stock returns for 12 industry portfolios. DirectLiNGAM~\citep{shimizu2011directlingam} gives a point estimate of the causal ordering where the Utilities industry is first and causally precedes the other 11 industries. \updatedText{The set $\hat \Theta(\mathbf{Y}, .05)$ contains approximately $1/45,000$ of the $12!$ possible total orderings, and indeed Utilities is first in every ordering in the confidence set. 
Nonetheless, many orderings in $\hat \Theta(\mathbf{Y}, \alpha = .05)$---i.e., those not rejected by the data---have other causal implications which differ from the point estimate. As shown in Section~\ref{sec:numerics-data-analysis}, most orderings in $\hat \Theta(\mathbf{Y}, \alpha = .05)$ are relatively far from the point estimate as opposed to the Fr\'{e}chet Mean. Finally, in the estimated causal ordering, Manufacturing precedes Chemicals, so a naive analysis would conclude that the total effect of Chemicals onto Manufacturing is $0$. In contrast, when accounting for model uncertainty, we produce a 90\% confidence interval for the total effect of Chemicals onto Manufacturing of $\{0\} \cup (.268, .413) \cup (.980, 1.093)$.}

\subsection{Related work}

Previous work on uncertainty in causal discovery predominantly focuses on specific parameters within a causal model, rather than uncertainty across the entire model selection procedure. In linear SEMs with equal error variances, \citet{jankova2018inference} provide confidence intervals for the linear coefficients and \citet{li2019lrt} test the absence of edges; \citet{shi2021testing} consider the same problem for more general additive noise models.  However, this work either assumes that a causal ordering is known or requires accurate estimation of a causal ordering to properly calibrate the test. Thus, they are poorly suited for our setting of interest: where the ``signal strength'' is small, or modeling assumptions may be violated. In contrast, \citet{strieder2023equal, strieder2024dual} focus on the equal variance case with bivariate data and form confidence intervals for causal effects that account for model uncertainty.

A confidence set of models has previously been proposed in work such as \citet{hansen2011model} and \citet{lei2017cross}, who consider a set of candidate models and remove all models determined to be ``strictly worse'' than any other candidate in the set. In contrast, \citet{ferrari2015confidence} and \citet{zheng2019model} form confidence sets by including all models that are not rejected when compared to some saturated model.

As an intermediate step, our framework requires a goodness-of-fit test for regression models. This is a classical problem~\citep{breusch1979heteroscedastic, cook2983diagnostics} that has attracted renewed interest in recent work such as \citet{sen2014testing}, \citet{shah2018goodness}, \citet{berett2019nonparametric}, and \citet{schultheiss2021assessing}. In principle, it is possible to adopt any of these existing procedures into our proposed framework; however, the high computational cost renders them unusable in all but the smallest problems. Thus, we propose a specific new test that possesses both statistical and computational properties that are particularly advantageous for our goal of targeting causal orderings, which requires us to test a very large number of regression models. We provide a detailed comparison of our proposal and the existing work in Section~\ref{sec:singleReg} after describing our procedure.

There is a large literature on testing model fit for a specific SEM, particularly in the linear case. Testing model fit is then classically done by comparing empirical and model-based covariances~\citep{bollen1993testing}. However, in some settings, as discussed in Section~\ref{sec:DAG}, a unique graph may be identified, but simply comparing covariances will fail to falsify graphs in the same Markov equivalence class. Furthermore, the models we consider do not constrain covariances, and thus require alternative approaches.

We note that, by their very nature, Bayesian approaches also quantify uncertainty for causal structures and have seen numerous computational advances, e.g., by focusing on causal orderings \citep{friedman2003being,niinimaki2016structure,kuipers2017partition}. However, nearly all Bayesian causal discovery procedures focus on cases where the graph can only be identified up to an equivalence class. The few exceptions---e.g., \citet{hoyer2009bayesian, shimizu2014bayesian, chang2022order}---require specifying a likelihood for the data, rather than adopting the semi-parametric approach that we employ.  At a more fundamental level, credible intervals differ conceptually from confidence regions that are our focus; especially, since in a complex model selection problem as we consider, there is no Bernstein-von Mises connection between the two concepts.
	
\subsection{Outline}
In Section~\ref{sec:background}, we give a background on causal discovery. In Section~\ref{sec:singleReg}, we propose a computationally attractive goodness-of-fit test for a single causal regression and show in Section~\ref{sec:SEMtest} that it can be used to test a causal ordering and form the confidence set $\hat \Theta(\mathbf{Y}, \alpha)$. We establish theoretical guarantees in Section~\ref{sec:theory} and examine empirical performance in Section~\ref{sec:numerics}.	

\section{Background on causal discovery}
\label{sec:background}

For expository simplicity, we initially focus on linear SEM in~\eqref{eq:SEM} where each $f_v$ is linear. Thus, assuming zero means,
\begin{equation} \label{eq:SEM_2}\small
	Y_{v} = \sum_{u \in \pa(v)} \beta_{v,u}Y_{u} + \varepsilon_{v}, \qquad v\in V.
\end{equation}
Collecting $\varepsilon = (\varepsilon_{v} : v \in V)$ and letting $B \in \mathbb{R}^{p \times p}$ denote the matrix of causal effects where $B_{v,u} = \beta_{v,u}$ if $u \in \pa(v)$ and  $B_{v,u} = 0$ if  $u \not \in \pa(v)$, we have the multivariate model $Y= B Y + \varepsilon$. We use $Y_U = (Y_u : u \in U)$ to denote the sub-vector corresponding to the elements in $U$. We use bold font to denote the collection of $i = 1, \ldots, n$ observations; i.e., $Y_{v,i}$ denotes the $i$th observation of the $v$th variable and $\vecY = (Y_{v,i} : i \in [n], v \in [p])\in \mathbb{R}^{n \times p}$ and $\vecv = (Y_{v,i} : i \in [n]) \in \mathbb{R}^n$. When we pass sets of observations to a function, it should be interpreted as the function applied to each observation; i.e., $h(\vecv) = \left(h(Y_{v,i}) : i \in [n] \right)$.

For linear SEMs, \citet{shohei2006lingam} show that the exact graph can be identified when the errors, $\varepsilon_v$, are mutually independent and non-Gaussian. The identification result relies on the following key observation. Let $\eta_{v \setminus U}$ denote the residuals when $Y_v$ is regressed---using population values---onto a set of variables $Y_U$. If $U$ contains all the parents of $v$ but no descendants (i.e., all variables that have a direct causal effect on $v$ and no variables that are directly or indirectly caused by $v$), then the residuals resulting from population regression are independent of the regressors. Thus, the hypothesis in \eqref{eq:nullHypSingle} implies the hypothesis in \eqref{eq:nullHypSingleIndep} where $\nd(v) = V \setminus \{v \cup \de(v)\}$ denotes the non-descendants of $v$:
\begin{align}
H_0:&\; \pa(v) \subseteq U \subseteq \nd(v) \label{eq:nullHypSingle},\\
H_0&: \eta_{v \setminus U} \independent Y_U. \label{eq:nullHypSingleIndep}
\end{align}
\updatedText{If $U_1 = U\cap \de(v)$ is non-empty, then $\eta_{v \setminus U}$ will contain terms involving $Y_{U_1}$---unless $B$ and $\cov(\varepsilon)$ take specific pathological values. The residuals, $\eta_{v \setminus U}$ will always be uncorrelated with $Y_U$; but if the errors are non-Gaussian then the Darmois-Skitovich Theorem~\citep{darmois1953analyse, skitovich1954linear} implies that $\eta_{v \setminus U} \not \independent Y_U$ (see Theorem 4.3 in~\citet{peters2017elements} for a concise restatement).}
Thus, testing the independence of residuals and regressors may falsify the hypothesis in~\eqref{eq:nullHypSingleIndep} and subsequently~\eqref{eq:nullHypSingle}.

A simple bivariate case is given in Figure~\ref{fig:causalDisc}, where the correct model is $Y_1 \rightarrow Y_2$. When viewing the raw data (left plot), no specific causal relationship is immediately apparent. However, in the middle plot, we have identified the correct model ($Y_1 \rightarrow Y_2$), and the residuals when regressing $Y_2$ onto $Y_1$ are independent of the regressor, $Y_1$. On the right-hand side, we have posited the incorrect model $Y_2 \rightarrow Y_1$. When regressing $Y_1$ onto $Y_2$, the residuals remain uncorrelated with $Y_1$, but are no longer independent of $Y_2$.

\begin{figure}[bt]
\centering
\includegraphics[scale = .95]{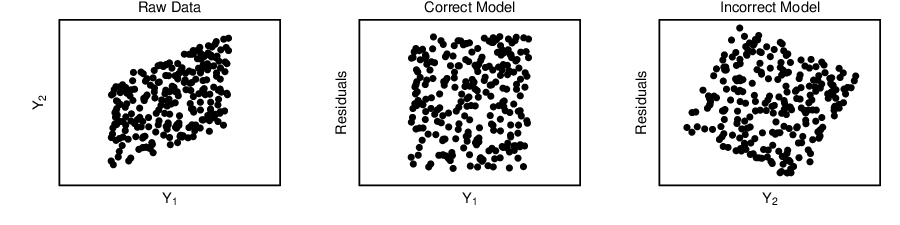}
\caption{\label{fig:causalDisc}Left: Raw data. Middle/Right: Residuals from regressing the posited child onto the posited parent for the correct model (middle) and the incorrect model (right).}
\end{figure}

Of course, we typically do not have access to population values, and the linear coefficients are nuisance parameters that need to be estimated before conducting an independence test. \citet{wang2020high} show---in the linear non-Gaussian SEM setting---consistent recovery of the graph is still possible when estimating the nuisance parameters, even in the high-dimensional setting. Thus, to test~\eqref{eq:nullHypSingleIndep} from data, one might naively use least squares regression and directly test whether the residuals, $\hat \eta_{v \setminus U} = Y_v - \hat \beta Y_{U}$, are independent of $Y_{U}$. Unfortunately, even when the null hypothesis holds, $\hat \eta_{v \setminus U} = \varepsilon_v + (\beta -  \hat \beta) Y_{U}$ so $\hat \eta_{v \setminus U} \not\independent Y_{U}$ and the naive test does not control the Type I error rate. Example~\ref{exm:nuissance} in the appendix provides an illustration. A more careful approach is required for a valid test of \eqref{eq:nullHypSingleIndep}, and this problem has previously been addressed, e.g., by \citet{sen2014testing}. In Section~\ref{sec:singleReg}, we discuss a procedure that is particularly suited for our setting and contrast our approach with existing procedures.

\section{Goodness-of-fit for regression}\label{sec:singleReg}
We now propose a procedure for testing the null hypothesis in~\eqref{eq:nullHypSingleIndep} as a proxy for~\eqref{eq:nullHypSingle}; this test will be used as a building block for testing causal orderings as described in Section~\ref{sec:SEMtest}. We first describe how this can be done in the linear SEM setting and then generalize the procedure when $f_v$ may be non-linear. 

\subsection{Residual bootstrap test for linear models}\label{sec:resBootstrap-linMod}
For some $v \in V$ and set $U \subseteq V\setminus v$, let $b_{v,U} = \arg\min_{ b} \E\left([Y_{v} - b^T Y_{U.1}]^2\right)$ where $Y_{U.1}$ denotes the random vector $Y_U$ augmented by a term for the intercept and let $\eta_{v\setminus U} = Y_{v} - b_{v,U}^TY_{U.1}$; i.e., $b_{v,U} $ is the population regression coefficient and $\eta_{v\setminus U}$ is the resulting residual. The quantity $b_{v,U}$ and random variable $\eta_{v\setminus U}$ are well defined for all $U$ and $v$, and when $\pa(v) \subseteq U \subseteq \nd(v)$ then $b_{v,U}$ coincide with the causal parameters---i.e., $b_{v,U} = \beta_{v,U}$---and $\eta_{v \setminus U} = \varepsilon_{v}$. Given the data, we denote the population residuals as $\vecEta = (\eta_{v\setminus U, 1}, \ldots, \eta_{v\setminus U, n}) $. Furthermore, let $\hat b_{v,U}$ be regression coefficients estimated from sample moments, let $\hat f_v(\vecU) = \vecUint \hat b_{v,U}$, and let $\bm{\hat \eta}_{v \setminus U} = \vecv -  \hat f_v(\vecU) $ denote the residuals calculated using $\hat b_{v,U}$. 

Our test will require a set of functions, $\mathcal{H} = \{h_j(\cdot)\}_{j = 1}^J$ for fixed $J$, which we refer to as \emph{test functions}; these are selected by the analyst and we give practical guidance below. 
\updatedText{We can measure the dependence between $Y_u$ and $\hat  \eta_{v \setminus U}$ using 
$\tau_j(\vecv, u, U; \vecY) = \frac{1}{\sqrt{n}} h_j(\vecu)^T \bm{\hat  \eta}_{v \setminus U}$
which are collected into the vector $ \tau(\vecv, U; \vecY) = \left( \tau_j(\vecv, u, U; \vecY) : u \in U, j \in [J] \right) \in \mathbb{R}^{|U|J}$.
Finally, we aggregate the test statistics into a single measure of dependence 
\begin{equation}\label{eq:tDef}
T(\vecv, U; \vecY) = \vert \tau(\vecv, U; \vecY)  \vert_\infty. 
\end{equation}
When the arguments are clear from context, we will simply use 
$\tau^{(v)} := \tau(\vecv, U; \vecY)$ and 
$T^{(v)} := T(\vecv, U; \vecY)$.
Under the null, 
\begin{equation}\label{eq:oracleStat}\small
\begin{aligned}
\tau_j(\vecv, u, U; \vecY) = \frac{1}{\sqrt{n}}h_j(\vecu)^T \bm{\hat  \eta}_{v \setminus U}  
&= \frac{1}{\sqrt{n}}h_j(\vecu)^T[I - \mathbf{Y}_{U.1}(\mathbf{Y}_{U.1}^T\mathbf{Y}_{U.1})^{-1}\mathbf{Y}_{U.1}^T]\vecepsv,
\end{aligned}
\end{equation}
so that $\E[\tau_j(\vecv, u, U; \vecY)] = 0$ because $\E(\varepsilon_v) = 0$ and $\varepsilon_v \independent Y_{U}$. }

However, when the null hypothesis in~\eqref{eq:nullHypSingle} does not hold, the population regression coefficients are generally not equal to the causal coefficients; i.e., $b_{v,U} \neq \beta_{v,U}$. Letting $U' = \pa(v)  \cup  U$, with a slight abuse of notation, we define $ b_{v,  U'} = (b_{v,U})_u$ if $u \in U$ and $0$ otherwise, and similarly let $\beta_{v,  U'} = (\beta_{v,U})_u$ if $u \in \pa(v)$ and $0$ otherwise. \updatedText{Then, $\eta_{v \setminus U} = \varepsilon_v + (\beta_{v,U'} - b_{v, U'})^T Y_{U'.1}$ and 
$\tau_j(\vecv, u, U; \vecY) = \frac{1}{\sqrt{n}} h_j(\vecY_u)^T \bm{\hat  \eta}_{v \setminus U}$ equals: 
\begin{equation}\footnotesize
\label{eq:altHyp}
\left(\frac{1}{\sqrt{n}}h_j(\vecY_u)^T\bm{\eta}_{v \setminus U} - \sqrt{n}\E( h_j(Y_u)\eta_{v \setminus U})\right) + \frac{1}{\sqrt{n}}h_j(\vecY_u)^T\bm{Y}_{ U'.1}[b_{v,U'} - \hat{ b}_{v,U'}]+ \sqrt{n} \E( h_j(Y_u)\eta_{v \setminus U}).
\end{equation}
The first term is mean $0$, and the second term is asymptotically mean $0$; however, if $\E( h_j(Y_u)^T\eta_{v \setminus U}) \neq 0$ then $\vert \tau_j(\vecv, u, U; \vecY)\vert$---and subsequently $T^{(v)}$---will grow with $n$. By the first order conditions of least squares regression, $\bm{ \hat \eta}_{v \setminus U}$ is always uncorrelated with $\mathbf{Y}_U$. Thus, if $h_j$ is a linear function, $\tau_j(\vecv, u, U; \vecY) = 0$. Furthermore, in multivariate Gaussians, uncorrelated is equivalent to independent, so for linear Gaussian SEMs, $\eta_{v \setminus U} \independent Y_U$ and \eqref{eq:nullHypSingleIndep} always holds regardless of whether \eqref{eq:nullHypSingle} holds. We note, however, that this inability to falsify \eqref{eq:nullHypSingle} is not specific to our approach, but intrinsic to the non-identifiability of linear Gaussian SEMs. In these cases, the Type I error rate of our proposed procedure will be preserved, but---as previously mentioned---the test will have trivial power.

When the errors are non-Gaussian, selecting a non-linear function $h_j$ for which $\E( h_j(Y_u) \eta_{v \setminus U}) \neq 0$ will allow for falsification of the hypothesis in \eqref{eq:nullHypSingle}. In fact, when $p$ is fixed and $n \rightarrow \infty$, Theorem 2 implies that if $\E[h_j(Y_u) \eta_{v \setminus U}] \neq 0$ for at least one test function and $u\in U$, the null will be rejected with probability going to $1$. When $U$ contains descendants of $v$ (i.e., $u \in U \cap \de(v)$), the results of \citet{wang2020causal} imply that using polynomials as test functions will suffice in the following sense: when $h_j(Y_u) = Y_u^K$ for some $K > 1$, then $\E[h_j(Y_u) \eta_{v \setminus U}] \neq 0$ for generic higher-order moments of the errors (i.e., moments $\E(\varepsilon_v^k)$ for $1 < k \leq K + 1$).

For finite samples, the analysis of Theorem 2 also implies that good power may be achieved when $\vert \E[h_j(Y_u) \eta_{v \setminus U}]\vert$ is large relative to the variance of the bootstraped null distribution---discussed below. However, selecting test functions with ``optimal power'' is challenging because both of these quantities depend on the true distribution of $\varepsilon$, which can only be recovered given a true causal ordering. Furthermore, even if this could be optimized for a specific $v$ and $U$, the optimal test functions will change for different $v$ and $U$, making a ``globally optimal'' choice even more difficult. Nonetheless, in Section~\ref{sec:numerics}, we show that with relatively simple test functions, the test exhibits good empirical power when compared to other state-of-the-art tests. Additional simulations with different test functions are also included in the appendix. 
}

If we had access to new realizations of $\varepsilon_v$, we could sample directly from the distribution
of $\tau_j(Y_v, u, U ; Y)$ and ultimately $T^{(v)}$---conditional on $\vecU$---by replacing
$\vecepsv$ in~\eqref{eq:oracleStat} with new draws.
Comparing the observed $T^{(v)}$ to the distribution of these new realizations would yield an exact finite-sample test. We refer to this distribution as the \emph{oracle distribution} because, in practice, we cannot resample $\bm{\varepsilon}_v$ exactly.
Alternatively, conditioning on $\mathbf{Y}_U$, the quantity in~\eqref{eq:oracleStat} is asymptotically normal under the null hypothesis. Thus, the null distribution could also be approximated by samples which replace the $\bm{\varepsilon}_v$ in~\eqref{eq:oracleStat} with draws from a Gaussian. However, when $\varepsilon_v$ is not close to a Gaussian, we see drastic improvements by using the residual bootstrap procedure proposed below. We illustrate this explicitly with a simulation study in Section~\ref{sec:appendix-asymptotic-calibrate} of the appendix.

Instead, we calibrate our test with a residual bootstrap procedure. For each bootstrap draw, we condition on $\bm{Y}_U$ and replace $\bm{\varepsilon}_v$ in~\eqref{eq:oracleStat} with 
$\bm{\tilde \eta} = (\tilde \eta_i : i \in [n])$ where each $\tilde \eta_i$ is drawn i.i.d.~from the empirical distribution of $\hatEta$\updatedText{, denoted $\hat F_n$}.
This is equivalent to forming $\tildevecv = \hat f(\vecU) + \mathbf{\tilde \eta}$, regressing $\tildevecv$ onto $\vecUint$ to form the residuals $\bm{\hat{\tilde  \eta}}_{v \setminus U}$ and computing $\tau_j(\tildevecv, u, U; \vecY) = (1/ \sqrt{n}) h_j(\vecu)^T \bm{\hat{\tilde  \eta}}_{v \setminus U}$ as described in Alg.~\ref{alg:fullNGProcedure}. \updatedText{Similar to before, we then compute $\tilde T^{(v)}:=T(\tildevecv, U; \vecY) = \max_{u,j}\vert \tau_j(\tildevecv, u, U; \vecY)\vert $.} 
In the simulations, we use the asymptotically equivalent quantity which divides by $\sqrt{n - |U|}$ instead of $\sqrt{n}$. \updatedText{In Section~\ref{sec:theory}, we show that this approximation converges to the oracle distribution when $p\log^{11/2}(n) = o(n)$. }

	\begin{algorithm}[tb]
	\caption{\texttt{testAn}$(v, U,  \vecY)$ \label{alg:fullNGProcedure}}
	\begin{algorithmic}[1]
		\State{Calculate $T^{(v)} = T(Y_v, U; \vecU)$ by regressing $\vecv$ onto $\vecUint$}
		\For{$l =1, \ldots, L$}
		\State{Sample $\tilde \eta_i \stackrel{\text{i.i.d.}}{\sim} \hat F_n$ to form $\tildevecv = \hat f_v(\vecUint) +  \mathbf{\tilde \eta}$ and calculate $\tilde T^{(v)}_l = T(\tildevecv, U; \vecY)$}
		\EndFor
		\State{\textbf{Return:} $\left(1 + \sum_{l} \mathbbm{1}\{T^{(v)} < \tilde T^{(v)}_l\}\right)/(L+1)$} 
	\end{algorithmic}
\end{algorithm}

Various other procedures, discussed below, have also been proposed for testing goodness-of-fit for a linear model via the hypothesis in~\eqref{eq:nullHypSingleIndep}. In theory, these procedures could also be used in the framework we subsequently propose to test causal ordering. However, practically speaking and as shown in Section~\ref{sec:simulations}, the computational cost of these procedures---save perhaps \citet{schultheiss2021assessing}---renders them infeasible for our goal of computing confidence sets for causal orderings. 

Moreover, beyond the drastic computational benefits, the proposed procedure possesses some statistical advantages, which we briefly discuss now and also empirically demonstrate in Section~\ref{sec:simulations}. \citet{sen2014testing} use the Hilbert-Schmidt Independence Criterion (HSIC) \citep{gretton2007hsic} to measure dependence between regressors and residuals. They only consider the fixed $p$ setting, and the simulations in Section~\ref{sec:numerics} show that the type I error is inflated when $p$ is moderately sized compared to $n$. We conjecture this is partly because they bootstrap both the regressors and residuals to approximate the joint distribution rather than conditioning on the regressors and only approximating the distribution of the errors. In contrast, \citet{shah2018goodness} propose Residual Prediction (RP) and \citet{berett2019nonparametric} propose MintRegression, which test the goodness-of-fit conditional on the covariates; however, both procedures calibrate their tests using a parametric bootstrap that assumes the errors are Gaussian. In the supplement, \citet{shah2018goodness} do consider cases where the errors are non-Gaussian and use a residual bootstrap, which shows good empirical performance although they do not provide any theoretical guarantees. Finally, \citet{schultheiss2021assessing} also propose a goodness-of-fit test for individual covariates in a linear model using a statistic similar to~\citet{wang2020high}. They show that the statistic is asymptotically normal and calibrate the hypothesis test using an estimate of the limiting distribution. A direct comparison of required conditions is not straightforward because they focus on the high-dimensional sparse linear model; whereas we do not assume sparsity, but require $p < n$. Nonetheless, for valid testing, they require the number of non-zero coefficients to be $o(n^{1/2}/\log^3(p))$; in contrast, \updatedText{we require $p\log^{11/2}(n) = o(n)$}.

\subsection{Non-linear models via sieves}\label{sec:nonLinear}
Using a similar argument with the independence of residuals and regressors, \citet{peters2014additive} show that the causal graph may also be identified when the structural equations in~\eqref{eq:SEM}, $f_v$, are non-linear. The procedure proposed above is directly generalized to this setting.

Let $\Phi^{(v)} = \{\phi_k\}_{k \geq 1}$ be a basis of functions (e.g., the b-spline or polynomial basis) that take inputs $Y_{U}$. We approximate $f_v$ with $\hat f_v$ by regressing $Y_v$ onto the span of $\Phi^{(v)}$. \updatedText{We require the constant function to be in the span of $\Phi^{(v)}$ and assume $\vert \Phi^{(v)}\vert = \vert U\vert K + 1$.}  Given the residuals of $Y_v$, the test statistics can then be calculated in the same way as in the linear setting. In this case, to ensure the test statistics are not identically zero, we must select test functions that do not lie in the span of $\Phi^{(v)}$. If we do not know a good basis, for $f_v$ a priori, we can use a sieve estimator where $K$ grows with $n$.
In Section~\ref{sec:theory}, we show that even if $f_v \not \in \text{span}(\Phi^{(v)}_{K})$ for any finite ${K}$, the proposed test is valid as long as the \updatedText{approximation bias} decreases appropriately with ${K}$.

\section{Inference for causal orderings}\label{sec:SEMtest}

Before discussing details, we first discuss some of the trade-offs involved in our design decisions. \updatedText{Estimating a causal ordering can be seen as a preliminary task when estimating a DAG, and various causal discovery procedures have fruitfully employed a search over causal orderings instead of individual graphs; e.g., \citet{raskutti2018learning, solus2021permutation}. Thus, for computational reasons, we focus on confidence sets in a much smaller space of causal orderings rather than all possible DAGs. Given a correct causal ordering, the estimation of the graph simplifies to variable selection in a sequence of regressions~\citep{shojaie2010penalized}, and the conditions required for finite sample DAG recovery are weaker than the conditions typically required when a causal ordering is not known in advance. We emphasize, however, that consistent DAG estimation---even when an ordering is known---requires an assumption on the minimum strength of each edge, which we do not require for the confidence set of causal orderings. Such an assumption would preclude many settings with low ``signal strength'' where the graph cannot be recovered with high probability. In fact, these are exactly the settings where confidence statements are most needed. However, if the true graph could be recovered with probability at least $1-\alpha$ when given a correct ordering, then \eqref{eq:topConfSetSingle} implies that the set of graphs estimated using each $\theta \in \hat \Theta(\vecY, \alpha)$ is an asymptotically valid $1 - 2 \alpha$ confidence set for $G$.}

We also choose to form the confidence set $\hat \Theta(\mathbf{Y}, \alpha)$ by inverting a goodness-of-fit test. If the model assumptions are violated, then all possible orderings may be rejected, resulting in an empty confidence set. Alternatively, one could form a confidence set by considering a neighborhood around $\hat \theta$, a point estimate of a causal ordering. This would produce a non-empty confidence set even when the model is misspecified and might be preferred if $\hat \theta$ represents a useful ``projection'' into the considered class of models. However, in practice, it is difficult to know if the misspecification is ``mild,'' and we argue that observing an empty confidence set is important because it alerts the scientist to potentially choose a causal discovery procedure that makes less restrictive identifying assumptions.


Finally, we construct a test for each causal ordering by aggregating several regression tests. \updatedText{Although we only require $p\log^{11/2}(n)/n\rightarrow 0$ for asymptotic validity of all individual regression tests, the aggregation requires $p^2\log^{11/2}(n)/n\rightarrow 0$. Alternatively, a direct test that does not first consider individual tests might enjoy better statistical properties.} However, we trade statistical efficiency for computational efficiency and aggregating individual tests allows for a branch-and-bound type procedure, which in practice drastically decreases computation and enables feasible analysis of ``medium-sized'' problems. In Section~\ref{sec:numerics-ConfSets}, we show that this procedure can be applied to problems with $p \approx 20$.

\subsection{Testing a given ordering}\label{sec:testingOrd}

For an ordering $\theta \in\mathcal{S}_V$, let $\pr(v) = \{u : \theta(u) < \theta(v)\}$ be the set of nodes that precede $v$ in $\theta$. When $\theta$ is a valid ordering for $G$, then $\pa(v) \subseteq \pr(v) \subseteq \nd(v)$ for all $v$ such that $\theta(v) > 1$. However, when $\theta$ is not a valid causal ordering, there exists some $v$ such that $\pr(v) \not \subseteq \nd(v)$. Thus, testing whether $\theta$ is a valid causal ordering is equivalent to testing 
\begin{equation}\label{eq:semH0}\small
H_{0, \theta}: \pa(v) \subseteq \pr(v) \subseteq \nd(v) \quad \forall v \text{ such that } \theta(v) > 1.
\end{equation}
To operationalize a test for $H_{0, \theta}$, we use the procedure from Section~\ref{sec:resBootstrap-linMod} to test $
H_{0, \theta, v}: \eta_{v \setminus \pr(v)} \independent Y_{\pr(v)}$ for all $v$ such that  $\theta(v) > 1$. Using the residual bootstrap procedure to test $H_{0, \theta, v}$ produces a p-value, denoted as $\hat \gamma_{\theta, v}$, which approximates, $\gamma_{\theta, v}$, the p-value which would result from a test calibrated by the oracle distribution. We propose aggregating the $p-1$ p-values into a single test for~\eqref{eq:semH0} by taking the minimum p-value. 
Specifically, let
\begin{equation}\label{eq:tippet}\small
\hat \gamma_{\theta} = \min_{v : \theta(v) > 1} \hat \gamma_{\theta,v} \qquad  \text{ and } \qquad \gamma_{\theta} = \min_{v : \theta(v) > 1}  \gamma_{\theta,v}.
\end{equation}

By Lemma~\ref{lem:indepPvals}, under the null, the of p-values produced by the oracle procedure, $(\gamma_{\theta, v} \; : \; \theta(v) > 2)$, are mutually independent so $\gamma_{\theta}$ follows a $\text{Beta}(1, p-1)$. \updatedText{This allows us to compute a final p-value for the entire ordering which we denote as $\Gamma_{\theta}$. Of course, we do not have access to $\gamma_{\theta}$, but under conditions described in Section~\ref{sec:theory},  $\hat \gamma_{\theta} \rightarrow_p \gamma_{\theta}$ and comparing $\hat \gamma_{\theta}$ to a $\text{Beta}(1, p-1)$ yields an attainable final p-value for $H_{0, \theta}$, denoted as $\hat \Gamma_\theta$.
\begin{lemma}\label{lem:indepPvals}
If $\theta \in \Theta(G)$, then the p-values $(\gamma_{\theta, v} \; : \; \theta(v) > 2)$ calculated for each level using the oracle procedure are mutually independent.  \end{lemma}
}

\subsection{Efficient computation}
To satisfy~\eqref{eq:topConfSetSingle}, we construct $\hat \Theta(\vecY; \alpha)$ by including any $\theta$ where $H_{0, \theta}$ is not rejected by a level $\alpha$ test; i.e.,  
\begin{equation}\small
    \hat \Theta(\vecY; \alpha) = \{ \theta : \hat \Gamma_\theta \geq \alpha\}.
\end{equation}
Of course, enumerating all permutations is computationally prohibitive, so we propose a branch-and-bound style procedure to avoid unnecessary computation. The pseudocode is given in Alg.~\ref{alg:branchAndBound}. For any fixed $\theta$, we sequentially test $H_{0, \theta, v}$ for $z = \theta^{-1}(v) = 2,\ldots, p$, and we update a running record $
\hat \gamma_{\theta}^{(z)} =  \min_{v : \theta^{-1}(v) \leq z} \hat{\gamma}_{\theta,v}$.
Once $\hat \gamma_{\theta}^{(z)}$ is less than the $\alpha$ quantile of a $\text{Beta}(1, p-1)$, we can reject $\theta$ without testing the remainder of the ordering. 

Furthermore, we test each ordering in $\mathcal{S}_V$ simultaneously rather than sequentially. We first testing all orderings of length $z=2$ \updatedText{and store all unrejected orderings in $\hat \Theta_2$}. Subsequently, we only consider orderings of length $z=3$ that are formed by appending a node to \updatedText{an ordering in $\hat \Theta_{z-1}$. The unrejected orderings are placed in $\hat \Theta_{z}$ and we repeat this process for increasing $z$.} This approach avoids redundant computation because the test of $\pa(v) \subseteq \pr(v) \subseteq \nd(v)$, only depends on the combination of elements included in $\pr(v)$ and not the specific ordering $\theta$. For example, when $z = 3$, once we have tested $\pa(6) \subseteq (4,5) \subseteq \nd(6)$ for the incomplete ordering $(4, 5, 6)$ we do not need to recompute the test for $(5, 4, 6)$. In the worst case, when the signal is small and no orderings are rejected, the procedure is still an exhaustive search. However, in Section~\ref{sec:simulations}, we show that under reasonable signal-to-noise regimes, problems with $p = 20$ are feasible.

\begin{algorithm}[tb]
		\caption{$\texttt{branchAndBound}(\vecY, \alpha)$\label{alg:branchAndBound}}
		\begin{algorithmic}[1]
			\State{Initialize $\hat \Theta_1 = \{(1), \ldots, (p)\}$, $z = 2$, and $\gamma_\theta^{(1)} = 1$ for each \updatedText{$\theta \in \hat \Theta_1$}}
			
            \While{$z \leq p$ \text{ and } $\hat \Theta_{z-1} \neq \emptyset$}
            \State{Let $\Psi_z$ be the set of unique unordered sets derived from permutations in $\hat \Theta_{z-1}$ \updatedText{and $\hat \Theta_z = \emptyset$}}
                \For{$\psi \in \Psi_z$ and each $v \in V \setminus \psi$}
                   \If{$\texttt{testAn}(v, \psi, \vecY)\geq$ the $\alpha$ quantile of $\text{Beta}(1,p-1)$ }
                        \For{\updatedText{$\theta \in \hat \Theta_{z-1}$} which corresponds to $\psi$}
                        \State{Add $\theta' = (\theta, v)$ to $\hat \Theta_z$ and set $\gamma^{(z)}_{\theta'} = \min(\gamma^{(z-1)}_{\theta'}, \texttt{testAn}(v, \psi, \vecY))$ }  
                              \EndFor
                            \EndIf
                \EndFor
                \State{$z= z+1$}    
			\EndWhile
			\State{\textbf{Return}: $\hat \Theta(\mathbf{Y}, \alpha) =  \hat \Theta_{p}$ }
		\end{algorithmic}
\end{algorithm}

\subsection{Post-processing the confidence set}
\label{subsec:ci_for_causal_effects}
We now discuss how $\hat \Theta(\mathbf{Y}, \alpha)$ can be post-processed into other useful objects. Specifically, we consider: (1) confidence intervals for causal effects which incorporate model uncertainty, and (2) sub/super-sets of ancestral relations with confidence.

\citet{strieder2023equal} consider a linear SEM with equal variances and provide CIs for causal effects which account for the model uncertainty.
With a similar goal, we propose a procedure for the setting of linear SEMs with independent errors. 
We focus on the total effect of $v$ onto $u$, \updatedText{denoted} $\partial\E[Y_u \mid \text{do}(Y_v = y)]/\partial y$ using the do-operator~\citep{pearl2009causality}; a procedure for the direct effect of $v$ onto $u$, i.e., \updatedText{the coefficient} $\beta_{u,v}$, is analogous and discussed in Section~\ref{sec:ci-direct} of the appendix.
When $v \not \in \an(u)$, the total effect of $v$ on $u$ is $0$, and when $v \in \an(u)$, the total effect may be recovered by a regression of $Y_u$ onto $Y_v$ and a set of additional covariates---often called the \emph{adjustment set}. In particular, letting $\an(v)$ be the adjustment set yields an unbiased estimate. While adjustment sets which recover the total effect may not be unique, an incorrect adjustment set may bias the estimate; e.g., incorrectly including a descendant of $Y_u$ or excluding a parent of $Y_v$ from the adjustment set may induce bias. Thus, naively selecting a single adjustment set and calculating a confidence interval for the parameter of interest will not provide nominal coverage when there is considerable uncertainty in a ``correct'' adjustment set. Robust quantification of uncertainty must also account for uncertainty in the selected adjustment set.

Alg.~\ref{alg:TotalEffect} describes a procedure to calculate $1 - \alpha$ CIs for the total effect of $v$ onto $u$ which account for model uncertainty. Specifically, we consider the adjustment set $\pr(v)$ for each ordering $\theta \in \hat \Theta(\mathbf{Y}, \alpha / 2)$. We then calculate the $1 - \alpha / 2$ CI for the regression parameter of interest, conditional on that adjustment set. The final CI is given by the union of all conditional CIs. In practice, if $v$ and $u$ are fixed in advance, this can be calculated simultaneously with $\hat \Theta(\mathbf{Y}, \alpha/2)$ to avoid redundant regressions. This is similar in flavor to the  IDA procedure of \citet{maathuis2009estimating}, but we additionally account for uncertainty due to estimating the graph rather than just population level non-identifiability within a Markov equivalence class.

\begin{lemma}\label{lem:ci-modelUncertain}
Let $\pi_{u,v}$ denote the total causal effect of $v$ onto $u$. Suppose $\hat \Theta(\mathbf{Y}, \alpha/2)$ satisfies~\eqref{eq:topConfSetSingle}, and $C(S)$ is an asymptotically valid $1-\alpha/2$ confidence interval for the parameter of interest, conditional on $S$ being a valid adjustment set. Then, for the confidence interval produced by Alg.~\ref{alg:TotalEffect},
$
\lim_{n \rightarrow \infty} P(\pi_{u,v} \in \hat C_\alpha) \geq 1 - \alpha.
$
\end{lemma}

\begin{algorithm}[tb]
		\caption{\label{alg:TotalEffect}Get $1-\alpha$ CI for total effect of $v$ onto $u$}
		\begin{algorithmic}[1]
           \For{$S \in \mathcal{S} = \{S \, : \, S = \pr(v) \text{ for some } \theta \in \hat \Theta(\mathbf{Y},\alpha / 2) \text{ such that } \theta(v) < \theta(u) \}$}
                        \State{Form $C(S)$, the $1-\alpha/2$ CI for the coefficient of $Y_v$ when regressing $Y_u$ onto $Y_{S \cup \{v\}}$}
                \EndFor
			\If{$\theta(u) > \theta(v)$ for any $\theta \in \hat \Theta(\mathbf{Y},\alpha / 2)$}
			    \State{\textbf{Return}: $\hat C_\alpha = \{0 \} \cup \{\bigcup_{S \in \mathcal{S}} C(S)\}$}
			 \Else
			    \State{\textbf{Return}: $\hat C_\alpha = \bigcup_{S \in \mathcal{S}} C(S)$}
			 \EndIf
		\end{algorithmic}
\end{algorithm}

Furthermore, $\hat \Theta(\mathbf{Y}, \alpha)$ may be used to compute a sub/super-set of ancestral relations. 
Let $\mathcal{A}(G) = \left\{(u, v) \, : \, u \in \an(v) \right\}$ denote the set of true ancestral relationships in $G$,
$\mathcal{\hat A}_\cap = \{ (u, v) \, : \, \theta(u) < \theta(v) \; \forall \theta \in \hat \Theta(\vecY, \alpha) \}$  denote the set of ancestral relations that hold for all $\theta \in \hat \Theta(\mathbf{Y}, \alpha)$, and $\mathcal{\hat A}_\cup = \{ (u, v) \, : \, \exists \theta \in \hat \Theta(\vecY, \alpha) \text{ s.t. } \theta(u) < \theta(v)\}$ denote the set of ancestral relations which are implied by at least one $\theta \in \hat \Theta(\mathbf{Y}, \alpha)$.
Lemma~\ref{lem:ancestralSet} shows $\mathcal{\hat A}_\cap 
\subseteq \mathcal{A}\subseteq \mathcal{\hat A}_\cup$ with probability at least $1 - 2\alpha$.
The set $\mathcal{\hat A}_\cap$ is similar to the conservative set of causal predictors given in \citet{peters2016invariant}.

\begin{lemma}\label{lem:ancestralSet}
Suppose $\hat \Theta(\mathbf{Y}, \alpha)$ satisfies~\eqref{eq:topConfSetSingle}. Then,
$
\lim_{n \rightarrow \infty} P(\mathcal{\hat A}_\cap 
\subseteq \mathcal{A}\subseteq \mathcal{\hat A}_\cup) \geq 1 - 2\alpha. 
$
\end{lemma}

\section{Theoretical guarantees}\label{sec:theory}
\updatedText{We present conditions for the asymptotic validity of the residual bootstrap and analyze the power of the regression goodness-of-fit test. Let $\lambda_{\min}(C)$ and $\lambda_{\max}(C)$ represent the minimum and maximum eigenvalues of matrix $C$.

\subsection{Asymptotic validity}

Fix $\theta \in \Theta(G)$ and consider the hypothesis $H_{0,\theta,v}: \pa(v) \subseteq \pr(v) \subseteq \nd(v)$. Recall that $Y_v = f_v(Y_{\pa(v)}) + \varepsilon_v$, and we model $f_v$ using a function basis $\Phi^{(v)}$ which includes an intercept term where $\vert \Phi^{(v)} \vert = \vert\pr(v)\vert K + 1$ for some $K > 0$. Let $\B_v \in \mathbb{R}^{\vert \pr(v)\vert K + 1}$ denote the design matrix onto which we regress $\vecv$ and let $Z_{v,i} = (Z_{v,i,k} : k \in [|\pr(v)|K + 1])$ denote the $i$th row (i.e., observation) of $\B_v$ where $Z_{v,i,k} = \phi^{(v)}_{k}(Y_{\pr(v),i})$.

We first assume that $\varepsilon_{v,i}$ is sub-exponential and has a variance which is upper and lower bounded. The assumption that $\underline{\sigma}^2 \leq 1 \leq \overline{\sigma}^2 $ is only required for simplification.  Assumption~\ref{asm:designCondition} places weak requirements on the approximating basis and is satisfied if $Z_{v,i}$ is jointly sub-exponential.  

\begin{assumption}\label{asm:subExpErrors}
Assume $\max_v\Vert \varepsilon_{v,i}\Vert_{\psi_1}   \leq M < \infty$. Furthermore, suppose that
\[0 < \underline{\sigma}^2 \leq \min_v \E(\varepsilon_{v,i}^2) \leq  1 \leq \max_v \E(\varepsilon_{v,i}^2) \leq  \overline{\sigma}^2 < \infty \; \text{ and } \; \max_v\E([\varepsilon_{v,i}/\sigma_v]^4) \leq M^4.\]
\end{assumption}

\begin{assumption}\label{asm:designCondition}
Suppose that $\max_{v,k} \Vert Z_{v,i,k} \Vert_{\Psi_1} < M$. Furthermore, suppose $0 < \lambda_{\min,Z} \leq \lambda_{\min}(\E(Z_{v,i}Z_{v,i}^T))$ and for some $1 < \omega <\infty$ we have $\max_{v} \sqrt{\E\left[(s^TZ_{v,i})^4\right]} \leq \omega s^T\E(Z_{v,i}Z_{v,i}^T)s$ for all $s \in \mathbb{R}^{\vert\pr(v)\vert K + 1}$. 
\end{assumption} 

Let $b_v$ denote the coefficients of the population least squares projection of $f_v$ onto $\Phi^{(v)}$; i.e., $b_v = \min_b\E_Y[(Y_{v,i} - Z_{v,i}^T b)^2]$. Furthermore, let $ d_{v,i} =  f_v(Y_{\pr(v),i}) - Z_{v,i}^T b_v$ and $\bias_v = (d_{v,i}: i \in [n])$ denote the approximation bias. In the well-specified linear SEM case, $K = 1$, $\B_v$ is $\vecY_{\pr(v)}$ with an additional intercept term, and $\bias_v = 0$ for all $v$. When a basis for $f_v$ is not known a priori, we consider a sieve where $K$ grows with $n$ and $\bias_v \neq 0$. 

Assumption~\ref{asm:bias} requires that $d_{v,i}$ is sub-exponential and has an upper bound on the first and second moments of the absolute value. Note that $\E(d_{v,i}) = 0$ since $\Phi^{(v)}$ includes an intercept. 
In Corollary~\ref{cor:finalValidNonLinear}, we will make explicit assumptions on how these quantities decay with respect to $K$. 

\begin{assumption}\label{asm:bias}
Suppose that $\max_v\Vert d_{v,i}\Vert_{\psi_1}  = M_d \leq M$, $\max_v\E(d_{v,i}^2) = d^\star_2$, and $\max_v\E(\vert d_{v,i}\vert ) = d^\star_1$. 
\end{assumption}

Let $\Hv \in \mathbb{R}^{n \times \vert \pr(v)\vert J}$ be a matrix where each row is the test functions evaluated on $Y_{\pr(v),i}$; i.e., the $i$th row is $(h_j(Y_{u,i}) \,:\,  j \in [J], u \in \pr(v))$. The test statistic described in Section~\ref{sec:resBootstrap-linMod} is $  T^{(v)} := T(\vecv, \pr(v); \vecY) = \vert \tau^{(v)} \vert_\infty$ where
\begin{equation}
   \tau^{(v)} = \frac{1}{\sqrt{n}} \Hv^T(I - \B_v(\B_v^T\B_v)^{-1}\B_v^T)(\vecepsv + \bias_v) 
 = \frac{1}{\sqrt{n}}\sum_i \zeta^{(v)}_i (\varepsilon_{v,i} + d_{v,i}) \in \mathbb{R}^{\vert \pr(v)\vert J}
\end{equation}
and $\zeta^{(v)}_i \in \mathbb{R}^{\vert \pr(v)\vert J}$ is the $i$th column of $(\Hv^T(I - \B_v(\B_v^T\B_v)^{-1}\B_v^T))$. The vector $\zeta^{(v)}_i$ may be computed by regressing each of the $\vert\pr(v)\vert J$ test functions onto $\B_v$, and then collecting each of the residuals from observation $i$ into a single vector. Furthermore, we decompose $\tau^{(v)} = \kappa^{(v)} + \nu^{(v)}$ where $\nu^{(v)}$ is due to the approximation bias
\begin{equation}
\begin{aligned}
       \kappa^{(v)}  = \frac{1}{\sqrt{n}}\sum_i \zeta^{(v)}_i \varepsilon_{v,i} \quad \text{ and } \quad 
       \nu^{(v)} = \frac{1}{\sqrt{n}}\sum_i \zeta^{(v)}_i d_{v,i}.
\end{aligned}
\end{equation}
Because each element of $\vecepsv$ is mean $0$ and independent of all $\vecY_{\pr(v)}$, we have that $ \E(\kappa^{(v)} \mid \vecY_{\pr(v)} ) = 0$  and $\var(\kappa^{(v)} \mid \vecY_{\pr(v} ) = \sigma_v^2\Sigma^{(v)}$  where $\sigma^2_v = \var(\varepsilon_{v,i})$ and $\Sigma^{(v)} = \frac{1}{n}\Hv^T(I - \B_v(\B_v^T\B_v)^{-1}\B_v^T)H$ is the covariance of $\zeta^{(v)}_i$. In the oracle distribution, we condition on $\vecY_{\pr(v)}$ (thus fixing also $\Hv$, $\bias_v$, and $\nu^{(v)}$) and resample $\varepsilon_{v,i}$ from the population distribution. Comparing the observed $T^{(v)}$ to this distribution yields an exact test in finite samples. In practice, we compare the observed $T^{(v)}$ to a bootstrap distribution defined by
\begin{equation}
\begin{aligned}
  \tilde \tau^{(v)} &= \frac{1}{\sqrt{n}}\sum_i \zeta^{(v)}_i \tilde \varepsilon_{v,i} \in \mathbb{R}^{\vert\pr(v)\vert J}\quad \text{ and } \quad
  \tilde T^{(v)} & := T(\tildevecv, \pr(v); \vecY) = \vert \tilde \tau \vert_\infty,
  \end{aligned}
\end{equation}
where $\tilde \varepsilon_{v,i} $ is drawn i.i.d. from the empirical distribution of $\hatEta$. Note that $\hatEta$---and thus $\tilde T^{(v)}$---depends on $\vecY_{\pr(v)}$ and $\vecepsv$. 

Assumption~\ref{asm:testFuncWellCondition} restricts the test functions $\mathcal{H}$. In particular, we require that the correlation matrix of $\zeta_{i}^{(v)}$ is not too poorly conditioned. Furthermore, we will assume that the maximum value of any of these residuals is upper bounded. While this may initially seem opaque, we emphasize that each individual regression test is conditional on $Y_{\pr(v)}$, and Assumption~\ref{asm:testFuncWellCondition} can be empirically verified using only $Y_{\pr(v)}$ (without $Y_v$). Thus, we can always use test functions that are explicitly constructed---without looking at $Y_v$ so the test remains valid---to satisfy Assumption 4.

\begin{assumption}\label{asm:testFuncWellCondition}
Let $C^{(v)}$ denote the correlation matrix corresponding to $\Sigma^{(v)}$. Suppose that
\begin{equation*}
0 < \lambda_{\min,C} \leq \min_v \lambda_{\min}(C^{(v)}) \leq \max_v \lambda_{\max}(C^{(v)}) < \lambda_{\max,C} < \infty.
\end{equation*}
Furthermore, suppose that
\begin{equation*}
    h_{\max,1} = \max_{v,i,j} \left \vert \zeta_{i,j}^{(v)} / \sqrt{\Sigma^{(v)}_{jj}} \right \vert, \; h_{\max,2} = \max_{v,i,j} \left \vert \zeta_{i,j}^{(v)}  \right \vert = h_{\max, 2}, \; \text { and } \; h_{\max,3} = \max_{v,i,j} \left \vert h_{j}(Y_{v,i})  \right \vert.
\end{equation*}
\end{assumption}
 
Theorem~\ref{thm:highDCLT} shows that $\max_v \vert \gamma_{\theta,v} - \hat \gamma_{\theta, v}\vert$---the difference for each $v$ between p-values using the oracle and bootstrap distribution---can be upper bounded with overwhelming probability. The main idea for the proof is to use the high-dimensional CLT results of~\citet{chernozhukov2023nearly} to show that $\tau^{(v)}$ and $\tilde \tau^{(v)}$ can be well approximated over rectangles by the same multivariate normal. Lemma~\ref{lem:totalPValClose} then implies that the difference between the p-values for testing the entire ordering produced by the oracle and bootstrap procedure, $\Gamma_{\theta}$ and $\hat \Gamma_{\theta}$ respectively, is upper bounded by $p\times  \max_v \vert \gamma_{\theta,v} - \hat \gamma_{\theta, v}\vert$. The extra term comes from the fact that the quantiles of the null distribution of $\min_v \gamma_{\theta,v}$ depend on $p$. It could be avoided if we instead used a test statistic which aggregated $\tau^{(v)}$ across all levels of the ordering; e.g., $T = \max_v \vert \tau^{(v)} \vert$. However, this would not allow for the computationally efficient branch and bound procedure and would require an exhaustive search over all permutations.

\begin{theorem}\label{thm:highDCLT}
Suppose Assumptions~\ref{asm:subExpErrors}, \ref{asm:designCondition}, \ref{asm:bias}, and \ref{asm:testFuncWellCondition} hold, and $pK/n \rightarrow 0$. 

If $\bias_v = 0$ for all $v$ and for some universal constant $C$ we have $C(\overline{\sigma}^2 Kp/n + \log(n)/\sqrt{n}) < \underline{\sigma}^2/2$, then with probability $1-o(1)$ we have
\begin{equation}\label{eq:highDCLT1}
\begin{aligned}
\max_v \vert \gamma_{\theta, v} - \hat \gamma_{\theta,v} \vert &\lesssim     \left(n^{-1/2} + \frac{Kp}{n}\right)\frac{h_{\max,1}^4M^4 \log^{11/2}(n) \overline{\sigma}^2}{\lambda_{\min,Z}\lambda_{\min,C}\underline{\sigma}^2}\left(1 + \log\left[\frac{\overline{\sigma}^2
}{ \underline{\sigma}^2 \lambda_{\min,C}}\right]\right).
\end{aligned}
\end{equation} 

If $\bias_v \neq 0$, $\max_v \vert \nu^{(v)}\vert_\infty < \delta_1$, and for some universal constant $C$, we have $C((\overline{\sigma}^2  +\lambda_{\min,Z}^{-1})Kp\log^2(n)/n + \log(n)/\sqrt{n} + d^\star_2) < \underline{\sigma}^2/2$,
then with probability $1-o(1)$ we have
\begin{equation}\footnotesize
\begin{aligned}
\max_v \vert \gamma_{\theta, v} - \hat \gamma_{\theta,v} \vert &\lesssim     
    \delta_1\sqrt{\log(pJ)}  + d^\star_2\left( \frac{\log^2(n)(\overline{\sigma}^2 + \lambda_{\min, Z}^{-1})}{\lambda_{\min,C} \underline{\sigma}^2}\right)\left(1 + \log\left[\frac{\overline{\sigma}^2 + \lambda_{\min, Z}^{-1}}{\lambda_{\min,C} \underline{\sigma}^2}\right] \right)\\
    &\quad + 
    \left(n^{-1/2} + \frac{Kp}{n}\right) \left(\frac{ h_{\max,1}^4 M^4 \log^{11/2}(n) (\overline{\sigma}^2 + \lambda_{\min, Z}^{-1})}{\lambda_{\min,C} \underline{\sigma}^2}\right)\left(1 + \log\left[\frac{\overline{\sigma}^2 + \lambda_{\min, Z}^{-1}}{\lambda_{\min,C} \underline{\sigma}^2}\right] \right). 
\end{aligned}
\end{equation}

\end{theorem}

\begin{lemma}\label{lem:totalPValClose}
Suppose $\theta \in \Theta(G)$. Then
$\vert \Gamma_\theta - \hat \Gamma_\theta \vert \leq p \max_v\vert \gamma_{\theta,v}- \hat \gamma_{\theta,v}\vert$.
\end{lemma}

Finally, the following corollaries give sufficient conditions for $\hat \Gamma_{\theta}$ to converge in probability to $\Gamma_{\theta}$ so that the bootstrap test is asymptotically valid. When the data are generated by a linear SEM so that $K=1$, Corollary~\ref{cor:finalValidLinear} shows that $p^2\log^{11/2}(n) / n \rightarrow 0$ is sufficient for asymptotic validity. Corollary~\ref{cor:finalValidNonLinear} considers the setting in which the structural equations are unknown but can be approximated by a known basis. In particular, when the bias decreases at a rate of $K^{-r}$ for some $r > 1/2$, we require (up to log terms) $n^{(1-2r)/2(1+r)}p^{(1+2r)/(1+r)}\rightarrow 0$ for asymptotic validity. This is satisfied under various general conditions. For example, let $f$ be a function on $[0,1]$ with $r$th derivative $f^{(r)}$. If $\int (f^{(r)})^2 < \infty$, then, there exists a $K$-degree polynomial, $f_K$, such that $\E\left((f -f_{K})^2\right) \lesssim K^{-r}$~\citep{barron1991approximation}. Similar statements hold when $f$ is approximated by splines or the Fourier basis. Note that as $r \rightarrow \infty$ and the functions are ``easier'' to approximate, the condition converges to the well-specified setting of $p^2/n \rightarrow 0$.

\begin{corollary}\label{cor:finalValidLinear}
For a fixed $\theta \in \Theta(G)$, suppose that the conditions in Theorem~\ref{thm:highDCLT} hold. Furthermore, suppose the data are known to be generated by a linear structural equation model, so $K = 1$ and $d_{v,i} = 0$ for all $v$ and $i$. When  $\underline{\sigma}^2, \overline{\sigma}^2, \lambda_{\min,C}, \lambda_{\min,Z}, M, h_{\max, 1}$ are fixed and $p^2\log^{11/2}(n)/n \rightarrow 0$, then $\hat \Gamma_{\theta} \rightarrow_p \Gamma_{\theta}$ and
\begin{equation}
\lim_{n \rightarrow \infty} P(\theta \in \hat \Theta(\mathbf{Y}, \alpha)) \geq 1 - \alpha. 
\end{equation}
\end{corollary}

\begin{corollary}\label{cor:finalValidNonLinear}
For a fixed $\theta \in \Theta(G)$, suppose that the conditions in Theorem~\ref{thm:highDCLT} hold. Furthermore, suppose that the data are generated by a structural equation model with unknown functions but that an approximating basis is known such that $d^\star_1\lesssim K^{-r}$, $d^\star_2\lesssim K^{-r}$, and $M_d\lesssim K^{-r}$ for some $r > 1/2$. 
Suppose $\underline{\sigma}^2, \overline{\sigma}^2, \lambda_{\min,C}, \lambda_{\min,Z}, h_{\max,1}, h_{\max,2}, M$ are fixed and let $K =[n^{3/2}/p]^{1/(r+1)}$. If $\log^{11/2}(n)n^{\frac{1-2r}{2(1+r)}}p^{\frac{1+2r}{1+r}} \rightarrow 0$, then $\hat \Gamma_{\theta} \rightarrow_p \Gamma_{\theta}$ and
\begin{equation}\label{eq:asymptoticValid-confSet}\small
\lim_{n \rightarrow \infty} P(\theta \in \hat \Theta(\mathbf{Y}, \alpha)) \geq 1 - \alpha. 
\end{equation}
\end{corollary}

\subsection{Power analysis}
We now consider the power of the goodness-of-fit test. We consider a single regression and define the ``signal strength'' as
\begin{equation}
	\tau^\star = \max_{j \in [J], u \in \pr(v)} \vert\E(h_j(Y_{u,i})\eta_{v, \setminus \pr(v), i})\vert.
\end{equation} 
Theorem~\ref{thm:power} implies that if $(pK)^2/n\rightarrow 0$ and $pK\log^4(n)/\sqrt{n} = o(\tau^\star)$, then a single regression test will be rejected with overwhelming probability for any fixed $\alpha$ level test. Under the alternative, we no longer have $\eta_{v \setminus \pr(v), i} = \varepsilon_{v,i}$ so we require a condition analogous to Assumption~\ref{asm:subExpErrors}.

\begin{assumption}\label{asm:altResid}
Suppose $\Vert \eta_{v \setminus \pr(v), i} \Vert_{\psi_1} < M <\infty$.
\end{assumption}

\begin{theorem}\label{thm:power}
Fix an ordering $\theta \not \in \Theta(G)$ and $v \in V$. Suppose Assumptions~\ref{asm:designCondition}, \ref{asm:testFuncWellCondition}, and \ref{asm:altResid} hold and $(pK)^2/n \rightarrow 0$. When $\lambda_{\min,Z}^{-1}\max(h_{\max,3}, h_{\max, 2})pK\log^4(n)/\sqrt{n} = o(\tau^\star)$, then an $\alpha$-level test for $H_{0,\theta, v}$ will be rejected with probability $1-o(1)$ for any $\alpha \in (0,1)$.

\end{theorem}

 }

\section{Numerical experiments}\label{sec:numerics}
\label{sec:simulations}

In Table~\ref{tab:comp} we compare the proposed goodness-of-fit test for a single linear regression to the procedures of \citet{sen2014testing} (denoted in the table as ``S''), RP Test \citet{shah2018goodness} (``RO'' for OLS version and ``RL'' for Lasso version), MINT \citet{berett2019nonparametric} (``M''), and higher-order least squares \citet{schultheiss2021assessing} (``H''). \updatedText{For our procedure, we use $\mathcal{H} = \{y^2, y^3, \mathrm{sign}(y) \times \vert y \vert^{2.5}, \sin(y), \cos(y), \sin(2y), \cos(2y) \}$ where each of the non-trigonometric functions are standardized to have mean 0 and unit variance.} 

For each replication, we construct a graph by starting with edges $v \rightarrow v+1$ for all $v < p$; for any $u < v - 1$, $u \rightarrow v$ is added with probability $1/2$. \updatedText{Thus, the graphs we consider are dense as our procedure does not require any sparsity assumptions.} For each edge, we sample a linear coefficient uniformly from $\pm (.1, .95)$. We consider settings where all error terms are either uniform, lognormal, gamma, Weibull, or Laplace random variables and a setting---called mixed---where the distribution of each variable in the SEM is randomly selected. We set $p = 10, 20, 45$ and $n \approx p^{5/4}$ or $p^2$. The data is standardized before applying the goodness-of-fit tests.  For each setting of $p, n$, and error distribution, we complete 500 replications.

Table~\ref{tab:compTime} shows the computation time for each procedure; the average time is similar across error distribution so we aggregate the results. RP Test is typically the slowest, followed by \citet{sen2014testing} and MINT. \updatedText{Our proposed procedure is 40-9000x faster than these procedures, and these procedures would be prohibitively slow for computing confidence sets of causal orderings. HOLS is the fastest of the existing procedures and is actually slightly faster than our procedure when $n = p^2$ and $p = 20,45$.}

In addition to the stark computational benefits, the proposed test also performs well statistically. In Table~\ref{tab:compSizePower}, we compare the empirical size and power for tests with nominal level $\alpha = .1$. \updatedText{Recall that the data is generated with true causal ordering $1, 2, \ldots, p$. Thus,} to measure size we test the (true) $H_{0}: \pa(p) \subseteq \{1, \ldots, p-1\} \subseteq \nd(p)$; to measure power we test the (false) $H_{0}: \pa(1) \subseteq \{2, 3, \ldots, p\} \subseteq \nd(1)$. In each setting, if a procedure exhibits empirical size which is significantly larger than $.1$ (corrected for multiple testing), then we do not display the empirical power. 

Our proposed procedure controls the size within the nominal rate in every setting. It also has the highest (or comparable) power in many settings where the errors are skewed but tends to do less well when the errors are symmetric. \updatedText{In Appendix D, we also show results for other choices of test functions. Using only trigonometric functions results in much larger power when errors are symmetric, but they do not perform as well when considering all settings.} \citet{sen2014testing} tends to exceed the nominal size for skewed distributions and performs worse when $n = p^2$ as opposed to $n \approx p^{5/4}$. The OLS variant of \citet{shah2018goodness} performs well across a variety of settings but does not control the size when the errors are uniform and $n = p^2$; the Lasso variant generally fails to control the type I error when $p = 45$ as the linear model is not sparse. MINT exhibits an inflated size when the errors are heavy tailed, but generally has good power in settings where the size is controlled. Finally, HOLS controls empirical size across a wide variety of settings---except the lognormal errors. When $n = p^2$, HOLS tends to have good power when the errors are symmetric, but suffers when the errors are skewed.

\begin{table}[p]
\centering
\caption{\label{tab:comp}Comparison of goodness-of-fit tests ``W'' is proposed procedure, ``S'' \citet{sen2014testing}, ``RO'' and ``RL'' are the OLS and lasso variants of \citet{shah2018goodness}, ``M'' \citet{berett2019nonparametric}, and ``H'' \citet{schultheiss2021assessing}.}
\centering
\footnotesize 
\subcaption{\label{tab:compTime}Average time (sec) for each goodness of fit test}\updatedText{
\begin{tabular}{|c|cccccc|cccccc|}
\hline
  \multicolumn{1}{|c|}{} & \multicolumn{6}{c|}{$n\approx p^{5/4}$} & \multicolumn{6}{c|}{$n = p^2$} \\
 \hline
 $p$ & W & S & RO & RL & M & H & W & S & RO & RL & M & H \\  
  \hline
10 & .001 & .100 & 2 & 2 & .054 & .004 & .003 & .279 & 19 & 10 & .252 & .005 \\ 
  20 & .003 & .172 & 9 & 5 & .101 & .007 & .022 & 3 & 195 & 107 & 3 & .011 \\ 
  45 & .011 & .661 & 73 & 24 & .447 & .018 & .596 & 94 & 3708 & 1403 & 113 & .115 \\ \hline
\end{tabular}}
\subcaption{\label{tab:compSizePower}Empirical size and power of $\alpha = .1$ tests. Size: bolded values exceed the nominal $\alpha = .1$ by 2 standard deviations. Power: bolded values indicate the procedure has the largest power (or is within 2 standard deviations) for that particular setting. Each proportion in the table has been multiplied by 100. If the procedure's empirical size is significantly above the nominal level, the empirical power is not displayed.}
\updatedText{
\begin{tabular}{|c|c|c|cccccc|cccccc|}
 \hline
& &  & \multicolumn{6}{c|}{Size} & \multicolumn{6}{c|}{Power} \\
 & Dist & $p$ & W & S & RO & RL & M & H  & W & S & RO & RL & M & H \\ 
\hline\multirow{18}{*}{\rotatebox[origin=c]{90}{$n \approx p^{5/4}$}} & \multirow{3}{*}{\rotatebox[origin=c]{0}{\bftab gamma}} & 10 & 4 & 5 & 9 & 8 & 11 & 0 & 14 & 7 & 18 & \bftab 25 & 14 & 0 \\ 
   &  & 20 & 7 & 9 & 9 & 8 & 11 & 0 & 26 & 16 & 24 & \bftab 35 & 20 & 0 \\ 
   &  & 45 & 9 & \bftab 19 & 4 & \bftab 20 & \bftab 19 & 1 & \bftab 44 &  & 33 &  &  & 1 \\ 
  \cline{2-15} & \multirow{3}{*}{\rotatebox[origin=c]{0}{\bftab laplace}} & 10 & 5 & 4 & 9 & 8 & 8 & 0 & 7 & 6 & \bftab 13 & 9 & \bftab 12 & 0 \\ 
   &  & 20 & 8 & 9 & 8 & 8 & 11 & 0 & \bftab 11 & 10 & \bftab 12 & 10 & \bftab 13 & 0 \\ 
   &  & 45 & 11 & \bftab 14 & 4 & \bftab 16 & 11 & 1 & 12 &  & 10 &  & \bftab 16 & 1 \\ 
  \cline{2-15} & \multirow{3}{*}{\rotatebox[origin=c]{0}{\bftab lognormal}} & 10 & 8 & 7 & 9 & 8 & 9 & 0 & \bftab 21 & 9 & \bftab 25 & \bftab 24 & 20 & 0 \\ 
   &  & 20 & 10 & 13 & 4 & 8 & \bftab 22 & 0 & \bftab 44 & 18 & 36 & \bftab 41 &  & 0 \\ 
   &  & 45 & 12 & \bftab 20 & 2 & \bftab 24 & \bftab 51 & 4 & \bftab 71 &  & 49 &  &  & 6 \\ 
  \cline{2-15} & \multirow{3}{*}{\rotatebox[origin=c]{0}{\bftab mixed}} & 10 & 8 & 6 & 11 & 7 & 9 & 0 & 11 & 6 & 15 & \bftab 19 & 13 & 0 \\ 
   &  & 20 & 9 & 10 & 6 & 7 & 11 & 0 & 25 & 14 & 24 & \bftab 37 & 19 & 0 \\ 
   &  & 45 & 10 & \bftab 16 & 6 & \bftab 18 & \bftab 19 & 1 & \bftab 39 &  & \bftab 36 &  &  & 4 \\ 
  \cline{2-15} & \multirow{3}{*}{\rotatebox[origin=c]{0}{\bftab uniform}} & 10 & 5 & 3 & 10 & 9 & 8 & 0 & 5 & 3 & 9 & \bftab 14 & 6 & 0 \\ 
   &  & 20 & 8 & 7 & 10 & 7 & 8 & 0 & 8 & 6 & 8 & \bftab 12 & 5 & 0 \\ 
   &  & 45 & 7 & 11 & 10 & \bftab 20 & 5 & 1 & 8 & \bftab 11 & \bftab 10 &  & 6 & 1 \\ 
  \cline{2-15} & \multirow{3}{*}{\rotatebox[origin=c]{0}{\bftab weibull}} & 10 & 5 & 5 & 10 & 11 & 10 & 0 & 21 & 8 & 24 & \bftab 30 & 22 & 0 \\ 
   &  & 20 & 10 & 13 & 5 & 7 & \bftab 15 & 0 & 43 & 21 & 39 & \bftab 48 &  & 0 \\ 
   &  & 45 & 12 & \bftab 17 & 2 & \bftab 17 & \bftab 26 & 1 & \bftab 66 &  & 47 &  &  & 2 \\ 
   \hline
  \hline\multirow{18}{*}{\rotatebox[origin=c]{90}{$n = p^{2}$}} & \multirow{3}{*}{\rotatebox[origin=c]{0}{\bftab gamma}} & 10 & 10 & \bftab 16 & 1 & 8 & 11 & 9 & 88 &  & 81 & 79 & \bftab 92 & 46 \\ 
   &  & 20 & 9 & \bftab 17 & 0 & 6 & 6 & 12 & \bftab 99 &  & \bftab 99 & \bftab 100 & \bftab 100 & 90 \\ 
   &  & 45 & 10 & \bftab 14 & 0 & \bftab 32 & 7 & 14 & \bftab 100 &  & \bftab 100 &  & \bftab 100 & \bftab 100 \\ 
  \cline{2-15} & \multirow{3}{*}{\rotatebox[origin=c]{0}{\bftab laplace}} & 10 & 11 & 13 & 4 & 11 & 12 & 9 & 23 & 30 & 18 & 18 & \bftab 50 & 32 \\ 
   &  & 20 & 10 & 10 & 0 & 9 & \bftab 34 & 9 & 36 & 44 & 31 & 35 &  & \bftab 65 \\ 
   &  & 45 & 9 & 10 & 0 & \bftab 28 & \bftab 52 & 10 & 41 & 55 & 66 &  &  & \bftab 96 \\ 
  \cline{2-15} & \multirow{3}{*}{\rotatebox[origin=c]{0}{\bftab lognormal}} & 10 & 12 & \bftab 26 & 1 & 7 & \bftab 44 & \bftab 14 & \bftab 96 &  & 88 & 80 &  &  \\ 
   &  & 20 & 8 & \bftab 18 & 0 & 7 & \bftab 74 & \bftab 17 & \bftab 100 &  & \bftab 100 & 99 &  &  \\ 
   &  & 45 & 11 & \bftab 18 & 0 & \bftab 32 & \bftab 98 & \bftab 19 & \bftab 100 &  & \bftab 100 &  &  &  \\ 
  \cline{2-15} & \multirow{3}{*}{\rotatebox[origin=c]{0}{\bftab mixed}} & 10 & 8 & \bftab 20 & 4 & 10 & \bftab 14 & 11 & \bftab 80 &  & \bftab 77 & 75 &  & 61 \\ 
   &  & 20 & 11 & \bftab 19 & 3 & 8 & \bftab 17 & 10 & \bftab 97 &  & \bftab 97 & \bftab 98 &  & \bftab 97 \\ 
   &  & 45 & 10 & 11 & 7 & \bftab 34 & \bftab 30 & 13 & \bftab 100 & 98 & 99 &  &  & \bftab 100 \\ 
  \cline{2-15} & \multirow{3}{*}{\rotatebox[origin=c]{0}{\bftab uniform}} & 10 & 11 & 10 & 12 & 10 & 1 & 6 & 5 & 15 & 20 & \bftab 28 & 6 & 21 \\ 
   &  & 20 & 9 & 11 & \bftab 18 & 10 & 0 & 8 & 8 & 25 &  & \bftab 64 & 0 & \bftab 62 \\ 
   &  & 45 & 12 & 12 & \bftab 33 & \bftab 28 & 0 & 8 & 35 & 28 &  &  & 0 & \bftab 93 \\ 
  \cline{2-15} & \multirow{3}{*}{\rotatebox[origin=c]{0}{\bftab weibull}} & 10 & 11 & \bftab 21 & 2 & 10 & \bftab 31 & 11 & \bftab 95 &  & 91 & 83 &  & 75 \\ 
   &  & 20 & 9 & \bftab 19 & 0 & 6 & \bftab 27 & 12 & \bftab 100 &  & \bftab 100 & \bftab 100 &  & 99 \\ 
   &  & 45 & 8 & \bftab 18 & 0 & \bftab 35 & \bftab 38 & \bftab 15 & \bftab 100 &  & \bftab 100 &  &  &  \\ 
   \hline
\end{tabular}
}

\end{table}

\subsection{Confidence sets}\label{sec:numerics-ConfSets} 
Fig.~\ref{fig:confSet} shows results for \updatedText{400} replicates when constructing  90\% confidence sets using Alg.~\ref{alg:branchAndBound} for linear SEMs. We fix $p = 10$ and let $n = 500, 1000, 2500, 5000$. We generate random graphs and data as before with two changes: $u \rightarrow v$ for all $u < v - 1$ is included with probability $1/3$ and each linear coefficient is drawn from $\beta_{u,v} = z_{u,v} \times g_{u,v}$ where $z_{u,v}$ is a Rademacher random variable and $g_{u,v} \sim \text{Gamma}(n^{-1/10},1)$.

The upper left panel shows the proportion of times that the true causal ordering is recovered by DirectLiNGAM~\citep{shimizu2011direct} \updatedText{implemented in the \texttt{causalXtreme} package by \citet{gnecco2023causalXtreme}}. By construction, the average edge weight decreases with $n$ so the causal ordering is not consistently estimated as $n$ increases. Nonetheless, the bottom left panel shows that the empirical coverage of the confidence sets are all very close to the nominal rate of $.9$. In addition, the upper middle panel shows that the confidence sets are still increasingly informative in that the proportion of all $10!$ possible orderings which are included in $\hat \Theta(\mathbf{Y}, .1)$ decreases. The proposed procedure is more powerful (i.e., returns a smaller confidence set) for the skewed distributions; however, even with symmetric errors, the confidence sets still only contains roughly $ 2\%$ of all orderings when $n = 5000$. The bottom middle panel shows the proportion of all pairwise ancestral relationships which are certified into $\mathcal{\hat A}_\cap$. Again, despite inconsistent estimation of the causal ordering, the proportion of ancestral relations which are recovered increases with $n$. 

The top right panel shows the time required to calculate the confidence set when $p = 10$. The computation time is not monotonic with $n$ because a larger sample size requires more computation for each considered ordering; however, this may be offset by increased power to reject incorrect orderings so the branch and bound procedure considers fewer orderings.
Finally, in the bottom right panel, we show computational feasibility for larger $p$ by displaying the median computation time (sec$\times 1000$) for $10$ replicates with $n = 10000$. We consider gamma errors for $p = 10, \ldots, 20$. We draw random graphs and data as before, except the linear coefficients are selected uniformly from $(-1,1)$. Although the computational time increases rapidly, the procedure is still feasible for $p = 20$.

\begin{figure}[t] 
\centering
\includegraphics[scale = .7]{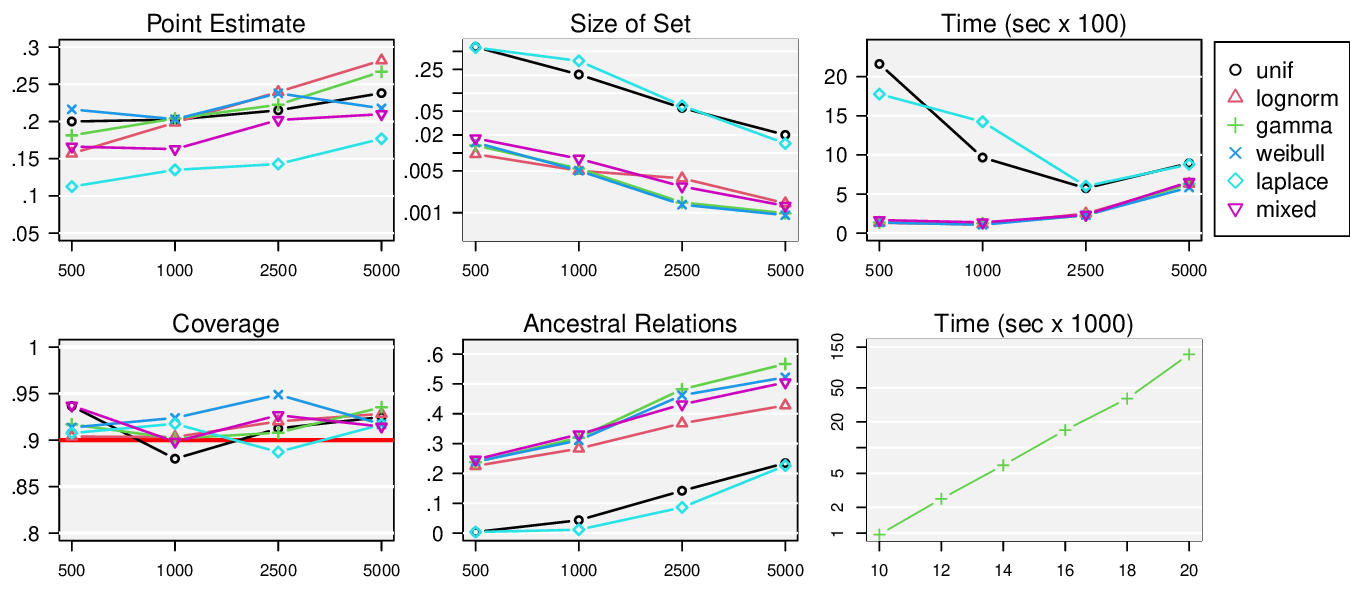}
\caption{\label{fig:confSet} \updatedText{Results for linear SEMs.
\emph{Top Left}: \% of times the point estimate of ordering is correct.
\emph{Top Middle}: Avg proportion of all possible orderings included in $\hat \Theta(\mathbf{Y}, \alpha = .9)$.
\emph{Top Right}: Avg time (sec$\times 100$) for each confidence set.
\emph{Bottom Left}: Coverage for $\hat \Theta(\mathbf{Y}, \alpha = .9)$.
\emph{Bottom Middle}:  Avg proportion of all ancestral relations which are included in $\mathcal{\hat A}_\cap$.
\emph{Bottom Right}: Median time (sec$\times 1000$) for each confidence set with $n = 10,000$ and $p$ varying.}
}
\end{figure}

\updatedText{Fig.~\ref{fig:confSetNL} shows results for 400 replicates when constructing  90\% confidence sets for non-linear SEMs. We let $p = 7$, $n = 2500, 5000, 7500, 1000$ and consider gamma or Laplace errors. We let $f_v = \sum_{u \in \pa(v)}f_{vu}(Y_u)$ and consider two settings: (1) $f_{vu}$ is a 5th degree polynomial and (2) a much more difficult setting from ~\citet{buhlmann2014cam} where $f_{vu}$ is a sigmoid. To model $f_v$ in the polynomial setting we use a $K=2, \ldots,5$ degree polynomial basis. In the sigmoid setting, we use b-splines with $K = 20,40,60$ degrees of freedom and $f_v \not \in \Phi^{(v)}$ for any $K$. Additional details are given in the appendix. 

In the polynomial setting, when $K$ is small and the bias is large, almost all orderings are rejected and coverage is poor. However, as expected, when $K =4, 5$ and the approximation error is small or vanishes, the linear results generalize directly to this non-linear setting. The sigmoid setting shows a qualitatively similar result. When $K= 20$, the approximation error is large and the confidence sets undercover when $n$ is large. When $K = 40, 60$, the procedure attains nominal coverage; however, because $pK$---the number of regressors---is large, we do not have good power to reject incorrect orderings. Thus, the resulting confidence set is less informative.}

\begin{figure}[t] 
\centering
\includegraphics[scale = .65]{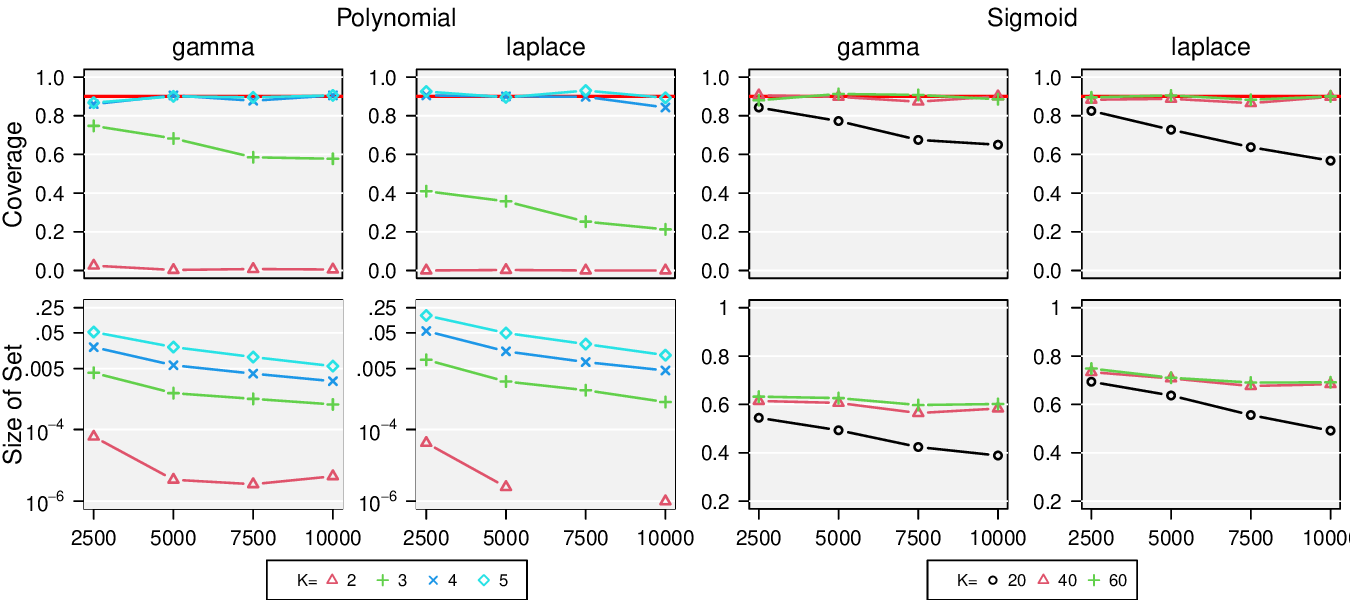}
\caption{\label{fig:confSetNL} \updatedText{Results for non-linear SEMs; Horizontal axis denotes $n$.
\emph{Left half:} true functions are $5$th degree polynomial with gamma and Laplace errors. 
\emph{Right half:} true functions are sigmoids. 
\emph{Top half:} proportion of times the true ordering is contained in the confidence set.
\emph{Bottom half:} proportion of all $7!$ orderings which are included in the confidence set. If the proportion is $0$, then the point is omitted.}}
\end{figure}

\subsection{Confidence intervals with model selection uncertainty}
We now consider CIs for causal effects which account for model uncertainty. We draw random graphs and data under the same setup in Section~\ref{sec:numerics-ConfSets} with $p = 10$, except we draw the magnitude of the linear coefficients from $ \text{Gamma}(1/2,1)$, so the ``signal strength'' is fixed instead of decreasing with $n$ as before. We compute 80\% CIs for the total effect of $Y_4$ onto $Y_7$ using Alg.~\ref{alg:TotalEffect}. We contrast this with a na\"{i}ve procedure that computes CIs using only the adjustment set implied by $\hat \theta$, the causal ordering estimated by DirectLiNGAM. Specifically, if $\hat \theta(4) < \hat \theta(7)$, then the na\"{i}ve procedure uses an adjustment set of ${\rm pr}_{\hat \theta}(4)$ and returns the typical 80\% CI for the regression coefficient of $Y_4$. 
When $\hat \theta(4) > \hat \theta(7)$, $7$ is estimated to be an non-descendant of $4$, so the returned CI is $0$. 

Table~\ref{tab:ci-causal-effects} compares the empirical coverage and lengths of the proposed CIs and the na\"{i}ve CIs. The proposed CIs have empirical coverage above the nominal rate. The ``Adj'' column shows the proportion of times the point estimate yields a valid adjustment set for the total effect. Given that these values are much smaller than $1$, it is unsurprising that the na\"{i}ve CIs cover well below the nominal rate. Under gamma errors---when the model uncertainty is typically smaller---the proposed CIs have median length that may be only twice as large as the naive procedure. However, in the Laplace setting where model uncertainty is larger, the median length is roughly 5 times longer.

\begin{table}[bt] 
\footnotesize 
\centering\caption{\label{tab:ci-causal-effects}Comparison of 80\% CIs for the total effect of $Y_4$ onto $Y_7$. ``MU'' denotes CIs which account for model uncertainty; ``NV'' denotes na\"{i}ve CIs. Each setting has 400 replicates.}
\updatedText{\begin{tabular}{|c|cc|cc|cc|c|cc|cc|cc|c|}
  \hline
  \multirow{3}{*}{n} &\multicolumn{7}{|c|}{Gamma} &\multicolumn{7}{|c|}{Laplace}\\ \cline{2-15}
 & \multicolumn{2}{|c|}{Coverage} & \multicolumn{2}{|c|}{Avg Len} & \multicolumn{2}{|c|}{Med Len} & \multirow{2}{*}{Adj} & \multicolumn{2}{|c|}{Coverage} & \multicolumn{2}{|c|}{Avg Len} & \multicolumn{2}{|c|}{Med Len} & \multirow{2}{*}{Adj} \\ 
 &  MU & NV & MU & NV & MU & NV & & MU & NV & MU & NV & MU & NV &  \\ 
  \hline
250 & .97 & .62 & .68 & .18 & .50 & .18 & .37 & .99 & .56 & 1.1 & .19 & .71 & .18 & .26 \\ 
  500 & .96 & .65 & .46 & .14 & .29 & .12 & .44 & .99 & .63 & .91 & .13 & .57 & .12 & .33 \\ 
  1000 & .95 & .68 & .30 & .10 & .19 & .09 & .55 & .98 & .69 & .69 & .10 & .45 & .09 & .46 \\ 
  2000 & .93 & .69 & .17 & .07 & .11 & .06 & .53 & .98 & .68 & .49 & .07 & .30 & .06 & .53 \\ 
   \hline
\end{tabular}}
\end{table}

\subsection{Data example}
\label{sec:numerics-data-analysis}
\updatedText{We now analyze data consisting of the daily value-weighted average stock returns for 12 different industry portfolios from 2019 to 2023 ($n = 1258$)\footnote{Data available at: \url{https://mba.tuck.dartmouth.edu/pages/faculty/ken.french/data_library.html}. Accessed: Mar 2024}. All stocks from the NYSE, AMEX, and NASDAQ are placed into one of 12 different industries.

Using DirectLiNGAM, the estimated causal ordering is: Utilities, Business Equipment, Healthcare, Finance, Telecomm, Consumer Non-durables, Manufacturing, Other, Energy, Wholesale, Consumer Durables, Chemicals. The 95\% confidence set of causal orderings returned for the data contains approximately $1/45000$ of the $12!$ total orderings. Notably, Utilities is first in every non-rejected ordering so $ \mathcal{\hat A}_\cap = \{(\text{Utilities}, v) \, : \, v \neq \text{Utilities}\}$, which agrees with the point estimate. At first glance, this may seem odd; however, Utilities are often viewed as a proxy for bonds and directly capture the effect of changing interest rates and market uncertainty. From 2020 to 2022, the performance of American stock markets was largely driven by uncertainty around COVID-19 and  federal monetary interventions. Thus, it makes sense that Utilites are estimated to be an ancestor of all other industries.

Nonetheless, the other orderings in $\hat \Theta(\mathbf{Y}, \alpha = .05)$ have causal implications that differ from the point estimate. The right panel in Fig.~\ref{fig:barycenterDist} summarizes the pairwise ancestral relationships for all non-rejected orderings in the 95\% confidence set, where darker shades of the $(u,v)$ element indicates that $v$ precedes $u$ in a larger proportion of non-rejected orderings. We also compute the Fr\'{e}chet mean of $\hat \Theta(\mathbf{Y}, \alpha = .05)$ using a distance between two orderings which counts the number of implied ancestral relations present in one ordering but not the other; i.e.,
$
d(\theta, \theta') = \left\vert\{(u,v) \, : \, u \in \text{pr}_{\theta}(v) \text{ and } u \not\in \text{pr}_{\theta'}(v)\}\right\vert.    
$
The mean ordering is: Utilities, Energy, Wholesale, Consumer Durables, Finance, Health, Consumer Non-durables, Telecomm, Chemicals, Business Equipment, Manufacturing, and Other. The left panel of Fig.~\ref{fig:barycenterDist} shows the distance for all orderings in $\hat \Theta(\mathbf{Y}, \alpha = .05)$ from the Fr\'{e}chet mean as well as to the point estimate; as expected, the distances to the mean are generally smaller than the distances to the point estimate.

Finally, the naive 90\% CI for the total effect of Energy onto Consumer Durables is $(-.096, 0.021)$. Furthermore, since Manufacturing precedes Chemicals in the estimated ordering, we would naively conclude that the total effect of Chemicals onto Manufacturing is $0$. In contrast, when accounting for model uncertainty, we produce a 90\% CI for the total effect of Energy onto Consumer Durables of $(-.134, .375)$ and a CI for the effect of Chemicals onto Manufacturing of $\{0\} \cup (.268, .413) \cup (.980, 1.093)$. 

\begin{figure}[tb] 
\centering
\includegraphics[scale = .6]{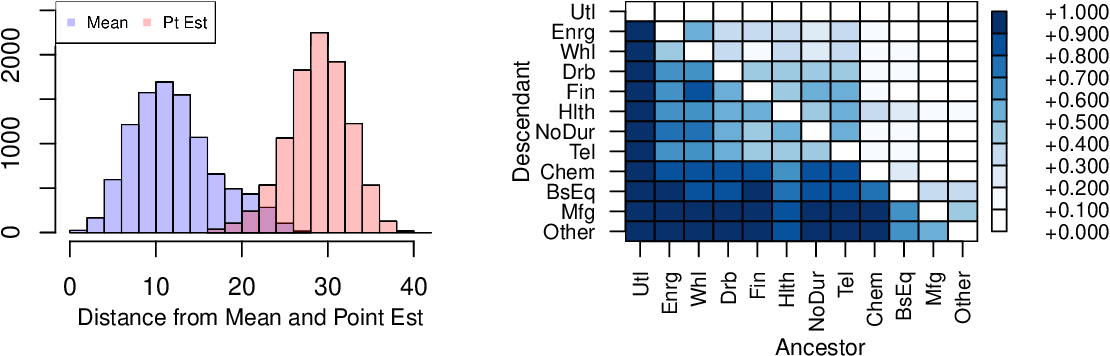}
\caption{\label{fig:barycenterDist}\updatedText{Left: distance from Fr\'{e}chet mean and point estimate to all other orderings in $\hat \Theta(\mathbf{Y}, .05)$. Right: A darker shade of the $(u,v)$ cell indicates that $v$ precedes $u$ in a larger proportion of orderings in $\hat \Theta(\mathbf{Y}, .05)$.}}
\end{figure}
}

\section{Discussion}
We have proposed a procedure for quantifying uncertainty when estimating the causal structure. Our goodness-of-fit testing framework returns a confidence set that may be informative about the validity of the posited identification assumptions, as well as which causal orderings the identifying assumptions cannot rule out. The confidence set can also be used to compute various other objects of interest. Notably, this includes confidence intervals for causal effects that also account for model uncertainty and a sub/superset of ancestral relations. 
Our specific goodness-of-fit test is designed for models in which residuals are independent of regressors under the null hypothesis. 
Future work could extend this procedure to settings where causal sufficiency may not hold~\citep[see, e.g.,][]{wang2020causal} \updatedText{ or models where feedback cycles are permitted}.

While we believe the proposed approach has many desirable characteristics, there are, of course, also are a few previously mentioned disadvantages which could be addressed in future work.
Most notably, the primary disadvantage of our approach is the computational expense required to test the set of all possible orderings. While the proposed branch and bound procedure can handle medium-sized problems and could scale even larger with a careful parallel implementation, the approach is unlikely to scale too far beyond $p = 20$. 
Nonetheless, we believe that this approach is a useful first step. We believe that this initial method is valuable for practitioners and hope it also spurs further research in statistical methodology for quantifying model uncertainty in estimating causal structures.

\bibliographystyle{apalike}
\bibliography{topOrdering}

\begin{thebibliography}{}

\bibitem[Barron and Sheu, 1991]{barron1991approximation}
Barron, A.~R. and Sheu, C.-H. (1991).
\newblock Approximation of density functions by sequences of exponential
  families.
\newblock {\em Ann. Statist.}, 19(3):1347--1369.

\bibitem[Bergsma and Dassios, 2014]{bergsma2014consistent}
Bergsma, W. and Dassios, A. (2014).
\newblock A consistent test of independence based on a sign covariance related
  to {K}endall's tau.
\newblock {\em Bernoulli}, 20(2):1006--1028.

\bibitem[Berrett and Samworth, 2019]{berett2019nonparametric}
Berrett, T.~B. and Samworth, R.~J. (2019).
\newblock Nonparametric independence testing via mutual information.
\newblock {\em Biometrika}, 106(3):547--566.

\bibitem[Bollen and Long, 1993]{bollen1993testing}
Bollen, K.~A. and Long, J.~S. (1993).
\newblock {\em Testing structural equation models}, volume 154.
\newblock Sage.

\bibitem[Boucheron et~al., 2013]{boucheron2013Concentration}
Boucheron, S., Lugosi, G., and Massart, P. (2013).
\newblock {\em {Concentration Inequalities: A Nonasymptotic Theory of
  Independence}}.
\newblock Oxford University Press.

\bibitem[Breusch and Pagan, 1979]{breusch1979heteroscedastic}
Breusch, T.~S. and Pagan, A.~R. (1979).
\newblock A simple test for heteroscedasticity and random coefficient
  variation.
\newblock {\em Econometrica}, 47(5):1287--1294.

\bibitem[B\"{u}hlmann et~al., 2014]{buhlmann2014cam}
B\"{u}hlmann, P., Peters, J., and Ernest, J. (2014).
\newblock C{AM}: causal additive models, high-dimensional order search and
  penalized regression.
\newblock {\em Ann. Statist.}, 42(6):2526--2556.

\bibitem[Chang et~al., 2024]{chang2022order}
Chang, H., Cai, J.~J., and Zhou, Q. (2024).
\newblock Order-based structure learning without score equivalence.
\newblock {\em Biometrika}, 111(2):551--572.

\bibitem[Chen et~al., 2019]{chen2019equal}
Chen, W., Drton, M., and Wang, Y.~S. (2019).
\newblock On causal discovery with an equal-variance assumption.
\newblock {\em Biometrika}, 106(4):973--980.

\bibitem[Chernozhukov et~al., 2023a]{chernozhukov2023high}
Chernozhukov, V., Chetverikov, D., Kato, K., and Koike, Y. (2023a).
\newblock High-dimensional data bootstrap.
\newblock {\em Annu. Rev. Stat. Appl.}, 10:427--449.

\bibitem[Chernozhukov et~al., 2023b]{chernozhukov2023nearly}
Chernozhukov, V., Chetverikov, D., and Koike, Y. (2023b).
\newblock Nearly optimal central limit theorem and bootstrap approximations in
  high dimensions.
\newblock {\em Ann. Appl. Probab.}, 33(3):2374--2425.

\bibitem[Cook and Weisberg, 1983]{cook2983diagnostics}
Cook, R.~D. and Weisberg, S. (1983).
\newblock Diagnostics for heteroscedasticity in regression.
\newblock {\em Biometrika}, 70(1):1--10.

\bibitem[Darmois, 1953]{darmois1953analyse}
Darmois, G. (1953).
\newblock Analyse g\'{e}n\'{e}rale des liaisons stochastiques. {E}tude
  particuli\`ere de l'analyse factorielle lin\'{e}aire.
\newblock {\em Rev. Inst. Internat. Statist.}, 21:2--8.

\bibitem[Fang and Koike, 2021]{xiao2021HighDCLT}
Fang, X. and Koike, Y. (2021).
\newblock High-dimensional central limit theorems by {S}tein's method.
\newblock {\em Ann. Appl. Probab.}, 31(4):1660--1686.

\bibitem[Ferrari and Yang, 2015]{ferrari2015confidence}
Ferrari, D. and Yang, Y. (2015).
\newblock Confidence sets for model selection by {$F$}-testing.
\newblock {\em Statist. Sinica}, 25(4):1637--1658.

\bibitem[Friedman and Koller, 2003]{friedman2003being}
Friedman, N. and Koller, D. (2003).
\newblock Being {B}ayesian about network structure. {A} {B}ayesian approach to
  structure discovery in {B}ayesian networks.
\newblock {\em Machine learning}, 50:95--125.

\bibitem[Gnecco et~al., 2023]{gnecco2023causalXtreme}
Gnecco, N., Meinshausen, N., Peters, J., and Engelke, S. (2023).
\newblock {\em causalXtreme: Causal discovery in heavy-tailed models}.
\newblock R package version 0.0.0.9000.

\bibitem[G\"{o}tze et~al., 2021]{gotze2021polynomials}
G\"{o}tze, F., Sambale, H., and Sinulis, A. (2021).
\newblock Concentration inequalities for polynomials in
  {$\alpha$}-sub-exponential random variables.
\newblock {\em Electron. J. Probab.}, 26:Paper No. 48, 22.

\bibitem[Gretton et~al., 2007]{gretton2007hsic}
Gretton, A., Fukumizu, K., Teo, C.~H., Song, L., Sch{\"{o}}lkopf, B., and
  Smola, A.~J. (2007).
\newblock A kernel statistical test of independence.
\newblock In Platt, J.~C., Koller, D., Singer, Y., and Roweis, S.~T., editors,
  {\em Advances in Neural Information Processing Systems 20, Proceedings of the
  Twenty-First Annual Conference on Neural Information Processing Systems,
  Vancouver, British Columbia, Canada, December 3-6, 2007}, pages 585--592.
  Curran Associates, Inc.

\bibitem[Hansen et~al., 2011]{hansen2011model}
Hansen, P.~R., Lunde, A., and Nason, J.~M. (2011).
\newblock The model confidence set.
\newblock {\em Econometrica}, 79(2):453--497.

\bibitem[Hoyer and Hyttinen, 2009]{hoyer2009bayesian}
Hoyer, P.~O. and Hyttinen, A. (2009).
\newblock {B}ayesian discovery of linear acyclic causal models.
\newblock In {\em Proceedings of the Twenty-Fifth Conference on Uncertainty in
  Artificial Intelligence}, pages 240--248.

\bibitem[Jankov\'{a} and van~de Geer, 2019]{jankova2018inference}
Jankov\'{a}, J. and van~de Geer, S. (2019).
\newblock Inference in high-dimensional graphical models.
\newblock In {\em Handbook of graphical models}, Chapman \& Hall/CRC Handb.
  Mod. Stat. Methods, pages 325--349. CRC Press, Boca Raton, FL.

\bibitem[Kuipers and Moffa, 2017]{kuipers2017partition}
Kuipers, J. and Moffa, G. (2017).
\newblock Partition {MCMC} for inference on acyclic digraphs.
\newblock {\em J. Amer. Statist. Assoc.}, 112(517):282--299.

\bibitem[Lei, 2020]{lei2017cross}
Lei, J. (2020).
\newblock Cross-validation with confidence.
\newblock {\em J. Amer. Statist. Assoc.}, 115(532):1978--1997.

\bibitem[Li et~al., 2020]{li2019lrt}
Li, C., Shen, X., and Pan, W. (2020).
\newblock Likelihood ratio tests for a large directed acyclic graph.
\newblock {\em J. Amer. Statist. Assoc.}, 115(531):1304--1319.

\bibitem[Loh and B\"{u}hlmann, 2014]{loh2014learning}
Loh, P.-L. and B\"{u}hlmann, P. (2014).
\newblock High-dimensional learning of linear causal networks via inverse
  covariance estimation.
\newblock {\em J. Mach. Learn. Res.}, 15:3065--3105.

\bibitem[Maathuis et~al., 2019]{handbook}
Maathuis, M., Drton, M., Lauritzen, S., and Wainwright, M., editors (2019).
\newblock {\em Handbook of graphical models}.
\newblock Chapman \& Hall/CRC Handbooks of Modern Statistical Methods. CRC
  Press, Boca Raton, FL.

\bibitem[Maathuis et~al., 2009]{maathuis2009estimating}
Maathuis, M.~H., Kalisch, M., and B\"{u}hlmann, P. (2009).
\newblock Estimating high-dimensional intervention effects from observational
  data.
\newblock {\em Ann. Statist.}, 37(6A):3133--3164.

\bibitem[Nazarov, 2003]{nazarov2003maximal}
Nazarov, F. (2003).
\newblock On the maximal perimeter of a convex set in {${\Bbb R}^n$} with
  respect to a {G}aussian measure.
\newblock In {\em Geometric aspects of functional analysis}, volume 1807 of
  {\em Lecture Notes in Math.}, pages 169--187. Springer, Berlin.

\bibitem[Niinim\"{a}ki et~al., 2016]{niinimaki2016structure}
Niinim\"{a}ki, T., Parviainen, P., and Koivisto, M. (2016).
\newblock Structure discovery in {B}ayesian networks by sampling partial
  orders.
\newblock {\em J. Mach. Learn. Res.}, 17:Paper No. 57, 47.

\bibitem[Nowack et~al., 2020]{nowack2020causal}
Nowack, P., Runge, J., Eyring, V., and Haigh, J.~D. (2020).
\newblock Causal networks for climate model evaluation and constrained
  projections.
\newblock {\em Nature {C}ommunications}, 11(1):1--11.

\bibitem[Oliveira, 2016]{oliveira2016lower}
Oliveira, R.~I. (2016).
\newblock The lower tail of random quadratic forms with applications to
  ordinary least squares.
\newblock {\em Probab. Theory Related Fields}, 166(3-4):1175--1194.

\bibitem[Pearl, 2009]{pearl2009causality}
Pearl, J. (2009).
\newblock {\em Causality}.
\newblock Cambridge university press.

\bibitem[Peters et~al., 2016]{peters2016invariant}
Peters, J., B\"{u}hlmann, P., and Meinshausen, N. (2016).
\newblock Causal inference by using invariant prediction: identification and
  confidence intervals.
\newblock {\em J. R. Stat. Soc. Ser. B. Stat. Methodol.}, 78(5):947--1012.
\newblock With comments and a rejoinder.

\bibitem[Peters et~al., 2017]{peters2017elements}
Peters, J., Janzing, D., and Sch\"{o}lkopf, B. (2017).
\newblock {\em Elements of causal inference}.
\newblock Adaptive Computation and Machine Learning. MIT Press, Cambridge, MA.
\newblock Foundations and learning algorithms.

\bibitem[Peters et~al., 2014]{peters2014additive}
Peters, J., Mooij, J.~M., Janzing, D., and Sch\"{o}lkopf, B. (2014).
\newblock Causal discovery with continuous additive noise models.
\newblock {\em J. Mach. Learn. Res.}, 15:2009--2053.

\bibitem[Pfister et~al., 2018]{pfister2018kernel}
Pfister, N., B\"{u}hlmann, P., Sch\"{o}lkopf, B., and Peters, J. (2018).
\newblock Kernel-based tests for joint independence.
\newblock {\em J. R. Stat. Soc. Ser. B. Stat. Methodol.}, 80(1):5--31.

\bibitem[Raskutti and Uhler, 2018]{raskutti2018learning}
Raskutti, G. and Uhler, C. (2018).
\newblock Learning directed acyclic graph models based on sparsest
  permutations.
\newblock {\em Stat}, 7:e183, 14.

\bibitem[Sachs et~al., 2005]{sachs2005causal}
Sachs, K., Perez, O., Pe'er, D., Lauffenburger, D.~A., and Nolan, G.~P. (2005).
\newblock Causal protein-signaling networks derived from multiparameter
  single-cell data.
\newblock {\em Science}, 308(5721):523--529.

\bibitem[Schultheiss et~al., 2023]{schultheiss2021assessing}
Schultheiss, C., Bühlmann, P., and Yuan, M. (2023).
\newblock Higher-order least squares: Assessing partial goodness of fit of
  linear causal models.
\newblock {\em J. Amer. Statist. Assoc.}, 0(0):1--13.

\bibitem[Sen and Sen, 2014]{sen2014testing}
Sen, A. and Sen, B. (2014).
\newblock Testing independence and goodness-of-fit in linear models.
\newblock {\em Biometrika}, 101(4):927--942.

\bibitem[Shah and B\"{u}hlmann, 2018]{shah2018goodness}
Shah, R.~D. and B\"{u}hlmann, P. (2018).
\newblock Goodness-of-fit tests for high dimensional linear models.
\newblock {\em J. R. Stat. Soc. Ser. B. Stat. Methodol.}, 80(1):113--135.

\bibitem[Shen et~al., 2020]{shen2020challenges}
Shen, X., Ma, S., Vemuri, P., and Simon, G. (2020).
\newblock Challenges and opportunities with causal discovery algorithms:
  application to alzheimer’s pathophysiology.
\newblock {\em Scientific {R}eports}, 10(1):1--12.

\bibitem[Shi et~al., 2023]{shi2021testing}
Shi, C., Zhou, Y., and Li, L. (2023).
\newblock Testing directed acyclic graph via structural, supervised and
  generative adversarial learning.
\newblock {\em J. Amer. Statist. Assoc.}, 0(0):1--14.

\bibitem[Shimizu and Bollen, 2014]{shimizu2014bayesian}
Shimizu, S. and Bollen, K. (2014).
\newblock Bayesian estimation of causal direction in acyclic structural
  equation models with individual-specific confounder variables and
  non-{G}aussian distributions.
\newblock {\em J. Mach. Learn. Res.}, 15:2629--2653.

\bibitem[Shimizu et~al., 2006]{shohei2006lingam}
Shimizu, S., Hoyer, P.~O., Hyv\"{a}rinen, A., and Kerminen, A. (2006).
\newblock A linear non-{G}aussian acyclic model for causal discovery.
\newblock {\em J. Mach. Learn. Res.}, 7:2003--2030.

\bibitem[Shimizu et~al., 2011a]{shimizu2011directlingam}
Shimizu, S., Inazumi, T., Sogawa, Y., Hyv{\"a}rinen, A., Kawahara, Y., Washio,
  T., Hoyer, P.~O., and Bollen, K. (2011a).
\newblock Directlingam: A direct method for learning a linear non-{G}aussian
  structural equation model.
\newblock {\em Journal of Machine Learning Research}, 12(Apr):1225--1248.

\bibitem[Shimizu et~al., 2011b]{shimizu2011direct}
Shimizu, S., Inazumi, T., Sogawa, Y., Hyv\"arinen, A., Kawahara, Y., Washio,
  T., Hoyer, P.~O., and Bollen, K. (2011b).
\newblock Direct{L}i{NGAM}: a direct method for learning a linear
  non-{G}aussian structural equation model.
\newblock {\em J. Mach. Learn. Res.}, 12:1225--1248.

\bibitem[Shojaie and Michailidis, 2010]{shojaie2010penalized}
Shojaie, A. and Michailidis, G. (2010).
\newblock Penalized likelihood methods for estimation of sparse
  high-dimensional directed acyclic graphs.
\newblock {\em Biometrika}, 97(3):519--538.

\bibitem[Skitovich, 1954]{skitovich1954linear}
Skitovich, V.~P. (1954).
\newblock Linear forms of independent random variables and the normal
  distribution law.
\newblock {\em Izvestiya Rossiiskoi Akademii Nauk. Seriya Matematicheskaya},
  18(2):185--200.

\bibitem[Solus et~al., 2021]{solus2021permutation}
Solus, L., Wang, Y., and Uhler, C. (2021).
\newblock Consistency guarantees for greedy permutation-based causal inference
  algorithms.
\newblock {\em Biometrika}, 108(4):795--814.

\bibitem[Spirtes et~al., 2000]{spirtes2000causation}
Spirtes, P., Glymour, C., and Scheines, R. (2000).
\newblock {\em Causation, prediction, and search}.
\newblock Adaptive Computation and Machine Learning. MIT Press, Cambridge, MA,
  second edition.
\newblock With additional material by David Heckerman, Christopher Meek,
  Gregory F. Cooper and Thomas Richardson, A Bradford Book.

\bibitem[Strieder and Drton, 2023]{strieder2023equal}
Strieder, D. and Drton, M. (2023).
\newblock Confidence in causal inference under structure uncertainty in linear
  causal models with equal variances.
\newblock {\em Journal of Causal Inference}, 11(1):20230030.

\bibitem[Strieder and Drton, 2024]{strieder2024dual}
Strieder, D. and Drton, M. (2024).
\newblock Dual likelihood for causal inference under structure uncertainty.
\newblock In Locatello, F. and Didelez, V., editors, {\em Proceedings of the
  Third Conference on Causal Learning and Reasoning}, volume 236 of {\em
  Proceedings of Machine Learning Research}, pages 1--17. PMLR.

\bibitem[Strobl et~al., 2019]{strobl2019pvalues}
Strobl, E.~V., Spirtes, P.~L., and Visweswaran, S. (2019).
\newblock Estimating and controlling the false discovery rate of the pc
  algorithm using edge-specific p-values.
\newblock {\em ACM Trans. Intell. Syst. Technol.}, 10(5).

\bibitem[Wang and Drton, 2020]{wang2020high}
Wang, Y.~S. and Drton, M. (2020).
\newblock High-dimensional causal discovery under non-{G}aussianity.
\newblock {\em Biometrika}, 107(1):41--59.

\bibitem[Wang and Drton, 2023]{wang2020causal}
Wang, Y.~S. and Drton, M. (2023).
\newblock Causal discovery with unobserved confounding and non-{G}aussian data.
\newblock {\em J. Mach. Learn. Res.}, 24:Paper No. [271], 61.

\bibitem[Zheng et~al., 2019]{zheng2019model}
Zheng, C., Ferrari, D., and Yang, Y. (2019).
\newblock Model selection confidence sets by likelihood ratio testing.
\newblock {\em Statist. Sinica}, 29(2):827--851.

\end{thebibliography}

\bigskip

\newpage

\appendix
\section{Additional data analysis details}
We show two additional plots for the data analysis in Section~\ref{sec:numerics-data-analysis}. In particular, the left hand plot shows the length of the ``honest'' 90\% confidence interval for the total effect between each pair of nodes. The plot on the right hand side shows the number of adjustment sets considered when constructing the confidence interval. Generally, if a total effect considers more adjustment sets this results in a wider confidence interval. 

Because Utilities appears first in every ordering in the 95\% confidence set of causal orderings, when estimating the total effect of Utilities onto any other industry, only one adjustment set---the empty set---is considered. Thus, the resulting confidence intervals are quite short (all are less than .25). In addition, because Utilities always precedes the other industries, the estimated total effect of any other industry onto Utilities is $0$.

\begin{figure}[h]
    \centering
    \includegraphics[scale = .8]{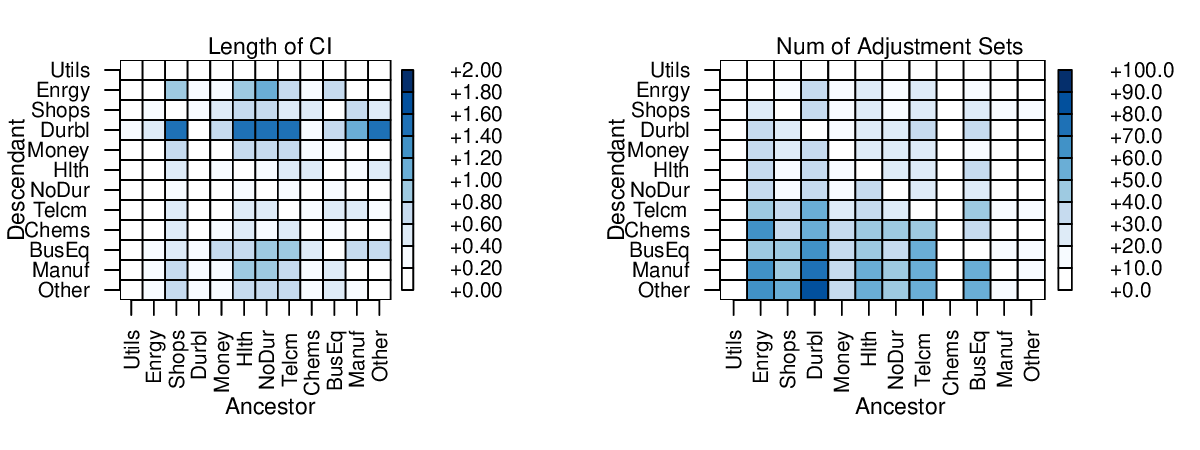}
    \caption{\label{fig:supp_numAdjSets}Left: Length of 90\% confidence interval for the total effect between each pair of nodes. Right: Number of adjustment sets considered when computing 90\% confidence interval for the total effect between each pair of nodes.}
    \label{fig:enter-label}
\end{figure}

\newpage

\section{Non-linear SEMs}
We describe the details for the simulations with non-linear SEMs. 
In particular, we generate a random graph with $p = 7$ by always including the edge $u \rightarrow u + 1$ for $u = 1, \ldots, 6$ and including the edge $u \rightarrow v$ with probability $\pi$ if $u < v -1$. In the settings below, we consider $\pi = 1/3$. Each setting is replicated 400 times.  To generate the data, we consider two types of functions. In the first setting, we consider non-linear functions which lie in a well-specified parametric class. Specifically, we generate data from the following simulation setting
\[Y_v = f_v(Y_{\pa(v)}) + \varepsilon_v = \frac{1}{\vert \pa(v)\vert }\sum_{u \in pa(v)}f_{v,u}(Y_u) + \varepsilon_v.\]
We set $f_{v,u}$ to be a scaled and centered version of $g_{v,u}$, where $g_{v,u}(Y_u) = \sum_{w=1}^5 \frac{b_w}{w!}Y_u^w$ with $b_w \sim N(0,1)$. We approximate the unknown function with a polynomial basis so that the misspecification vanishes when the basis is ``large enough.'' We also consider a very challenging setting where the misspecification decreases with the size of the approximating basis, but never actually disappears. Specifically, we follow the setup used in \citet{buhlmann2014cam} where $Y_v$ is again an additive function of univariate non-linear functions of it's parents
\[Y_v = f_v(Y_{\pa(v)}) + \varepsilon_v = \sum_{u \in pa(v)}f_{v,u}(Y_u) + \varepsilon_v.\]
However, in this setting each $f_{v,u}$ we use the sigmoid type functions:
\[f_{v,u}(Y_u) = a \frac{b(Y_u + c)}{1 + \vert b(Y_u + c) \vert}\]
where $a \sim 1 + \text{Exp}(4)$, $b \sim \pm \text{Unif}(.5, 2)$ and $c \sim \text{Unif}(-2, 2)$. 

Following~\citet{buhlmann2014cam} in both the polynomial and sigmoid setting: if $v$ is a root, then $\var(\varepsilon_v) \sim \text{Unif}(1, \sqrt{2})$, if $v$ is not a root, then $\var(\varepsilon_v) \sim \text{Unif}(1/5, \sqrt{2}/5)$. In a slight departure from \citet{buhlmann2014cam} who consider Gaussian errors, we let $\varepsilon_v$ be either a centered Gamma or Laplace random variable, as opposed to a Gaussian.

For each ordering $\theta$, we estimate $\hat f_v$ by regressing $Y_v$ onto $\Phi= (\phi_{u,k}(Y_u) :  u \in \text{pr}_\theta(v), k \in [K])$. In the polynomial setting, $(\phi_{u,k} : k \in [K])$ is the polynomial basis of $Y_u$ with degree $K = 2, \ldots, 5$. In the sigmoid setting, $(\phi_{u,k} : k \in [K])$ is a univariate b-spline basis of $Y_u$ with $K = 20, 40, 60$ degrees of freedom.

We use the test functions which are the union of the 
\begin{itemize}
    \item $y^2, y^3, \text{sign}(y)\vert y\vert ^{2.5}$
    \item $\bigcup_{j =1 }^{10} \{\sin(y\omega_j), \cos(y\omega_j)\}$ where $\omega_j \sim N(0,1)$
    \item $\sin(y^2),\cos(y^2),\sin(y)y, \cos(y)y, \sin(y^2)y, \cos(y^2)y, \tanh(y)$
\end{itemize}

\FloatBarrier
\newpage

\section{Example of naive strategy}
The following example examines the naive strategy of regressing a parent onto a child and directly testing independence of the residuals and regressors. Specifically, Table~\ref{tab:ex1} shows that the naive procedure does not control the Type I error rate.

\begin{example}\label{exm:nuissance}	
Suppose $Y = (Y_1, Y_2)$ is generated as
\begin{align*}
Y_1 &\gets \varepsilon_1, &   \varepsilon_1 \sim \rm{gamma}(1, 1) - 1, \\
Y_2 &\gets .5 \times Y_1 + \varepsilon_2, &  \varepsilon_2 \sim \rm{gamma}(1, 1) - 1,
\end{align*}
so that the true graph is $Y_1 \rightarrow Y_2$. We consider two competing hypotheses which suppose $H_0: Y_p \rightarrow Y_c$ for $(p = 1, c = 2)$ or $(p = 2, c = 1)$, and we consider two testing approaches. 

\vspace{.2cm}
\noindent \textbf{Direct approach}: Regress $\mathbf{Y}_c$ onto $\mathbf{Y}_p$ to form the residuals $\bm{\hat \eta}_{c\setminus p}$, and subsequently test $\bm{\hat \eta}_{c\setminus p} \independent \bm{Y}_p$ using dHSIC~\citep{pfister2018kernel} or $\tau^\star$~\citep{bergsma2014consistent}. 

\vspace{.2cm}
\noindent \textbf{Sample splitting}: Split the data into a training and test set. Using the training set, regress $\mathbf{Y}^{(train)}_c$ onto $\mathbf{Y}^{(train)}_p$ and estimate $\hat \beta_{c,p}$. Using the test set, form the residuals $\bm{\hat \eta}_{c\setminus p} = \mathbf{Y}^{(test)}_p - \hat \beta_{c,p} \mathbf{Y}^{(test)}_c$, and subsequently test $\bm{\hat \eta}_{c \setminus p} \independent \bm{Y}^{(test)}_p$.

\vspace{.3cm}
\noindent

Table~\ref{tab:ex1} shows the proportion of hypothesis tests that are rejected when the null is true ($H_0: Y_1 \rightarrow Y_2$) and when the null is false ($H_0: Y_2 \rightarrow Y_1$) for sample sizes of $n = 100, 1000$. We see that the naive tests with dHSIC or $\tau^\star$ (both the direct approach or sample splitting) do not control the Type I error rate. Under the direct approach, $\tau^\star$ has a Type I error rate close to the nominal rate when $n = 100$, but when $n = 1000$ and the test has more power, the Type I error rate is not controlled.

\end{example}

\begin{table}[ht] \footnotesize
\centering
\caption{\label{tab:ex1}\textbf{Cols 2-6}: empirical size of a level $\alpha = .1$ test. \textbf{Cols 7-12}: power of a level $\alpha = .1$ test. `Direct' indicates the direct approach, `Split' indicates the sample splitting approach, and `Prop' indicates the proposed approach. `D' indicates using the dHSIC statistic, $\tau$ indicates using the $\tau^\star$ statistic.}
\begin{tabular}{|r|cc|cc|c|cc|cc|c|}
  \hline
    & \multicolumn{5}{|c|}{$H_0: Y_1 \rightarrow Y_2$} &   \multicolumn{5}{c|}{$H_0: Y_2 \rightarrow Y_1$}\\ \hline 
  & \multicolumn{2}{|c|}{Direct} & \multicolumn{2}{c|}{Split} & Prop &   \multicolumn{2}{c}{Direct} & \multicolumn{2}{|c|}{Split} & Prop\\ \hline
 & D & $\tau$ & D & $\tau$ &  & D & $\tau$ & D & $\tau$ & \\ 
  \hline
$n=100$ & 0.15 & 0.11 & 0.23 & 0.32 & 0.10 & 1.00 & 1.00 & 0.97 & 0.94 & 0.64 \\ 
  $n=1000$ & 0.18 & 0.20 & 0.24 & 0.34 & 0.09 & 1.00 & 1.00 & 1.00 & 1.00 & 0.97 \\ 
   \hline
\end{tabular}
\end{table}

\newpage

\section{Calibrating the test with a limiting Gaussian}
\label{sec:appendix-asymptotic-calibrate}
We show that the asymptotic normal distribution of~\eqref{eq:oracleStat} may provide a poor approximation when $\varepsilon_v$ is far from Gaussian. 
Specifically, we consider log-normal data with $n = 500$ and $p = 5, 10, \ldots, 25$ and record the Type I error rate for a nominally $\alpha = .05$ test calibrated using the limiting Gaussian with a plug-in estimate of the variance. We see that this test performs quite poorly when compared to the tests calibrated by the oracle and proposed residual bootstrap distributions. 

\begin{table}[thb] 
\centering\footnotesize
\caption{\label{tab:whyNotAsymp}The empirical size for 2000 replications of a nominally $\alpha = .05$ tests with $n = 500$ when calibrating with the asymptotic distribution, the oracle distribution, and the proposed residual bootstrap distribution. All values are multiplied by 1000. The confidence interval with 2 standard errors is given in the parenthesis.}
\begin{tabular}{|l|c|c|c|c|c|}
  \hline
  & $p= 5$ & 10 & 15 & 20 & 25 \\ 
  \hline
Asymp & 72 (60, 83) & 84 (71, 96) & 97 (84, 110) & 107 (93, 121) & 126 (111, 141) \\ 
  Oracle & 52 (42, 62) & 52 (43, 62) & 54 (44, 64) & 44 (35, 53) & 55 (45, 65) \\ 
  Proposed & 52 (42, 61) & 54 (43, 64) & 52 (43, 62) & 55 (45, 65) & 64 (53, 74) \\ 
   \hline
\end{tabular}
\end{table}

We generated data for the simulation in the following way. Let $X \in \mathbb{R}^p$. We draw $X = \exp(Z)$ from a log-normal distribution where $Z \sim N(0, \Sigma)$ with $\Sigma_{v,v} = 1$ and all off-diagonals $\Sigma_{u,v} = .2$. Let $\bm{X}$ denote the matrix where each row denotes an i.i.d observation of $X$ with $\bm{X}$ scaled and centered so each column has mean $0$ and variance $1$. Furthermore, let $\bm{\tilde X}$ denote the matrix $\bm{X}$ augmented with a column of $1$s for the intercept. 
We use a single test function $h(X_v) = X_v^2$, and let $\bm{H}$ be the matrix where the $u$th column corresponds to a scaled and centered version of $h(X_v) = X_v^2$.

We draw  $\bm{Y} = \bm{X} \beta + \bm{\varepsilon}$ where $\beta_v = 1$ for all $v$ and $\bm{\varepsilon} = (\varepsilon_i : i \in [n])$ where $\varepsilon_i = \exp(w_i) - \E(\exp(w_i))$ is also log-normal with $w_i \sim N(0, \log(1 + \sqrt(5)) - \log(2))$ so that $\varepsilon_i$ has mean $0$ and variance $1$. We regress $Y$ onto $\tilde X$ (including an intercept) and let $\bm{\hat \eta}$ denote the resulting residuals. We also calculate, $\hat \sigma^2 = \frac{\Vert \bm{\hat \eta} \Vert^2}{n-p}$, an unbiased estimate of the variance of $\varepsilon$

For the test calibrated by the asymptotic distribution, we then compare the test statistic $T_2 = \Vert \frac{1}{n}\bm{H}^T \bm{\hat \eta}\Vert_2^2$ to the distribution of
\begin{equation}
  \left \Vert \frac{1}{\sqrt{n-p}}\bm{H}^T \left(I - \bm{\tilde X}(\bm{\tilde X}^T \bm{\tilde X})^{-1}\bm{\tilde X}^T \right)  \bm{\check \varepsilon}_v \right\Vert_2^2
\end{equation}
where each element of $\bm{\check \varepsilon}_v$ is drawn from $N(0, \hat \sigma^2)$. Note that we divide by $\sqrt{n-p}$ instead of $n$ in the null distribution which should help lower the Type I error rate in finite samples. Indeed, this correction improves performance; however, even with this additional correction the Type I error is inflated.

For the test calibrated by the oracle distribution, we then compare the test statistic $T_2 = \Vert \frac{1}{n}\bm{H}^T \bm{\hat \eta}\Vert_2^2$ to the distribution of
\begin{equation}
  \left \Vert \frac{1}{\sqrt{n-p}}\bm{H}^T \left(I - \bm{\tilde X}(\bm{\tilde X}^T \bm{\tilde X})^{-1}\bm{\tilde X}^T \right)  \bm{\check \varepsilon}_v \right\Vert_2^2
\end{equation}
where each element of $\bm{\check \varepsilon}_v$ is drawn from $\exp(w_i) - \E(\exp(w_i))$ with $w_i \sim N(0, \log(1 + \sqrt(5)) - \log(2))$. 

This entire procedure is replicated 2000 times with $n = 500$ and $p = 5, 10, \ldots, 25$. In Figure~\ref{fig:whyAsympBad}, we show the distribution of the resulting p-values for each procedure in each setting.

\begin{figure}[h]
    \centering
    \includegraphics[scale = .8]{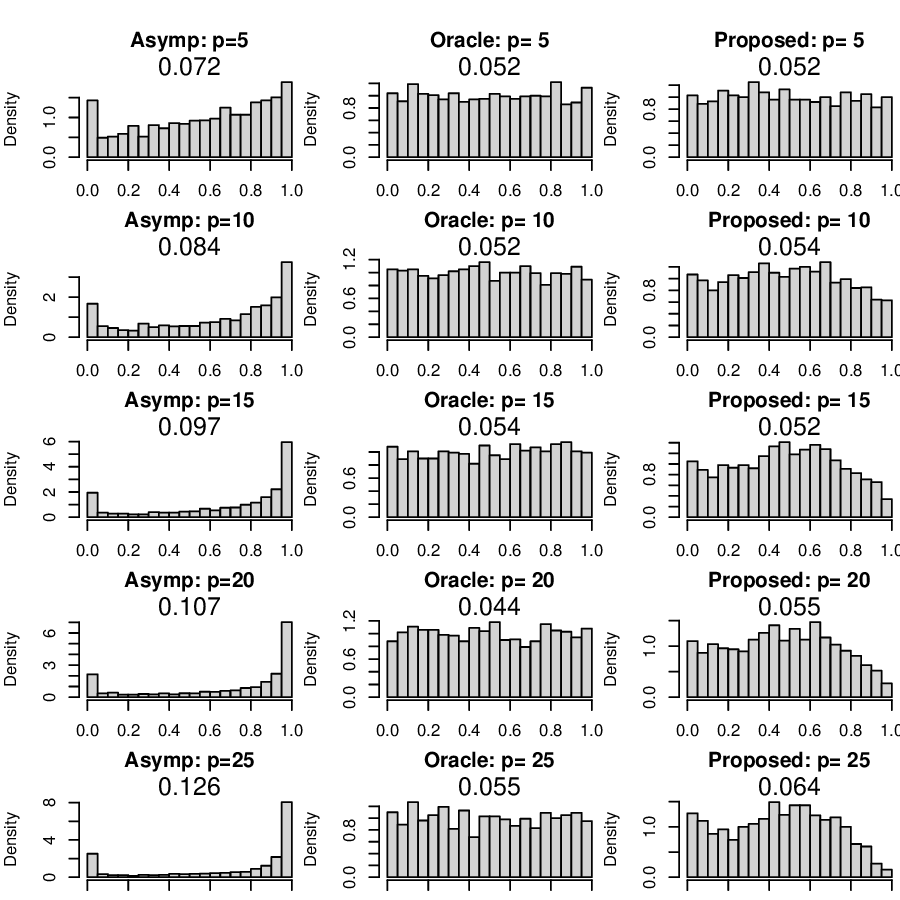}
    \caption{\label{fig:whyAsympBad}Comparison of p-values from tests calibrated by the estimated asymptotic normal distribution, the oracle distribution, and the proposed residual bootstrap distribution. Results show $2000$ replications with $n = 500$ and various values of $p$. The empirical size for nominally $\alpha = .05$ tests are displayed under each title.}
\end{figure}

\FloatBarrier
\clearpage

\section{Additional Simulation Results}
We examine power for various other test functions using the setup from Section~\ref{sec:numerics}. Recall that for each replication, we construct a graph by starting with edges $v \rightarrow v+1$ for all $v < p$; for any $u < v - 1$, $u \rightarrow v$ is added with probability $1/2$. For each edge, we sample a linear coefficient uniformly from $\pm (.1, .95)$. We consider settings where all error terms are either uniform, lognormal, gamma, Weibull, or Laplace random variables and a setting---called mixed---where the distribution of each variable in the SEM is randomly selected. We set $p = 10, 20, 45$ and $n \approx p^{5/4}$ or $p^2$. The data is standardized before applying the goodness-of-fit tests.  For each setting of $p, n$, and error distribution, we complete 500 replications. 

We consider the following set of test functions
\begin{itemize}
    \item Trigonometric ({\bf T}): $\sin(y), \cos(y), \sin(2y), \cos(2y)$
    \item Random trigonometric functions ({\bf R}): $\bigcup_{j=1}^{10} \{\sin(\omega_j y), \cos(\omega_j y)\}$ where $\omega_j \sim N(0,1)$.
    \item Trigonometric products ({\bf Z}): $\tanh(y), \tanh(y)y, \tanh(y)y^2, \tanh(y)\sin(y), \tanh(y)\cos(y)$ where each function is standardized to have mean $0$ and variance $1$.
    \item Moments centered ({\bf MC}): $y^2, y^3, \text{sign}(y)\times \vert y\vert^{2.5}$ where the each function is centered to have mean $0$.
    \item Moments Standardized ({\bf MS}): $y^2, y^3, \text{sign}(y)\times \vert y\vert^{2.5}$ where the each function is standardized to have mean $0$ and variance $1$.
    \item Combined functions {\bf TM, RM, ZM} which are the union of {\bf T, R, Z} and {\bf MS}.
\end{itemize}

In Table~\ref{tab:altTestFunc}, we see that {\bf TM} generally performs well across a wide range of simulation settings. However, it suffers when the errors are symmetric. T, R and Z tend to do well when $n = p^2$, particularly when the errors are uniform. However, they perform poorly when $n = p^{5/4}$. 

\begin{table}
\centering
\caption{\label{tab:altTestFunc}Empirical size and power of $\alpha = .1$ tests. Size: bolded values exceed the nominal $\alpha = .1$ by 2 standard deviations. Power: bolded values indicate the procedure has the largest power (or is within 2 standard deviations) for that particular setting. Each proportion in the table has been multiplied by 100. If the procedure's empirical size is significantly above the nominal level, the empirical power is not displayed.}
\scriptsize
\begin{tabular}{|c|cc|cccccccc|cccccccc|}
 \hline
& &  & \multicolumn{8}{|c|}{Size} & \multicolumn{8}{c|}{Power} \\
 \hline
 &  & p & TM & RM & ZM & MC & MS & T & R & Z & TM & RM & ZM & MC & MS & T & R & Z \\ 
  \hline
\hline\multirow{18}{*}{\rotatebox[origin=c]{90}{$n \approx p^{5/4}$}} & \multirow{3}{*}{\rotatebox[origin=c]{0}{\bftab gamma}} & 10 & 4 & 4 & 4 & 9 & 6 & 4 & 3 & 4 & \bftab 14 & \bftab 15 & 9 & \bftab 15 & \bftab 17 & 11 & 11 & 9 \\ 
   &  & 20 & 7 & 8 & 7 & 11 & 8 & 6 & 8 & 7 & \bftab 26 & \bftab 27 & 18 & 22 & \bftab 28 & 14 & 13 & 18 \\ 
   &  & 45 & 9 & 9 & 7 & 10 & 10 & 6 & 6 & 7 & \bftab 44 & \bftab 44 & 30 & 23 & \bftab 46 & 21 & 23 & 28 \\ 
  \cline{2-19} & \multirow{3}{*}{\rotatebox[origin=c]{0}{\bftab laplace}} & 10 & 5 & 5 & 4 & 8 & 6 & 4 & 4 & 4 & 7 & 8 & 3 & \bftab 14 & 9 & 5 & 5 & 4 \\ 
   &  & 20 & 8 & 8 & 7 & 9 & 8 & 6 & 8 & 6 & 11 & 10 & 9 & \bftab 15 & 11 & 8 & 9 & 9 \\ 
   &  & 45 & 11 & 11 & 9 & 8 & 11 & 9 & 9 & 9 & \bftab 12 & \bftab 13 & 9 & \bftab 14 & \bftab 13 & 8 & 8 & 9 \\ 
  \cline{2-19} & \multirow{3}{*}{\rotatebox[origin=c]{0}{\bftab lognormal}} & 10 & 8 & 8 & 5 & 9 & 11 & 5 & 6 & 5 & \bftab 21 & 21 & 14 & \bftab 23 & \bftab 25 & 15 & 15 & 15 \\ 
   &  & 20 & 10 & 10 & 9 & \bftab 14 & 12 & 8 & 7 & 9 & 44 & 43 & 29 &  & \bftab 49 & 23 & 22 & 27 \\ 
   &  & 45 & 12 & 12 & 13 & 11 & \bftab 14 & 6 & 9 & 11 & \bftab 71 & \bftab 71 & 50 & 35 &  & 34 & 35 & 49 \\ 
  \cline{2-19} & \multirow{3}{*}{\rotatebox[origin=c]{0}{\bftab mixed}} & 10 & 8 & 7 & 5 & 10 & 8 & 7 & 6 & 6 & 11 & 12 & 9 & \bftab 17 & \bftab 13 & 8 & 8 & 10 \\ 
   &  & 20 & 9 & 8 & 6 & 9 & 9 & 6 & 7 & 5 & \bftab 25 & \bftab 24 & 18 & \bftab 24 & \bftab 27 & 14 & 16 & 17 \\ 
   &  & 45 & 10 & 9 & 8 & 10 & 10 & 7 & 8 & 7 & \bftab 39 & \bftab 37 & 30 & 29 & \bftab 39 & 18 & 16 & 28 \\ 
  \cline{2-19} & \multirow{3}{*}{\rotatebox[origin=c]{0}{\bftab uniform}} & 10 & 5 & 6 & 3 & 6 & 6 & 5 & 4 & 3 & \bftab 5 & \bftab 5 & 2 & \bftab 4 & \bftab 5 & \bftab 5 & \bftab 5 & 3 \\ 
   &  & 20 & 8 & 9 & 7 & 7 & 9 & 7 & 8 & 6 & \bftab 8 & 7 & 5 & \bftab 10 & \bftab 7 & 6 & 6 & 5 \\ 
   &  & 45 & 7 & 7 & 7 & 8 & 8 & 9 & 7 & 7 & 8 & 7 & \bftab 10 & 6 & 7 & \bftab 9 & \bftab 11 & \bftab 10 \\ 
  \cline{2-19} & \multirow{3}{*}{\rotatebox[origin=c]{0}{\bftab weibull}} & 10 & 5 & 5 & 4 & 9 & 7 & 5 & 6 & 4 & \bftab 21 & \bftab 20 & 13 & \bftab 20 & \bftab 24 & 13 & 15 & 13 \\ 
   &  & 20 & 10 & 11 & 7 & 11 & 13 & 7 & 7 & 8 & \bftab 43 & \bftab 42 & 28 & 24 & \bftab 46 & 24 & 23 & 28 \\ 
   &  & 45 & 12 & 11 & 10 & 11 & 12 & 9 & 9 & 9 & \bftab 66 & \bftab 63 & 46 & 26 & \bftab 67 & 29 & 29 & 44 \\ 
   \hline
   \hline\multirow{18}{*}{\rotatebox[origin=c]{90}{$n \approx p^{2}$}} & \multirow{3}{*}{\rotatebox[origin=c]{0}{\bftab gamma}} & 10 & 10 & 10 & 10 & 11 & 10 & 10 & 11 & 9 & \bftab 88 & \bftab 88 & 83 & 42 & \bftab 88 & 70 & 70 & 81 \\ 
   &  & 20 & 9 & 10 & 10 & 10 & 9 & 10 & 9 & 10 & \bftab 99 & \bftab 99 & 98 & 62 & \bftab 99 & 93 & 92 & \bftab 98 \\ 
   &  & 45 & 10 & 10 & 11 & 9 & 10 & 9 & 9 & 11 & \bftab 100 & \bftab 100 & \bftab 100 & 86 & \bftab 100 & 99 & 99 & \bftab 100 \\ 
  \cline{2-19} & \multirow{3}{*}{\rotatebox[origin=c]{0}{\bftab laplace}} & 10 & 11 & 11 & 11 & 9 & 11 & 11 & 12 & 12 & \bftab 23 & \bftab 24 & \bftab 24 & \bftab 28 & \bftab 24 & 13 & 15 & 21 \\ 
   &  & 20 & 10 & 11 & 11 & 10 & 10 & 10 & 9 & 11 & 36 & 35 & 36 & \bftab 43 & 35 & 20 & 22 & 30 \\ 
   &  & 45 & 9 & 10 & 8 & 8 & 11 & 10 & 7 & 8 & 41 & 41 & 54 & \bftab 71 & 40 & 42 & 39 & 53 \\ 
  \cline{2-19} & \multirow{3}{*}{\rotatebox[origin=c]{0}{\bftab lognormal}} & 10 & 12 & 11 & 9 & 11 & 13 & 9 & 9 & 10 & \bftab 96 & \bftab 96 & 94 & 66 & \bftab 96 & 90 & 89 & 92 \\ 
   &  & 20 & 8 & 8 & 10 & 7 & 8 & 9 & 10 & 9 & \bftab 100 & \bftab 100 & \bftab 100 & 79 & \bftab 100 & 99 & 99 & \bftab 100 \\ 
   &  & 45 & 11 & 10 & 9 & 8 & 11 & 10 & 9 & 10 & \bftab 100 & \bftab 100 & \bftab 100 & 86 & \bftab 100 & \bftab 100 & \bftab 100 & \bftab 100 \\ 
  \cline{2-19} & \multirow{3}{*}{\rotatebox[origin=c]{0}{\bftab mixed}} & 10 & 8 & 8 & 6 & 9 & 9 & 7 & 10 & 7 & \bftab 80 & \bftab 79 & \bftab 77 & 53 & \bftab 78 & 65 & 65 & 75 \\ 
   &  & 20 & 11 & 10 & 8 & 11 & 10 & 7 & 8 & 7 & \bftab 97 & \bftab 98 & \bftab 98 & 70 & \bftab 97 & 91 & 92 & \bftab 97 \\ 
   &  & 45 & 10 & 9 & 11 & 10 & 10 & 10 & 9 & 9 & \bftab 100 & \bftab 100 & \bftab 100 & 78 & \bftab 100 & 99 & 99 & \bftab 100 \\ 
  \cline{2-19} & \multirow{3}{*}{\rotatebox[origin=c]{0}{\bftab uniform}} & 10 & 11 & 11 & 9 & 11 & 11 & 10 & 8 & 9 & 5 & 5 & \bftab 12 & 9 & 6 & \bftab 13 & \bftab 11 & \bftab 12 \\ 
   &  & 20 & 9 & 9 & 10 & 11 & 9 & 8 & 10 & 9 & 8 & 8 & \bftab 29 & 17 & 6 & \bftab 32 & 25 & \bftab 31 \\ 
   &  & 45 & 12 & 11 & 10 & 11 & 11 & 9 & 11 & 11 & 35 & 28 & \bftab 67 & 48 & 14 & \bftab 71 & 56 & \bftab 69 \\ 
  \cline{2-19} & \multirow{3}{*}{\rotatebox[origin=c]{0}{\bftab weibull}} & 10 & 11 & 11 & 7 & 11 & 12 & 8 & 7 & 7 & \bftab 95 & \bftab 95 & \bftab 94 & 53 & \bftab 95 & 86 & 87 & 92 \\ 
   &  & 20 & 9 & 9 & 10 & 9 & 9 & 9 & 8 & 10 & \bftab 100 & \bftab 100 & \bftab 100 & 73 & \bftab 100 & 97 & 98 & \bftab 100 \\ 
   &  & 45 & 8 & 9 & 8 & 9 & 9 & 8 & 9 & 7 & \bftab 100 & \bftab 100 & \bftab 100 & 87 & \bftab 100 & \bftab 100 & \bftab 100 & \bftab 100 \\ 
   \hline
\end{tabular}
\end{table}

\FloatBarrier
\clearpage

\section{Proofs for Section~\ref{sec:testingOrd}}

\subsection*{Lemma~\ref{lem:indepPvals}}
If $\theta \in \Theta(G)$, then $(\gamma_{\theta, v} \; : \; \theta(v) > 2)$---the p-values calculated for each level using the oracle procedure---are mutually independent. 
\begin{proof}
Without loss of generality, let $\theta = (1, \ldots, p)$ and let $\bm{\gamma} = (\gamma_{\theta, 2}, \ldots, \gamma_{\theta,p})$. 
First, note that because the oracle distribution provides an exact test for any finite sample $\gamma_{\theta,p}\mid \mathbf{Y}_{\pr(p)} \sim U(0,1)$ for any fixed $\mathbf{Y}_{\pr(p)}$. Thus, we have 
\[E_{\varepsilon_p}\left( \exp(it_p \gamma_{\theta,p}) \mid \mathbf{Y}_{\pr(p)}\right) = \phi_{U(0,1)}(t_p),\] where $\phi_{U(0,1)}(t_p)$ is the characteristic function of $U(0,1)$. 
Then, for any $t\in \mathbb{R}^{p-1}$, the characteristic function of $\bm{\gamma}$ is
\begin{equation}\begin{aligned}
    \E_\mathbf{Y}\left( \exp(i t \bm{\gamma})\right) &= \E_{\mathbf{Y}_{\pr(p)}}\left\{ \E_{Y_p}\left( \exp(it_p \gamma_{\theta,p}) \mid \vecY_{\pr(p)} \right) \exp(it_{\pr(p)} \bm{\gamma}_{\pr(p)}) \right\} \\
        &= \E_{\mathbf{Y}_{\pr(p)}}\left\{ \E_{\varepsilon_p}\left( \exp(it_p \gamma_{\theta,p}) \mid \vecY_{\pr(p)} \right) \exp(it_{\pr(p)} \bm{\gamma}_{\pr(p)}) \right\} \\
    &= \phi_{U(0,1)}(t_p)\E_{\mathbf{Y}_{\pr(p)}}\left\{ \exp\left(it_{\pr(p)} \bm{\gamma}_{\pr(p)}\right)  \right\}.
\end{aligned}    
\end{equation}
Inductively applying the same argument for $p-1, \ldots, 2$ yields  
\begin{equation}\begin{aligned}
    \E\left( \exp(i t \bm{\gamma})\right) &= \prod_{v = 2}^p \phi_{U(0,1)}(t_v),
\end{aligned}    
\end{equation}
which shows that the p-values are independent.
\end{proof}

\clearpage
\section{Proofs for Section~\ref{subsec:ci_for_causal_effects}}

\subsection*{Lemma~\ref{lem:ci-modelUncertain}}

Let $\pi_{u,v}$ denote the total causal effect of $v$ onto $u$. Suppose $\hat \Theta(\mathbf{Y}, \alpha/2)$ satisfies~\eqref{eq:topConfSetSingle}, and $C(S)$ is an asymptotically valid $1-\alpha/2$ confidence interval for the parameter of interest, conditional on $S$ being a valid adjustment set. Then, for the confidence interval produced by Alg.~\ref{alg:TotalEffect},
$
\lim_{n \rightarrow \infty} P(\pi_{u,v} \in \hat C_\alpha) \geq 1 - \alpha.
$
\begin{proof}
If $u \in \an(v)$, then for every $\theta \in \Theta(G)$, $\theta(u) < \theta(v)$. 
Fix an arbitrary $\theta \in \Theta(G)$. Let $S_\theta$ denote the appropriate adjustment set for the effect of interest given $\theta$ and let $C(S_\theta)$ denote the $1-\alpha/2$ confidence interval for the effect of interest when using the adjustment set $S_\theta$. Then,
\begin{equation}
\lim_{n \rightarrow \infty} P(\pi_{u,v} \in \hat C_\alpha) \geq 1 - \lim_{n \rightarrow \infty}P\left(\theta \not \in \hat \Theta(\mathbf{Y}, \alpha / 2) \cup \pi_{u,v} \not \in C(S_\theta)\right) \geq 1 - (\alpha / 2 + \alpha/2).
\end{equation}
If $u \not \in \an(v)$, then $\pi_{v,u} = 0$ and there exists a $\theta \in \Theta(G)$ such that $\theta(v) < \theta(u)$. Then,
\begin{equation}
\lim_{n \rightarrow \infty} P(\pi_{u,v} \in \hat C_\alpha) \geq 1 - \lim_{n \rightarrow \infty}P(\theta \not \in \hat \Theta(\mathbf{Y}, \alpha / 2) \geq 1 - \alpha / 2,
\end{equation}
which completes the proof.
\end{proof}
\subsection*{Lemma~\ref{lem:ancestralSet}}
Suppose $\hat \Theta(\mathbf{Y}, \alpha)$ satisfies~\eqref{eq:topConfSetSingle}. Then,
$
\lim_{n \rightarrow \infty} P(\mathcal{\hat A}_\cap 
\subseteq \mathcal{A}\subseteq \mathcal{\hat A}_\cup) \geq 1 - 2\alpha. 
$

\begin{proof}
Suppose $\Theta(G) = \{\theta_1\}$ so that $|\Theta(G)| = 1$. Then, by definition, $\theta_1(v) < \theta_1(u)$ for every  $(u,v) \not \in \mathcal{A}$. Thus, $\theta_1 \in \hat \Theta(\mathbf{Y}, \alpha)$ implies that $\mathcal{A}_\cap \subseteq \mathcal{A}$. Similarly, since  $\theta_1(u) < \theta_1(v)$ for every $(u,v) \in \mathcal{A}$ when $\theta_1 \in \hat \Theta(\mathbf{Y}, \alpha)$, then $\theta_1 \in \hat \Theta(\mathbf{Y}, \alpha)$ also implies that
$\mathcal{A}\subseteq \mathcal{\hat A}_\cup$. If $\hat \Theta(\mathbf{Y}, \alpha)$ satisfies Eq.~\eqref{eq:topConfSetSingle}, then $\theta_1 \in \hat \Theta(\mathbf{Y}, \alpha)$ occurs with probability bounded below by $1-\alpha$ as $n\rightarrow \infty$ so that 
\begin{equation}
\lim_{n \rightarrow \infty} P(\mathcal{\hat A}_\cap 
\subseteq \mathcal{A}\subseteq \mathcal{\hat A}_\cup) \geq 1 - \alpha. 
\end{equation}

Now consider the case where $\vert \Theta(G) \vert > 1$. There exist a pair $\theta_1, \theta_2 \in \Theta(G)$ such that: 
\begin{enumerate}
    \item For every $(u,v) \in \mathcal{A}$, we have $\theta_1(u) < \theta_1(v)$ and $\theta_2(u) < \theta_2(v)$
    \item For every $(u,v) \not \in \mathcal{A}$, we have  $\theta_1(u) < \theta_1(v)$ and $\theta_2(v) < \theta_2(u)$.
\end{enumerate}
Then, for every $(u,v) \not \in \mathcal{A}$ either $\theta_1(v) < \theta_1(u)$ or $\theta_2(v) < \theta_2(u)$. Thus, the event $\{\theta_1\in  \hat \Theta(\mathbf{Y}, \alpha)\} \cap \{\theta_2\in  \hat \Theta(\mathbf{Y}, \alpha)\}$ implies that $\mathcal{\hat A}_\cap \subseteq \mathcal{A}$. 
Furthermore, for every $(u,v) \in \mathcal{A}$, $\theta_1(v) < \theta_1(u)$. Thus, the event $\{\theta_1 \in  \hat \Theta(\mathbf{Y}, \alpha)\}$ also implies $\mathcal{A}\subseteq \mathcal{\hat A}_\cup$. If $\hat \Theta(\mathbf{Y}, \alpha)$ satisfies Eq.~\eqref{eq:topConfSetSingle}, then $\{\theta_1\in  \hat \Theta(\mathbf{Y}, \alpha)\} \cap \{\theta_2\in  \hat \Theta(\mathbf{Y}, \alpha)\}$ occurs with probability bounded below by $1-2\alpha$ as $n\rightarrow \infty$ so that 
\begin{equation}
\lim_{n \rightarrow \infty} P(\mathcal{\hat A}_\cap 
\subseteq \mathcal{A}\subseteq \mathcal{\hat A}_\cup) \geq 1 - 2\alpha. 
\end{equation}

\end{proof}

\clearpage
\section{Proofs for Section~\ref{sec:theory}}
We first give proofs of the results stated in the main manuscript. Several supporting lemmas are proven later.

\subsection{Proof of Theorem~\ref{thm:highDCLT}}

Suppose Assumptions~\ref{asm:subExpErrors}, \ref{asm:designCondition}, \ref{asm:bias}, and \ref{asm:testFuncWellCondition} hold, and $pK/n \rightarrow 0$. 

If $\bias_v = 0$ for all $v$ and for some universal constant $C$ we have $C(\overline{\sigma}^2 Kp/n + \log(n)/\sqrt{n}) < \underline{\sigma}^2/2$, then with probability $1-o(1)$
\begin{equation}
\begin{aligned}
\max_v \vert \gamma_{\theta, v} - \hat \gamma_{\theta,v} \vert &\lesssim     \left(n^{-1/2} + \frac{Kp}{n}\right)\frac{h_{\max,1}^4M^4 \log^{11/2}(n) \overline{\sigma}^2}{\lambda_{\min,Z}\lambda_{\min,C}\underline{\sigma}^2}\left(1 + \log\left[\frac{\overline{\sigma}^2
}{ \underline{\sigma}^2 \lambda_{\min,C}}\right]\right).
\end{aligned}
\end{equation} 

If $\bias_v \neq 0$, $\max_v \vert \nu^{(v)}\vert_\infty < \delta_1$, and for some universal constant $C$, we have $C((\overline{\sigma}^2  +\lambda_{\min,Z}^{-1})Kp\log^2(n)/n + \log(n)/\sqrt{n} + d^\star_2) < \underline{\sigma}^2/2$,
then with probability $1-o(1)$
\begin{equation}\footnotesize
\begin{aligned}
\max_v \vert \gamma_{\theta, v} - \hat \gamma_{\theta,v} \vert &\lesssim     
    \delta_1\sqrt{\log(pJ)}  + d^\star_2\left( \frac{\log^2(n)(\overline{\sigma}^2 + \lambda_{\min, Z}^{-1})}{\lambda_{\min,C} \underline{\sigma}^2}\right)\left(1 + \log\left[\frac{\overline{\sigma}^2 + \lambda_{\min, Z}^{-1}}{\lambda_{\min,C} \underline{\sigma}^2}\right] \right)\\
    &\quad + 
    \left(n^{-1/2} + \frac{Kp}{n}\right) \left(\frac{ h_{\max,1}^4 M^4 \log^{11/2}(n) (\overline{\sigma}^2 + \lambda_{\min, Z}^{-1})}{\lambda_{\min,C} \underline{\sigma}^2}\right)\left(1 + \log\left[\frac{\overline{\sigma}^2 + \lambda_{\min, Z}^{-1}}{\lambda_{\min,C} \underline{\sigma}^2}\right] \right). 
\end{aligned}
\end{equation}
\begin{proof}
The main idea of the proof is to use the high-dimensional CLT results of \citet{chernozhukov2023high} to show that the measure over rectangles of both $\tau^{(v)}$ and $\tilde \tau^{(v)}$ can be well approximated by multivariate normals which are also close to each other. 

Let $\mathcal{R}_v$ denote the set of all rectangles in $\mathbb{R}^{|\pr(v)|J}$ and let $R_v(t) \in \mathcal{R}_v$ denote the rectangle $[-t, t]^{\vert \pr(v)\vert J}$. Furthermore, let $\xi^{(v)} \sim N(0, \sigma_v^2\Sigma^{(v)})$ and $\tilde \xi^{(v)} \sim N(0, \tilde \sigma_v^2\Sigma^{(v)})$. Throughout the proof, we will make statements about $P(T^{(v)}_\infty \leq t \mid \vecU)$ and $P(\tilde T^{(v)}_\infty \leq t \mid \vecU, \vecepsv)$ but for notational brevity we will drop the conditioning event and write $P(T^{(v)}_\infty \leq t)$ and $P(\tilde T^{(v)}_\infty \leq t)$ where the conditioning event is implied but not explicitly written. Then
\begin{equation}
\begin{aligned}\label{eq:maxStat3parts}
    \max_v \sup_{t}\vert P(T^{(v)}_\infty \leq t)  - P(\tilde T^{(v)}_\infty \leq t) \vert & = \max_v \sup_{t}\vert P(\tau^{(v)}  \in R(t))  - P(\tilde \tau^{(v)} \in R(t)) \vert\\
    & \leq \max_v \sup_{t}\vert P(\tau^{(v)} \in R(t))  - P(\xi^{(v)} \in R(t)) \vert\\
    &\quad + \max_v \sup_{R \in \mathcal{R}_v}\vert P(\tilde \tau^{(v)}  \in R)  - P(\tilde \xi^{(v)} \in R) \vert\\
      &\quad + \max_v \sup_{R \in \mathcal{R}_v}\vert P(\xi^{(v)} \in R)   - P(\tilde \xi^{(v)} \in R) \vert\\
    &= I_1 + I_2 + I_3.
\end{aligned}
\end{equation}

We first bound $I_1$. Recall that $\tau^{(v)}= \kappa^{(v)} + \nu^{(v)}$ where $\nu^{(v)} = \frac{1}{\sqrt{n}}\sum_{i} \zeta_i^{(v)}d_{v,i}$ will generally not have mean $0$ in the misspecified setting. Thus, to bound $I_1$ we follow the analysis in~\citet{chernozhukov2023high} using Nazarov's Inequality~\citep{nazarov2003maximal}---though in our setting we have already conditioned on $\vert \nu^{(v)}\vert < \delta_1$. For any $t$, we have
\begin{equation}
\begin{aligned}
 P(\tau^{(v)}  \in R(t))  &=  P(\vert \tau^{(v)}\vert_\infty \leq t) \leq P(\vert \kappa^{(v)}\vert_\infty \leq t + \delta_1)\\   
  &\leq P(\vert \xi^{(v)}\vert_\infty \leq t + \delta_1) + \left \vert P(\vert \kappa^{(v)}\vert_\infty \leq t + \delta_1) - P(\vert \xi^{(v)}\vert_\infty \leq t + \delta_1)\right\vert \\ 
    &\leq P(\vert \xi^{(v)}\vert_\infty \leq t ) + \left\vert P(\vert \xi^{(v)}\vert_\infty \leq t + \delta_1) - P(\vert \xi^{(v)}\vert_\infty \leq t) \right \vert\\
    &\quad \left \vert P(\vert \kappa^{(v)}\vert_\infty \leq t + \delta_1) - P(\vert \xi^{(v)}\vert_\infty \leq t + \delta_1)\right\vert \\ 
    &\leq P(\vert \xi^{(v)}\vert_\infty \leq t ) + \frac{\delta_1}{\underline{\sigma}^2} \sqrt{\log(pJ)} 
    + \left \vert P(\vert \kappa^{(v)}\vert_\infty \leq t + \delta_1) - P(\vert \xi^{(v)}\vert_\infty \leq t + \delta_1)\right\vert. \\   
\end{aligned}
\end{equation}
Similarly,
\begin{equation}
\begin{aligned}
 P(\tau^{(v)}  \in R(t))  & \geq P(\vert \kappa^{(v)}\vert_\infty \leq t - \delta_1)\\   
  &\geq P(\vert \xi^{(v)}\vert_\infty \leq t - \delta_1) - \left \vert P(\vert \kappa^{(v)}\vert_\infty \leq t - \delta_1) - P(\vert \xi^{(v)}\vert_\infty \leq t - \delta_1)\right\vert \\ 
    &\geq P(\vert \xi^{(v)}\vert_\infty \leq t ) - \left\vert P(\vert \xi^{(v)}\vert_\infty \leq t - \delta_1) - P(\vert \xi^{(v)}\vert_\infty \leq t) \right \vert\\
    &\quad \left \vert P(\vert \kappa^{(v)}\vert_\infty \leq t - \delta_1) - P(\vert \xi^{(v)}\vert_\infty \leq t - \delta_1)\right\vert \\ 
    &\geq P(\vert \xi^{(v)}\vert_\infty \leq t ) - \frac{\delta_1}{\underline{\sigma}^2} \sqrt{\log(pJ)} 
    - \left \vert P(\vert \kappa^{(v)}\vert_\infty \leq t - \delta_1) - P(\vert \xi^{(v)}\vert_\infty \leq t - \delta_1)\right\vert. \\   
\end{aligned}
\end{equation}
Thus, we have
\begin{equation}
    I_1 \lesssim \frac{\delta_1}{\underline{\sigma}^2} \sqrt{\log(pJ)} + \underbrace{\max_v \sup_{R \in \mathcal{R}_v}\vert P(\kappa^{(v)}  \in R)  - P(\xi^{(v)} \in R) \vert}_{I_1'}.
\end{equation}
To bound $I_1'$, we apply Corollary 2.1 of \citet{chernozhukov2023nearly} which we state now for completeness. 

Suppose $W = \frac{1}{\sqrt{n}}\sum_{i} X_i$ where $X_i$ are independent mean $0$ vectors in $\mathbb{R}^d$ and $Z \sim N(0, \cov(W))$. Furthermore, let $\sigma^2_{\star, W}$ be the smallest eigenvalue of the correlation matrix of $W$ and $C$ be a universal constant.
\begin{enumerate}
    \item Suppose $\vert X_{ij} /\sqrt{\var(W_{j})} \vert \leq B_n$ for all $i = 1, \ldots, n$ and $j = 1, \ldots, d$ almost surely, then
    \begin{equation}
    \max_v \sup_{R \in \mathcal{R}_v}\vert P(W  \in R)  - P(\xi \in R) \vert \leq \frac{C B_n(\log(d)^{3/2} \log(n)}{\sqrt{n}\sigma^2_{\star, W}}
\end{equation}
\item Suppose for some $q \geq 4$ we have for all $i = 1, \ldots, n$ 
\[\E\left(\max_{j} \left\vert X_{ij} /\sqrt{\var(W_{j})}\right\vert^q\right)^{1/q} \leq B_n\]
and for all $j =, 1\ldots, d$ we have
\[n^{-1}\sum_{i}\E(\vert X_{ij} /\sqrt{\var(W_{j})}\vert^4) \leq B_n^2.\] Then, $\max_v \sup_{R \in \mathcal{R}_v}\vert P(W  \in R)  - P(\xi \in R) \vert $ is bounded above by
   \begin{equation}\footnotesize
C\left\{\frac{B_n \log^{3/2}(d)\log(n)}{\sqrt{n}\sigma^2_{\star,W}} + \frac{B_n^2\log^2(d)\log(n)}{n^{1-2/q}\sigma^2_{\star,W}} + \left(\frac{B_n^q(\log(d))^{3q/2 -4}\log(n)\log(dn)}{n^{q/2-1}\sigma^2_{\star,W}}\right)^{1/(q-2)} \right\}.
\end{equation}
\end{enumerate}

Marginally, $\zeta^{(v)}_{i} \not \independent \zeta^{(v)}_{i'}$; however, we are considering the distribution of $T^{(v)}$ conditional on $\vecY_{\pr(v)}$ so that each $\zeta^{(v)}_{i}$ is fixed and the only randomness in $\kappa^{(v)}$ is due to resampling $\varepsilon_{v,i}$. Thus, $\kappa^{(v)} \mid \vecY_{\pr(v)}$ is indeed the sum of independent centered vectors. 
By Assumption~\ref{asm:subExpErrors} and \ref{asm:testFuncWellCondition} we have for each $v$ and $i = 1, \ldots n$
\begin{equation}
    \left[\E(\max (\zeta_{i,j}^{(v)}\varepsilon_{v,i} / \sqrt{\Sigma_{jj}\sigma_v^2})^4)\right]^{1/4} \leq h_{\max_1}\left[\E( (\varepsilon_{v,i}/\sigma_v)^4)\right]^{1/4} \leq h_{\max_1} M.
\end{equation}
Furthermore, for each $v$ and $j = 1, \ldots \vert \pr(v)\vert J$
\begin{equation}
    n^{-1}\sum_i\E([\zeta_{i,j}^{(v)}\varepsilon_{v,i} / \sqrt{\Sigma_{jj}\sigma_v^2}]^4) \leq h_{\max_1}^4 \E( (\varepsilon_{v,i}/\sigma_v)^4) \leq h_{\max_1}^4 M^4
\end{equation}
and by construction $\kappa^{(v)}$ and $\xi^{(v)}$ have the same covariance matrix: $\sigma_v^2\Sigma^{(v)}$. Thus, we apply the second condition above from Corollary 2.1 of \citet{chernozhukov2023nearly} stated above with $d = pJ$, $q =4$ and $B_n = h_{\max,1}^2M^2$. This implies
\begin{equation} \footnotesize
\begin{aligned}
    I_1' &\leq C\left\{\frac{B_n \log^{3/2}(pJ)\log(n)}{\sqrt{n}\lambda_{\min,C}} + \frac{B_n^2\log^2(pJ)\log(n)}{n^{1-2/4}\lambda_{\min,C}} + \left(\frac{B_n^4(\log(pJ))^{3(4)/2 -4}\log(n)\log(pJn)}{n^{4/2-1}\lambda_{\min,C}^2}\right)^{1/(4-2)} \right\}\\
    &\lesssim C \frac{h_{\max,1}^2M^2 \log^{5/2}(n) + h_{\max,1}^4M^4\log^3(n) + h_{\max,1}^4M^4\log^2(n)  }{\sqrt{n}\lambda_{\min,C}}\\
       &\lesssim C \frac{h_{\max,1}^4M^4 \log^{3}(n) }{\sqrt{n}\lambda_{\min,C}}.
\end{aligned}
\end{equation}

To bound $I_2$, we again use Corollary 2.1 of \citet{chernozhukov2023nearly} which requires bounding \[ \max_v \max_{i,j} \left\vert \zeta^{(v)}_{i,j} \tilde \varepsilon_{v,i}/ \sqrt{\tilde \sigma^2_v\Sigma_{jj}^{(v)}} \right\vert.\] 
Lemma~\ref{lem:residInfBound} implies that with probability $1-o(1)$, we have
\begin{equation}
    \max_v \left\vert \hatEta\right\vert_\infty \lesssim\left(\frac{\lambda_{\min,Z}^{-1} pK \log^3(n)}{\sqrt{n}} + \log^2(n)\right).
\end{equation}
Furthermore, by Lemma~\ref{lem:varEst}, with probability $1-o(1)$ we have 
\begin{equation}
    \max_v \left\vert\frac{\tilde \sigma^2_v}{\sigma^2_v} -1\right\vert \lesssim \frac{(\overline{\sigma}^2 + \lambda_{\min, Z}^{-1})(Kp\log^2(n)/n) + \log(n)/\sqrt{n} + d^\star_2}{\underline{\sigma}^2}.
\end{equation}
When $C((\overline{\sigma}^2 + \lambda_{\min,Z}^{-1}) Kp\log^2(n)/n + \log(p)/\sqrt{n} + d^\star_2) < \underline{\sigma}^2/2$ holds, this implies
\begin{equation}
\min_v \tilde \sigma^2_v > \underline{\sigma}^2/2.
\end{equation}
Together, these imply with probability $1-o(1)$
\begin{equation}
\begin{aligned}\label{eq:boundedGammaBS}
\max_v \max_{i,j} \left\vert \zeta^{(v)}_{i,j} \tilde \varepsilon_{v,i}/ \sqrt{\sigma^2_v\Sigma_{jj}^{(v)}} \right\vert \leq {h}_{\max,1} \max_{v,i} \vert\tilde \varepsilon_{v,i} / \tilde \sigma_v \vert \lesssim \frac{{h}_{\max,1}}{\underline{\sigma}^2}\left(\frac{\lambda_{\min,Z}^{-1}pK \log^3(n)}{\sqrt{n}} +\log^2(n)\right). 
\end{aligned}
\end{equation}

Applying the first stated condition of Corollary 2.1 of \citet{chernozhukov2023nearly}, when the bound in \eqref{eq:boundedGammaBS} holds we have:
\begin{equation}
    \begin{aligned}
        I_2 &\lesssim \frac{{h}_{\max,1}\left(\frac{ \lambda_{\min,Z}^{-1}pK \log^3(n)}{\sqrt{n}} +\log^2(n)\right) \log^{3/2}(pJ)\log(n)}{\sqrt{n}\lambda_{\min,C}\underline{\sigma}^2} \\
        &\leq \frac{{h}_{\max,1}\log^{9/2}(n)}{\sqrt{n}\lambda_{\min,C}\underline{\sigma}^2} + \frac{\lambda_{\min,Z}^{-1}{h}_{\max,1}pK\log^{11/2}(n)}{n\lambda_{\min,C}\underline{\sigma}^2}.\\
    \end{aligned}
\end{equation}
When $\bias_v = 0$ for all $v$, the same result holds but we only require that $C(\overline{\sigma}^2 Kp/n + \log(n)/\sqrt{n}) < \underline{\sigma}^2/2$.

Since $\xi^{(v)}$ and $\tilde \xi^{(v)}$ are both multivariate Gaussian, to bound $I_3$ we appeal to a Gaussian-to-Gaussian comparison inequality given by Lemma 2.1 in \citet{chernozhukov2023high} which is a direct consequence of Theorem 1.1 in \citet{xiao2021HighDCLT}. In particular, the lemma states the following. Suppose $Z, Z' \in \mathbb{R}^d$ with $Z \sim N(0, S)$ and $Z' \sim N(0, S')$ where $S$ has unit diagonal entries. When
\begin{equation}
\begin{aligned}\label{eq:gaussian2Gaussian}
\sup_{R \in \mathcal{R}_v}\vert P(Z \in R)   - P(Z' \in R) \vert &\leq C\frac{\rho_1}{\sigma^2_{\star}}\log(d)\max\left(1, \vert \log(\rho_1 / \sigma^2_{\star})\vert\right)     
\end{aligned}
\end{equation}
where $\rho_1 = \vert S - S' \vert_\infty$ and $\sigma^2_{\star}$ is the smallest eigenvalue of $S$. 

Applying this to our case, let $D^{(v)}$ be a diagonal matrix whose diagonal entries coincide with $\Sigma^{(v)}$. Then $[\sigma^2_v D^{(v)}]^{-1/2}\xi^{(v)}$ has covariance $[D^{(v)}]^{-1/2}\Sigma^{(v)}D^{(v)}]^{-1/2}$ with unit diagonals by construction and $[\sigma^2_v D^{(v)}]^{-1/2}\tilde \xi^{(v)}$ has covariance $\frac{\tilde \sigma^2}{\sigma^2} [D^{(v)}]^{-1/2}\Sigma^{(v)}D^{(v)}]^{-1/2}$. Thus,
\begin{equation}
\rho_1 = \left\vert [D^{(v)}]^{-1/2}\Sigma^{(v)}[D^{(v)}]^{-1/2} - \frac{\tilde \sigma^2_v}{\sigma^2_v} [D^{(v)}]^{-1/2}\Sigma^{(v)}[D^{(v)}]^{-1/2}\right\vert_\infty = \left\vert\frac{\tilde \sigma^2_v}{\sigma^2_v} -1\right\vert.
\end{equation}
Plugging this into~\eqref{eq:gaussian2Gaussian} yields
\begin{equation}
\begin{aligned}
\sup_{R \in \mathcal{R}_v}\vert P(\xi^{(v)} \in R)   - P(\tilde \xi^{(v)} \in R) \vert &=
\sup_{R \in \mathcal{R}_v}\vert P([\sigma^2_vD^{(v)}]^{-1/2}\xi^{(v)} \in R)   - P([\sigma^2_vD^{(v)}]^{-1/2}\tilde \xi^{(v)} \in R) \vert \\
&\lesssim \frac{\vert\frac{\tilde \sigma^2_v}{\sigma^2_v} -1\vert}{\lambda_{\min,C}}\log(pJ)\max\left(1, \vert \log(\vert\frac{\tilde \sigma^2_v}{\sigma^2_v} -1\vert / \lambda_{\min,C})\vert\right).    
\end{aligned}
\end{equation}

Again, by Lemma~\ref{lem:varEst}, when $d_{v,i} \neq 0$, we have
\begin{equation}
   \max_v \left\vert\frac{\tilde \sigma^2_v}{\sigma^2_v} -1\right\vert \lesssim \frac{(\overline{\sigma}^2 + \lambda_{\min, Z}^{-1})(Kp\log^2(n)/n) + \log(n)/\sqrt{n} + d^\star_2}{\underline{\sigma}^2}.
\end{equation}
Note that
\begin{equation}\scriptsize
\begin{aligned}
\left\vert \log\left(\frac{(\overline{\sigma}^2 + \lambda_{\min, Z}^{-1})(Kp\log^2(n)/n) + \log(n)/\sqrt{n} + d^\star_2}{\underline{\sigma}^2\lambda_{\min,C}}\right)\right\vert & \leq    \left\vert \log( (Kp\log^2(n)/n) + \log(n)/\sqrt{n} + d^\star_2) + \log\left(\frac{\overline{\sigma}^2 + \lambda_{\min, Z}^{-1}}{\underline{\sigma}^2\lambda_{\min,C}}\right)\right\vert \\
&\leq  \left\vert \log\left( (Kp\log^2(n)/n) + \log(n)/\sqrt{n} + d^\star_2\right) \right\vert+ \left\vert \log\left(\frac{\overline{\sigma}^2 + \lambda_{\min, Z}^{-1}}{\underline{\sigma}^2 \lambda_{\min,C}}\right)\right\vert .
\end{aligned}
\end{equation}
Since we assume $[(Kp\log^2(n)/n) + \log(n)/\sqrt{n} + d^\star_2] < 1$, 
\begin{equation}\footnotesize
\begin{aligned}
 \left\vert \log( (Kp\log^2(n)/n) + \log(n)/\sqrt{n} + d^\star_2) \right\vert <  \left\vert \log(Kp\log^2(n)/n) \right\vert =  \left\vert \log(Kp\log^2(n)) - \log(n) \right\vert \leq 2\log(n).
\end{aligned}
\end{equation}
Putting everything together, we have with probability $1 - o(1)$, 
\begin{equation}
\begin{aligned}
I_3 &= \sup_{R \in \mathcal{R}_v}\vert P(\xi^{(v)} \in R)   - P(\tilde \xi^{(v)} \in R) \vert \\
&\lesssim \frac{\overline{\sigma}^2 + \lambda_{\min, Z}^{-1}}{\underline{\sigma}^2\lambda_{\min,C}} [(Kp\log^2(n)/n) + \log(n)/\sqrt{n} + d^\star_2]\log(n)\left[\log(n) +   \log\left(\frac{\overline{\sigma}^2 + \lambda_{\min, Z}^{-1}}{\underline{\sigma}^2\lambda_{\min,C}}\right)\right].
\end{aligned}
\end{equation}
When, $d_{v,i} = 0$ for all $v$ and $i$, then with probability $1-o(1)$,
\begin{equation}
    \max_v \left\vert\frac{\tilde \sigma^2_v}{\sigma^2_v} -1\right\vert \lesssim \frac{\overline{\sigma}^2}{\underline{\sigma}^2}  ( (Kp/n) + \log(p)/\sqrt{n}).
\end{equation}
One can similarly show that 
\begin{equation}
\begin{aligned}
I_3
&\lesssim \frac{\overline{\sigma}^2}{\underline{\sigma}^2\lambda_{\min,C}} [(Kp/n) + \log(n)/\sqrt{n}]\log(n)[\log(n) +   \log(\frac{\overline{\sigma}^2}{\underline{\sigma}^2\lambda_{\min,C}})].
\end{aligned}
\end{equation}

Putting everything together, when $\bias_v \neq 0$ for some $v$, we have with probability $1-o(1)$:
\begin{equation} \footnotesize
\begin{aligned}
    I_1 + I_2 + I_3 & \lesssim
     \delta_1\sqrt{\log(pJ)} + \frac{{h}_{\max,1}^4M^4\log^{3}(n)}{\sqrt{n}\lambda_{\min,C}}\\
    & \quad
    + \frac{h_{\max,1}\log^{9/2}(n)}{\sqrt{n}\lambda_{\min,C}\underline{\sigma}^2} + \frac{\lambda_{\min,Z}^{-1}h_{\max,1}pK\log^{11/2}(n)}{n\lambda_{\min,C}\underline{\sigma}^2}\\
    &\quad
    + \frac{\overline{\sigma}^2 + \lambda_{\min, Z}^{-1}}{\underline{\sigma}^2\lambda_{\min,C}} [(Kp/n) + \log(p)/\sqrt{n} + d^\star_2]\log(n)\left[\log(n) +   \log\left(\frac{\overline{\sigma}^2 + \lambda_{\min, Z}^{-1}}{\underline{\sigma}^2\lambda_{\min,C}}\right)\right]\\
    &\lesssim
    \delta_1\sqrt{\log(pJ)}  + d^\star_2\left( \frac{\log^2(n)(\overline{\sigma}^2 + \lambda_{\min, Z}^{-1})}{\lambda_{\min,C} \underline{\sigma}^2}\right)\left(1 + \log\left[\frac{\overline{\sigma}^2 + \lambda_{\min, Z}^{-1}}{\lambda_{\min,C} \underline{\sigma}^2}\right] \right)\\
    &\quad + 
    \left(n^{-1/2} + \frac{Kp}{n}\right) \left(\frac{ h_{\max,1}^4 M^4 \log^{11/2}(n) (\overline{\sigma}^2 + \lambda_{\min, Z}^{-1})}{\lambda_{\min,C} \underline{\sigma}^2}\right)\left(1 + \log\left[\frac{\overline{\sigma}^2 + \lambda_{\min, Z}^{-1}}{\lambda_{\min,C} \underline{\sigma}^2}\right] \right). 
\end{aligned}
\end{equation}

Similarly, when $\bias_v = 0$ for all $v$ we have with probability $1-o(1)$:
\begin{equation}
\begin{aligned}
    I_1 + I_2 + I_3 & \lesssim \frac{{h}_{\max,1}^4M^4\log^{3}(n)}{\sqrt{n}\lambda_{\min,C}}\\
    & \quad
    + \frac{{h}_{\max,1}\log^{9/2}(n)}{\sqrt{n}\lambda_{\min,C}\underline{\sigma}^2} + \frac{\lambda_{\min,Z}^{-1}{h}_{\max,1}pK\log^{11/2}(n)}{n\lambda_{\min,C}\underline{\sigma}^2}\\
    &\quad
    + \frac{\overline{\sigma}^2}{\underline{\sigma}^2\lambda_{\min,C}} [(Kp/n) + \log(n)/\sqrt{n}]\log(n)\left[\log(n) +   \log\left(\frac{\overline{\sigma}^2}{\underline{\sigma}^2\lambda_{\min,C}}\right)\right]\\
    &\lesssim
    \left(n^{-1/2} + \frac{Kp}{n}\right)\frac{h_{\max,1}^4M^4 \log^{11/2}(n) \overline{\sigma}^2}{\lambda_{\min,Z}\lambda_{\min,C}\underline{\sigma}^2}\left(1 + \log\left[\frac{\overline{\sigma}^2
}{ \underline{\sigma}^2 \lambda_{\min,C}}\right]\right).
\end{aligned}
\end{equation}
\end{proof}
\clearpage

\subsection{Proof of Lemma~\ref{lem:totalPValClose}}

Suppose $\theta \in \Theta(G)$. Then,
$\vert \Gamma_\theta - \hat \Gamma_\theta \vert \leq p \max_v\vert \gamma_{\theta,v}- \hat \gamma_{\theta,v}\vert$.

\begin{proof}
Note that $\Gamma_\theta = 1 - [1-\gamma_\theta]^{p-1}$ and
$\hat \Gamma_\theta = 1 - [1-\hat \gamma_\theta]^{p-1}$. Thus, we upper bound the difference between $\Gamma_\theta$ and $\hat \Gamma_\theta$ by
\begin{equation}\label{eq:totalPValClose}
    \begin{aligned}
        \vert \Gamma_\theta - \hat \Gamma_\theta \vert &= \vert 1 - [1 - \gamma_{\theta}]^{p-1} - (1 - [1 - \hat \gamma_{\theta}]^{p-1}) \vert\\
        & = \vert [1 - \hat \gamma_{\theta}]^{p-1} -[1 - \gamma_{\theta}]^{p-1} \vert\\
        & \stackrel{(1)}{\leq} (p-1) \max_{g \in (\gamma_{\theta}, \hat \gamma_{\theta})} [1 - g]^{p-2} \vert \gamma_{\theta}- \hat \gamma_{\theta}\vert  \\
        & \leq (p-1) \max_v\vert \gamma_{\theta,v}- \hat \gamma_{\theta,v}\vert,
    \end{aligned}
\end{equation} 
where (1) comes from applying the mean value theorem.
\end{proof}

\subsection{Proof of Corollary~\ref{cor:finalValidLinear}}

For a fixed $\theta \in \Theta(G)$, suppose that the conditions in Theorem~\ref{thm:highDCLT} hold. Furthermore, suppose the data is known to be generated by a linear structural equation model so $K = 1$ and $d_{v,i} = 0$ for all $v$ and $i$. When  $\underline{\sigma}^2, \overline{\sigma}^2, \lambda_{\min,C}, \lambda_{\min,Z}, M, h_{\max, 1}$ are fixed and $p^2\log^{11/2}(n)/n \rightarrow 0$, then $\hat \Gamma_{\theta} \rightarrow_p \Gamma_{\theta}$ and
\begin{equation}
\lim_{n \rightarrow \infty} P(\theta \in \hat \Theta(\mathbf{Y}, \alpha)) \geq 1 - \alpha. 
\end{equation}
\begin{proof}
    In the linear SEM setting, $K = 1$ and $\delta_1 = d^\star_2 = 0$ in Theorem~\ref{thm:highDCLT}. 
    Thus, the right hand side of Eq.~\eqref{eq:highDCLT1} which bounds $\vert \gamma_\theta - \hat \gamma_\theta\vert $ is of order $(n^{-1/2} + p/n)\log^{9/2}(n)$. 
    By Lemma~\ref{lem:totalPValClose}, the p-values for the entire causal ordering are then of order $p(n^{-1/2} + p/n)\log^{11/2}(n) \rightarrow 0$ when $p^2\log^{11/2}(n)/n \rightarrow 0$. Since the p-values converge and the oracle test has the correct size, the bootstrap procedure will then also have the right size asymptotically.     
\end{proof}

\subsection{Proof of Corollary~\ref{cor:finalValidNonLinear}}

For a fixed $\theta \in \Theta(G)$, suppose that the conditions in Theorem~\ref{thm:highDCLT} hold. Furthermore, suppose the data is generated by a structural equation model with unknown functions but that an approximating basis is known such that $d^\star_1\lesssim K^{-r}$, $d^\star_2\lesssim K^{-r}$, and $M_d\lesssim K^{-r}$ for some $r > 1/2$. 
Suppose $\underline{\sigma}^2, \overline{\sigma}^2, \lambda_{\min,C}, \lambda_{\min,Z}, h_{\max,1}, h_{\max,2}, M$ are fixed and let $K =[n^{3/2}/p]^{1/(r+1)}$. If $\log^{11/2}(n)n^{\frac{1-2r}{2(1+r)}}p^{\frac{1+2r}{1+r}} \rightarrow 0$, then $\hat \Gamma_{\theta} \rightarrow_p \Gamma_{\theta}$ and
\begin{equation}\small
\lim_{n \rightarrow \infty} P(\theta \in \hat \Theta(\mathbf{Y}, \alpha)) \geq 1 - \alpha. 
\end{equation}
\begin{proof}
Compared to the well-specified setting in Corollary~\ref{cor:finalValidLinear}, we have two additional terms in the upper bound of $\vert \gamma_{\theta} - \hat \gamma_{\theta}\vert$ due to the bias of which are of order $\max_v \vert\nu^{(v)}\vert \sqrt{\log(pJ)}  + d^\star_2\log^2(n)$. In addition, we let $K$ grow with $n$. 

By Lemma~\ref{lem:misspecificationBound}, with probability $1-o(1)$
\begin{equation}
\begin{aligned}
    \delta_1 = \max_v\left \vert \nu^{(v)}\right \vert_\infty &\leq h_{\max,2}(\sqrt{n}K^{-r} + \log(p)K^{-r}).
\end{aligned}
\end{equation}
Thus, combining the bounds in Theorem~\ref{thm:highDCLT} and Lemma~\ref{lem:totalPValClose} we have
\begin{equation}
\begin{aligned}\label{eq:orderNonLinearGamma}
        \vert \Gamma_\theta - \hat \Gamma_\theta \vert &\lesssim h_{\max,2}\sqrt{\log(n)}\sqrt{n}K^{-r}p + h_{\max,2}\log(p)\sqrt{\log(n)}K^{-r}p \\
    &\quad + p K^{-r}\log^2(n) + ph_{\max,1}^4 M^4\log^{11/2}(n)n^{-1/2} + h_{\max,1}^4M^4\log^{11/2}(n)Kp^2/n.
\end{aligned}
\end{equation}
Omitting log terms, and optimizing with respect to $K$, we have
\begin{equation}
\begin{aligned}
\frac{\partial}{\partial K}\sqrt{n}pK^{-r} +2K^{-r}p + pn^{-1/2} + Kp^2/n = -rK^{-(r+1)}(\sqrt{n}p + 2p) + p^2/n.
\end{aligned}
\end{equation}
The second derivative is $r(r+1)K^{-(r+2)}(\sqrt{n}p + 2p) > 0$ so solving for $K$ yields a minimizer:
\begin{equation}
\begin{aligned}
K = \left(n/p^2(\sqrt{n}p + 2p) \right)^{1/(r+1)} = O\left([n^{3/2}/p]^{1/(r+1)}]\right).
\end{aligned}
\end{equation}
Plugging $K = [n^{3/2}/p]^{1/(r+1)}$ back into Eq.~\eqref{eq:orderNonLinearGamma}, we have
\begin{equation}
\begin{aligned}
        \vert \Gamma_\theta - \hat \Gamma_\theta \vert &\lesssim h_{\max,2}\sqrt{\log(n)}\sqrt{n}[n^{3/2}/p]^{-r/(r+1)}p + h_{\max,2}\log(p)\sqrt{\log(n)}[n^{3/2}/p]^{-r/(r+1)}p \\
    &\quad + p [n^{3/2}/p]^{-r/(r+1)}\log^2(n) + ph_{\max,1}^4\log^{11/2}(n)n^{-1/2}\\
    &\quad + h_{\max,1}^4\log^{11/2}(n)[n^{3/2}/p]^{1/(r+1)}p^2/n\\
    &\lesssim  h_{\max,2}\log^{1/2}(n) n^{\frac{1-2r}{2(1+r)}}p^{\frac{1+2r}{1+r}} + h_{\max,2}\log^2(n)n^{\frac{-3r}{2(1+r)}}p^{\frac{1+2r}{1+r}} \\
    &\quad + h_{\max,1}^4 \log^{11/2}(n)pn^{-1/2} +  h_{\max,1}^4\log^{11/2}(n) n^{\frac{1-2r}{2(1+r)}}p^{\frac{1+2r}{1+r}} .
\end{aligned}
\end{equation}
When $\log^{11/2}(n)n^{\frac{1-2r}{2(1+r)}}p^{\frac{1+2r}{1+r}}\rightarrow 0$ then $p\log^{11/2}(n)/\sqrt{n}$ and $\log^2(n)n^{\frac{-3r}{2(1+r)}}p^{\frac{1+2r}{1+r}}$ also go to $0$. Thus, the entire upper bound goes to $0$, so $\Gamma_\theta \rightarrow_p \hat \Gamma_\theta$. Since the oracle procedure has exact size, then the bootstrap procedure will achieve nominal size asymptotically.  
\end{proof}

\clearpage

\subsection{Proof of Theorem~\ref{thm:power}}
Fix an ordering $\theta \not \in \Theta(G)$ and $v \in V$. Suppose Assumptions~\ref{asm:designCondition}, \ref{asm:testFuncWellCondition}, and \ref{asm:altResid} hold and $(pK)^2/n \rightarrow 0$. When $\lambda_{\min,Z}^{-1}\max(h_{\max,3}, h_{\max, 2})pK\log^4(n)/\sqrt{n} = o(\tau^\star)$, then an $\alpha$-level test for $H_{0,\theta, v}$ will be rejected with probability $1-o(1)$ for any $\alpha \in (0,1)$. 
\begin{proof}
Let $\hat \gamma_{\theta,v} = P_{\tilde \eta}\left(T(\tildevecv, \pr(v); \vecY) \geq T(\vecv, \pr(v); \vecY)  \mid \vecY \right)$ be the p-value resulting from Alg.~\ref{alg:fullNGProcedure}, where $P_{\tilde \eta}$ denotes the probability under the bootstrap distribution conditional on $\vecY$. Then by Markov's inequality we have
\begin{equation}
    \begin{aligned}
P_{\tilde \eta} &\left(T(\tildevecv, \pr(v); \vecY) \geq T(\vecv, \pr(v); \vecY)  \mid \vecY \right) \\
& \leq P_{\tilde \eta}\left( \left \vert T(\tildevecv , \pr(v); \vecY)- T(\vecv , \pr(v); \vecY) + \sqrt{n}\tau^\star \right\vert   \geq \sqrt{n}\tau^\star \mid \vecY  \right)\\
& \leq\frac{\E_{\tilde \eta}\left[ \left\vert T(\vecv , \pr(v); \vecY) - \tau^\star  - T(\tildevecv , \pr(v); \vecY) \right\vert \mid \vecY \right]}{ \sqrt{n}\tau^\star }\\
    & \leq \frac{\vert T(\vecv , \pr(v); \vecY) - \sqrt{n}\tau^\star\vert + \E_{\tilde \eta}\left(\left\vert T(\tildevecv , \pr(v); \vecY\right\vert  
 \mid \mathbf{Y}\right)}{\sqrt{n}\tau^\star}.
    \end{aligned}
\end{equation}

Considering the first term,
\begin{equation}
\begin{aligned}
    \vert T(\vecv , \pr(v); \vecY) - \sqrt{n}\tau^\star\vert &= \sqrt{n}\left\vert  \max_{j,u}\vert\frac{1}{\sqrt{n}}\tau_j(\vecY, u, \pr(v), \vecY) \vert  - \max_{j, u} \vert\E(h_j(Y_{u,i})\eta_{i})\vert\right\vert \\
    &\leq \sqrt{n} \max_{j,u}\left\vert\frac{1}{\sqrt{n}}\tau_j(\vecY, u, \pr(v), \vecY) - \E(h_j(Y_{u,i})\eta_{i})\right\vert.
\end{aligned}
\end{equation}

By Lemma~\ref{lem:power_testStat},
\begin{equation}
	\max_{u,j} \vert \frac{1}{\sqrt{n}}\tau_j(\vecY, u, \pr(v), \vecY) - \E(h_{j,u,i}\eta_i)\vert \lesssim \lambda_{\min,Z}^{-1}h_{\max,3} pK\log^3(n)/\sqrt{n}.
\end{equation} so $\vert T(\vecv , \pr(v); \vecY) - \sqrt{n}\tau^\star\vert < \lambda_{\min,Z}^{-1}h_{\max,3} pK\log^3(n)$ with probability $1-o(1)$. Furthermore, by Lemma~\ref{lem:power_expectTilde}, we have that $\E_{\tilde \eta}\left(\left\vert T(\tildevecv , \pr(v); \vecY\right\vert  
 \mid \mathbf{Y}\right) \lesssim \lambda_{\min,Z}^{-1}h_{\max,2}\log^4(n)$ with probability $1-o(1)$. Thus, the entire numerator is bounded above by \[\lambda_{\min,Z}^{-1}\max(h_{\max,2}, h_{\max,3})\log^4(n)pK.\] Thus, when $(\lambda_{\min,Z}^{-1}pKh_{\max}\log^4(n))/\sqrt{n} =o(\tau^\star)$ the $\hat \gamma_{\theta, v} \rightarrow_p 0$ so for any $\alpha$, the test will be rejected with probability $1-o(1)$.
\end{proof}

\clearpage

\subsection{Supporting Lemmas}

\begin{lemma}\label{lem:designConditionBound}
Suppose Assumption~\ref{asm:designCondition} holds and $pK/n \rightarrow 0$. Then with probability $1-o(1)$,
\begin{equation}\label{eq:wellCondDesign}
\min_{v}\lambda_{\min}\left(\frac{1}{n}\B_{v}^T\B_{v}\right) \geq \lambda_{\min,Z}/2
\end{equation}
and
\begin{equation}\label{eq:boundedDesign2norm}
\max_{v} \max_{i} \left\vert Z_{v,i} \left(\frac{1}{n}\B_{v}^T\B_{v}\right)^{-1} \right\vert_2 \lesssim \sqrt{pK}\log^2(n)\lambda_{\min,Z}^{-1}.
\end{equation} 
\end{lemma}
\begin{proof}
We first show Eq.~\eqref{eq:wellCondDesign} using Theorem 1.1 of \citet{oliveira2016lower} which we restate below for completeness below.

Let $X_i$ for $i = 1, \ldots, n$ be i.i.d copies with finite fourth moments and $\Sigma := \E(X_iX_i^T)$. Suppose that 
\[\forall s \in \mathbb{R}^p \; : \; \sqrt{\E\left[(s^TX)^4)\right]} \leq \omega s^T\Sigma s\]
for some $\omega > 1$. Then if $n \geq 81\omega^2(p + 2 \log(2/\delta))/\varepsilon^2$, we have
\[P\left(\forall s \in \mathbb{R}^p \; : \; s^T\hat \Sigma s \geq (1-\varepsilon) s^T \Sigma s\right) \geq 1 - \delta.\]

Since $Z_{v,i} \in \mathbb{R}^{\vert \pr(v)\vert K}$, letting $\varepsilon = 1/2$ implies that when $n \geq 324\omega^2(pK + 2 \log(2/\delta))$
\[P\left(\min_v\lambda_{\min}\left(\frac{1}{n}\B_v^T\B_v\right) \leq \lambda_{\min,Z}/2\right) \leq p\delta.\]

Letting $\delta = \exp(-n/1296\omega^2)$, we have
\[324\omega^2(pK + 2 \log(2/\exp(-n/1296\omega^2))) = 324\omega^2(pK + 2 \log(2)) + n/2. \]
Thus, 
\begin{equation}\label{eq:minEigenSample}
n \geq 628\omega^2(pK + 2 \log(2))
\end{equation}
implies $n \geq 324\omega^2(pK + 2 \log(2/\exp(-n/1296\omega^2)))$ and  
\[P\left(\min_v \lambda_{\min}\left(\frac{1}{n}\B_v^T\B_v\right) \leq \lambda_{\min,Z}/2\right) \leq p\exp(-n/1296\omega^2).\]
Thus, if $pK/n\rightarrow 0$, \eqref{eq:minEigenSample} will be satisfied so that
\eqref{eq:wellCondDesign} holds with probability $1-o(1)$. Since $\max_{v,k}\Vert Z_{v,i,k} \Vert_{\psi_1} < M$ we have $\max_{v,i,k} \vert Z_{v,i,k} \vert < \log^2(pKn) < 4\log^2(n)$ with probability $1-o(1)$. Thus, combining everything together we have with probability $1-o(1)$,
\begin{equation}
    \begin{aligned}
        \max_{v} \max_{i} \left\vert Z_{v,i} \left(\frac{1}{n}\B_{v}^T\B_{v}\right)^{-1} \right\vert_2 &\leq \sqrt{pK}\max_{v,i}  \left\vert Z_{v,i} \right\vert_\infty \left\Vert \left(\frac{1}{n}\B_{v}^T\B_{v}\right)^{-1} \right\Vert_2 \\
        &\lesssim \sqrt{pK}\log^2(n)\lambda_{\min,Z}^{-1}.
    \end{aligned}
\end{equation}

\end{proof}

\clearpage

\begin{lemma}\label{lem:misspecificationBound}
Suppose, $\max_v \E(\vert d_{v,i}\vert) = d^\star_1$ and $\max_v \Vert d_{v,i} \Vert_{\Psi_1} = M_d$ and \[\max_v \max_i \left \vert \zeta_{i}^{(v)} \right \vert_\infty \leq h_{\max,2}. \]
Then with probability $1-o(1)$, we have
\begin{equation}
\begin{aligned}
    \max_v\left \vert\frac{1}{\sqrt{n}} \Hv^T(I-\B_v(\B_v^T\B_v)^{-1}\B_v)\bias_v \right \vert_\infty &\leq h_{\max,2}(\sqrt{n}d^\star_1 + \log(p)M_d)
\end{aligned}
\end{equation}

\end{lemma}
\begin{proof}
Recall that $\Hv \in \mathbb{R}^{n \times \vert \pr(v)\vert J}$ is the matrix where each row is the test functions evaluated on $Y_{\pr(v),i}$; i.e., the $i$th row is $(h_j(Y_{u,i}) \,:\,  j \in [J], u \in \pr(v))$. Furthermore, $\zeta_{i}^{(v)}$ is the vector containing the residuals of the $i$th observation when regressing the test functions onto the approximating basis.

\begin{equation}
\begin{aligned}
    \left \vert\frac{1}{\sqrt{n}} \Hv^T(I-\B_v(\B_v^T\B_v)^{-1}\B_v)\bias_v \right \vert_\infty &= \max_{j \in [\vert\pr(v)\vert J ]} \sqrt{n}\vert  \frac{1}{n}\sum_i \zeta^{(v)}_{i,j} d_{v,i}\vert\\
    &\leq \sqrt{n} \max_{j \in [\vert\pr(v)\vert J]}  \max_{i} \vert \zeta_{i,j}^{(v)}\vert  \frac{1}{n} \vert \bias_v\vert_1 
\end{aligned}
\end{equation}
Using Corollary 1.4 of \citet{gotze2021polynomials} and letting $\delta = \log(p) M_d/\sqrt{n}$, we have
\begin{equation}
\begin{aligned}
    P\left(\max_v \frac{1}{n} \vert \bias_v\vert_1 \geq \delta + d_1^\star\right) &\leq P\left(\max_v \vert \frac{1}{n} \vert \bias_v\vert_1 - \E(\vert d_{v,i}\vert)\vert  \geq \delta \right) \\
    &\leq 2p\exp\left(-c \min\left\{\frac{n\delta^2}{M^2_d}, \frac{n\delta}{M_d} \right\}\right) \\
      &\leq 2p\exp\left(-c \min\left\{\log^2(p), \sqrt{n}\log(p) \right\}\right). 
\end{aligned}
\end{equation}
Thus, we have
\begin{equation}
\begin{aligned}
    \left \vert\frac{1}{\sqrt{n}} \Hv^T(I-\B_v(\B_v^T\B_v)^{-1}\B_v\bias_v \right \vert_\infty
    &\leq(\sqrt{n}d^\star_1 + \log(p)M_d)  h_{\max,2}.
\end{aligned}
\end{equation}
\end{proof}

\newpage

\begin{lemma}\label{lem:residInfBound}
Suppose that Assumptions~\ref{asm:subExpErrors}, \ref{asm:designCondition}, and \ref{asm:bias} hold and $pK/n\rightarrow 0$. Then, with probability $1-o(1)$
\begin{equation}
    \max_v \left\vert \hatEta\right\vert_\infty \lesssim \left(\frac{\lambda_{\min,Z}^{-1}pK \log^3(n)}{\sqrt{n}} + \log^2(n)\right). 
\end{equation}

\end{lemma}

\begin{proof}
Under the null hypothesis
\begin{equation}
\begin{aligned}
    \max_v \vert \hatEta \vert_\infty &=  \max_v \vert\vecv - \B_v \hat b_v \vert_\infty =\vert \vecepsv + \bias_v +\B_v (b_v - \hat b_v) \vert_\infty\\
    &\leq  \max_v  \vert\vecepsv\vert_\infty  +  \max_v \vert\bias_v \vert_\infty +  \max_v \vert \B_v (b_v - \hat b_v) \vert_\infty\\
\end{aligned}
\end{equation}
Since $\varepsilon_{v,i}$ and $d_{v,i}$ are both mean $0$ with Orlicz-1 norm bounded by $M$ and $M_d$ respectively, the first and second terms are bounded by $\log^2(pn)\leq 4\log^2(n)$ with probability $1 - o(1)$. Now consider the third term. For fixed $v$, using Lemma~\ref{lem:designConditionBound} we have
\begin{equation}\footnotesize
\begin{aligned}
    \vert \B_v (b_v - \hat b_v) \vert_\infty &=  \vert\B_v(\B_v^T\B_v)^{-1}\B_v^T(\vecepsv + \bias_v)  \vert_\infty \\
    & = \max_i\vert Z_{v,i}(\B_v^T\B_v)^{-1}\B_v^T(\vecepsv + \bias_v) \vert
   \leq  \max_i\left\vert Z_{v,i}\left(\frac{1}{n}\B_v^T\B_v\right)^{-1}\right\vert_2 \left\vert\frac{1}{n}\B_v^T(\vecepsv + \bias_v) \right\vert_2\\
    & \lesssim \lambda_{\min,Z}^{-1}\sqrt{ pK} \log^2(n)\left\vert\frac{1}{n}\B_v^T(\vecepsv + \bias_v) \right\vert_2.
\end{aligned} 
\end{equation}
By assumption $\varepsilon_{v,i} \independent Z_{v,i}$ so $\E(\varepsilon_{v,i}Z_{v,i,k}) = 0$. Furthermore, because $d_{v,i}$ is the bias of the least squares estimator, we also have $\E(d_{v,i}Z_{v,i,k}) = 0$. In addition, $\vert (\varepsilon_{v,i} + d_{v,i})Z_{v,i,k} \vert_{\psi_{1/2}} \leq \vert \varepsilon_{v,i} + d_{v,i} \vert_{\psi_{1}}  \vert Z_{v,i,k} \vert_{\psi_{1}} \leq 2M^2$.

Let $\overline{(\bias_v + \varepsilon_v)Z_{v,k}} = \frac{1}{n}\sum_i (d_{v,i} + \varepsilon_{v,i}) Z_{v,i,k}$
Then, by Corollary 1.4 of \citet{gotze2021polynomials}, we have for $\delta = \log(n)/\sqrt{n}$
\begin{equation}
\begin{aligned}
    P\left(\max_{v \in V, k \in [\vert\pr(v) \vert K]}\vert \overline{(\bias_v + \varepsilon_v)Z_{v,k}}  \vert \geq \delta\right) &\leq 2 \exp\left( - c \min\left(\frac{\delta^2}{4M^4 (1/n)}, \left[\frac{\delta}{2M^2(1/n)}\right]^{1/2}\right) \right) \\
    & = 2p^2K \exp\left( - c \min\left(\frac{\log^2(n) / n}{4M^4 (1/n)}, \left[\frac{\log(n)/\sqrt{n}}{2M^2(1/n)}\right]^{1/2}\right) \right) \\ 
     & \leq\exp\left(2\log(n) - c \min\left(\frac{\log^2(n) }{4M^4}, \left[\frac{\log(n)\sqrt{n}}{2M^2}\right]^{1/2}\right) \right). 
\end{aligned}
\end{equation}
Since $\left\vert\frac{1}{n}\B_v^T(\vecepsv + \bias_v) \right\vert_2 \leq \sqrt{pK}\left\vert\frac{1}{n}\B_v^T(\vecepsv + \bias_v) \right\vert_\infty$, we have with probability $1-o(1)$
\begin{equation}\footnotesize
\begin{aligned}
    \vert \B_v (b_v - \hat b_v) \vert_\infty & \lesssim   \lambda_{\min,Z}^{-1}pK \log^2(n)\left\vert\frac{1}{n}\B_v^T(\vecepsv + \bias_v) \right\vert_\infty    \lesssim \lambda_{\min,Z}^{-1} \frac{ pK \log^3(n)}{\sqrt{n}}.
\end{aligned} 
\end{equation}
Putting all the terms back together, we have with probability $1- o(1)$
\begin{equation}
\begin{aligned}
    \max_v \vert \hatEta \vert_\infty &\lesssim  \left(\frac{ \lambda_{\min,Z}^{-1}pK \log^3(n)}{\sqrt{n}} + \log^2(n)\right). 
\end{aligned}
\end{equation}

\end{proof}

\newpage

\begin{lemma}\label{lem:varEst}
Suppose Assumptions~\ref{asm:subExpErrors}, \ref{asm:designCondition}, \ref{asm:bias}, and \ref{asm:testFuncWellCondition} hold and $pK/n\rightarrow 0$. Let $\tilde \sigma^2_v$ denote the empirical variance of $\hatEta$. If $\bias_v = 0$ for all $v$, we have with probability $1- o(1)$,
  \begin{equation}
\max_v \vert \tilde \sigma^2_v - \sigma^2_v \vert \lesssim \overline{\sigma}^2 (Kp/n) + \log(n)/\sqrt{n}.
  \end{equation}
If $\bias_v \neq 0$ for some $v$, then with probability $1-o(1)$,
   \begin{equation}
\max_v \vert \tilde \sigma^2_v - \sigma^2_v \vert \lesssim (\overline{\sigma}^2 + \lambda_{\min,Z}^{-1}) (Kp\log^2(n)/n) + \log(n)/\sqrt{n} + d^\star_2.
  \end{equation} 
\end{lemma}

\begin{proof}
Let $\hat \sigma^2_v$ denote the sample variance of $\vecepsv$. Then, 
\[\max_v \vert\sigma^2_v - \tilde \sigma^2 \vert \leq \max_v\vert\sigma^2_v - \hat \sigma^2_v \vert + \max_v\vert\hat \sigma^2_v - \tilde \sigma^2 \vert. \]
The first term is simply the deviation of the sample variance (which we could calculate had we observed the errors) from the population variance. Since $\varepsilon_{v,i}$ is sub-exponential, we may apply Proposition 1.1 of \citet{gotze2021polynomials} using $\delta = \log(n)/\sqrt{n}$ so that
\begin{equation}
\begin{aligned}\label{eq:varPart1}
P\left(\max_v\left\vert\sigma^2_v - \frac{1}{n}\vert \vecepsv \vert_2^2 \right\vert \geq \delta \right) &\leq 2p\exp( -c_1 \min\left(\min\left\{\frac{n\delta^2}{M^4}, \left(\frac{n\delta}{M^2} \right)^{1/2}\right\}\right)\\
&\leq 2p\exp\left( -c_1 \min\left\{\frac{\log^2(n)}{M^4}, \left(\frac{\sqrt{n}\log(n)}{M^2} \right)^{1/2}\right\}\right).     
\end{aligned}
\end{equation}
The upper bound goes to $0$ as $p, n\rightarrow \infty$. For the second term, note that

\begin{equation}
\begin{aligned}\label{eq:varDiffBound1}
    \vert\hat \sigma^2_v - \tilde \sigma^2_v \vert &= \frac{1}{n}\left\vert \vert \vecepsv \vert_2^2 -  \vert \hatEta \vert_2^2 \right\vert = \frac{1}{n}\left\vert \vert \vecepsv \vert_2^2 - \vert \vecepsv + \B_vb_v - \B_v\hat b_v + \bias_v\vert_2^2 \right\vert\\
    &= \frac{1}{n}\left\vert \vert \vecepsv \vert_2^2 - \vert \vecepsv\vert_2^2 - \vert \B_v(b_v - \hat b_v) +\bias_v\vert_2^2 - 2\vecepsv^T(\B_v(b_v - \hat b_v) +\bias_v) \right\vert\\
    & \leq \frac{3}{n} \vert \B_v(b_v - \hat b_v) \vert_2^2 +\frac{3}{n}\vert \bias_v\vert_2^2 + \frac{2}{n} \vert\vecepsv^T(\B_v(b_v - \hat b_v) \vert +  \frac{2}{n} \vert\vecepsv^T\bias_v \vert.
\end{aligned}
\end{equation}
Because $\B_v(b_v - \hat b_v) = \B_v(\B_v^T\B_v)^{-1}\B_v^T(\vecepsv + \bias_v)$, we have
\begin{equation}
\begin{aligned}
    \vert \B_v(b_v - \hat b_v) \vert_2^2 &\leq 3\vert \B_v(\B_v^T\B_v)^{-1}\B_v^T\vecepsv \vert_2^2 + 3\vert \B_v(\B_v^T\B_v)^{-1}\B_v^T\bias_v \vert_2^2\\
    &\leq 3\vecepsv^T \B_v(\B_v^T\B_v)^{-1} \B_v^T\B_v(\B_v^T\B_v)^{-1}\B_v^T\vecepsv +3\vert \bias_v \vert_2^2 \\
    &=3\vecepsv^T \B_v(\B_v^T\B_v)^{-1}\B_v^T\vecepsv + 3\vert \bias_v \vert_2^2.
    \end{aligned}
\end{equation}
In addition, 
\begin{equation}
\begin{aligned}
    \vecepsv^T\B_v(b_v - \hat b_v) &=\vecepsv^T\B_v(\B_v^T\B_v)^{-1}\B_v^T(\vecepsv + \bias_v)\\
    &= \vecepsv^T\B_v(\B_v^T\B_v)^{-1}\B_v^T\vecepsv + \vecepsv^T\B_v(\B_v^T\B_v)^{-1}\B_v^T\bias_v.
    \end{aligned}
\end{equation}

Thus, we can upper bound Eq.~\eqref{eq:varDiffBound1} with
\begin{equation}\footnotesize
\begin{aligned}
    \max_v\vert\hat \sigma^2_v - \tilde \sigma^2_v \vert &\leq \max_v\left(\frac{11}{n}\vecepsv^T \B_v(\B_v^T\B_v)^{-1}\B_v^T\vecepsv + \frac{12}{n}\vert \bias_v \vert_2^2 + \frac{2}{n}\vert\vecepsv^T\B_v(\B_v^T\B_v)^{-1}\B_v^T\bias_v\vert + \frac{2}{n}\vert \vecepsv^T\bias_v\vert\right)\\
    & = I_1 + I_2 + I_3 + I_4.
\end{aligned}
\end{equation}
When, $\bias_v = 0$ for all $v$, then we are only left with the $I_1$ term. We now show that $I_1\leq \overline{\sigma}^2 (Kp/n) + \log(n)/\sqrt{n}$ with probability $1- o(1)$. Let $\mathbf{A}_v = \B_v(\B_v^T\B_v)^{-1}\B_v^T$. Note that $\mathbf{A}_v$ is a projection matrix of rank $K\vert \pr(v)\vert$.
Therefore, we have $\E\left( n^{-1} \vecepsv^T \mathbf{A}_v \vecepsv \right)  \leq \sigma^2_v (Kp) / n$, $\Vert \frac{1}{n}\mathbf{A}_v\Vert_F^2 \leq (Kp)/n^2$, and $\Vert \frac{1}{n}\mathbf{A}_v\Vert_2 = 1/n$. We then apply a Hansen-Wright Inequality for sub-exponential random variables from Proposition 1.1 of \citet{gotze2021polynomials}, so that for some constant $c_1$ when $\delta = \log(n)/\sqrt{n}$:
\begin{equation}\label{eq:varPart2}
\begin{aligned}
P&\left(\max_{v \in V} \frac{1}{n}\vecepsv^T \mathbf{A}_v \vecepsv \geq \overline{\sigma}^2(Kp)/n +  \delta \right) \leq P\left(\max_{v \in V} \left \vert \frac{1}{n}\vecepsv^T \mathbf{A}_v \vecepsv - \E\left( n^{-1} \vecepsv^T \mathbf{A}_v \vecepsv \right)
\right \vert \geq \delta\right) \\
&\leq 2p 
\exp\left(-\frac{c_1}{ M_1^4} \min\left\{\frac{\delta^2}{M^4 \max_v\Vert \frac{1}{n}\mathbf{A}_v\Vert_F^2}, \left(\frac{\delta}{M^2\max_v\Vert \frac{1}{n}\mathbf{A}_v\Vert_2} \right)^{1/2}\right\}\right) \\
&\leq 2p 
\exp\left(-c_1 \min\left\{\frac{n^2\delta^2}{M^4 Kp}, \left(\frac{n\delta}{M^2} \right)^{1/2}\right\}\right)\\
& = 2p\exp\left(\log(n)-c_1 \min\left\{\frac{n\log^2(n)}{M^4 Kp}, \left(\frac{\sqrt{n}\log(n)}{M^2} \right)^{1/2}\right\}\right).
\end{aligned}
\end{equation}

We now show that $I_2\leq d^\star_2 + \log(n)/\sqrt{n}$ with probability $1- o(1)$. Because $d_{v,i}$ is assumed to be sub-exponential, $d_{v,i}^2$ is sub-weibull 1/2. Thus, letting $\delta = \log(n)/\sqrt{n}$, we have
\begin{equation}
\begin{aligned}
    P\left(\max_{v \in V} \frac{1}{n}\vert d\vert_2^2 \geq d^\star_2 + \delta\right) &\leq P\left(\max_{v \in V} \vert \overline{d_{v,i}^2} - \E(d_{v,i}^2) \vert \geq \delta\right)  \\
&\leq 2p \exp\left( - c \min\left(\frac{\delta^2}{4M^4 (1/n)}, \left[\frac{\delta}{2M^2(1/n)}\right]^{1/2}\right) \right) \\
     & \leq2\exp\left(\log(n) - c \min\left(\frac{\log^2(n) }{4M^4}, \left[\frac{\log(n)\sqrt{n}}{2M^2}\right]^{1/2}\right) \right). 
\end{aligned}
\end{equation}

Similarly, $I_4\leq \log(n)/\sqrt{n}$ with probability $1- o(1)$. Because $\varepsilon_{v,i} \independent d_{v,i}$ and both are assumed to be sub-exponential, $\Vert \varepsilon_{v,i}d_{v,i} \Vert_{\Psi_{1/2}} < M^2$. Thus, letting $\delta = \log(n)/\sqrt{n}$, we have
\begin{equation}
\begin{aligned}
    P\left(\max_{v \in V} \overline{\varepsilon_vd_v} \geq \delta\right) 
     & \leq2\exp\left(\log(n) - c \min\left(\frac{\log^2(n) }{4M^4}, \left[\frac{\log(n)\sqrt{n}}{2M^2}\right]^{1/2}\right) \right). 
\end{aligned}
\end{equation}

Finally, note that
\begin{equation}
\begin{aligned}
        I_3 &=\frac{2}{n^2}\vert\vecepsv^T\B_v(\B_v^T\B_v/n)^{-1}\B_v^T\bias_v\vert \leq \Vert (\B_v^T\B_v / n)^{-1} \Vert_2 \vert\frac{1}{n}\vecepsv^T\B_v \vert_2\vert\frac{1}{n} \B_v\bias_v \vert_2\\
        &\leq  2\Vert (\B_v^T\B_v/n)^{-1} \Vert_2 pK\vert\frac{1}{n}\vecepsv^T\B_v \vert_\infty\vert\frac{1}{n} \B_v\bias_v \vert_\infty .
\end{aligned}
\end{equation} 

By definition, $\E(Z_{v,i,k} d_{v,i}) = 0$ and by assumption $\E(Z_{v,i,k} \varepsilon_{v,i}) = 0$. Since $Z_{v,i,k}$, $d_{v,i}$, and $\varepsilon_{v,i}$ are all sub-exponential, we have that $\Vert Z_{v,i,k} d_{v,i} \Vert_{\Psi_{1/2}} < M^2 $ and $\Vert Z_{v,i,k} \varepsilon_{v,i} \Vert_{\Psi_{1/2}} < M^2$. Thus, using Corollary 1.4 of \citet{gotze2021polynomials} and letting $\delta = \log(n)/\sqrt{n}$, we have
\begin{equation}
\begin{aligned}
    P\left(\max_{v \in V}\max_{k \in [|\pr(v)|J]} \vert\frac{1}{n}\sum_i Z{v,i,k} d_{v,i} \vert \geq \delta\right) &\leq 2p^2J \exp\left( - c \min\left(\frac{\delta^2}{M^4 (1/n)}, \left[\frac{\delta}{M^2(1/n)}\right]^{1/2}\right) \right) \\
    &\leq 2p^2J \exp\left( - c \min\left(\frac{\log^2(n)}{M^4 }, \left[\frac{\sqrt{n}\log(n)}{M^2}\right]^{1/2}\right) \right),
\end{aligned}
\end{equation}
where the upper bound on the probability is $o(1)$ when $p, n \rightarrow \infty$. Similarly, \[\max_{v \in V}\max_{k \in [|\pr(v)|J]} \vert\frac{1}{n}\sum_i\varepsilon_{v,i}Z^{(v)}_{i,k} \vert \leq \log(n)/\sqrt{n}\] 
with probability $1-o(1)$. By Lemma~\ref{lem:designConditionBound} we have with probability $1-o(1)$
\begin{equation}
    \begin{aligned}
        \left\Vert \left(\frac{1}{n}\B_v^T \B_v\right)^{-1} \right\Vert_2 & \leq \lambda_{\min,Z}^{-1}/2
    \end{aligned}   
\end{equation}
so that
\[I_3 \lesssim \lambda_{\min,Z}^{-1} pK \log^2(n)/n.\]

\end{proof}

\clearpage

\begin{lemma}\label{lem:power_expectTilde}
Suppose Assumptions~\ref{asm:designCondition}, \ref{asm:testFuncWellCondition} and \ref{asm:altResid} hold, and $(pK)^2/n < 1$. Then with probability $1-o(1)$, we have
	\begin{equation}\label{eq:expTildeAlt}
	\E(\tilde T^{(v)} \mid \vecY) \lesssim \lambda_{\min,Z}^{-1}h_{\max,2}\log^4(n)
	\end{equation}
\end{lemma}
\begin{proof}
We first show that with probability $1-o(1)$ (with respect to $\vecY$), $\tau_j(\tildevecv; u, \pr(v); \vecY)$ is sub-Gaussian; then~\eqref{eq:expTildeAlt} then immediately follows from well known results on the expectation of the maximum of sub-Gaussian random variables. Since $v$ and $\pr(v)$ are fixed, we will drop the sub-scripts and use $\mathbf{\eta}$ and $\eta_{i}$ for brevity. 

First, note that
\begin{equation}
	\begin{aligned}
\vert \mathbf{\hat \eta} \vert_\infty &= \vert \mathbf{\hat \eta} + \mathbf{\eta} -\mathbf{\eta} \vert_\infty\\
&\leq \vert \mathbf{\eta}  \vert_\infty + \vert \mathbf{\hat \eta} - \mathbf{\eta}   \vert_\infty.
	\end{aligned}
\end{equation}
By assumption, $\eta_i$ is sub-exponential so $\vert \mathbf{\eta} \vert_\infty \leq \log^2(n)$ with probability $1-o(1)$. In addition, by Lemma~\ref{lem:designConditionBound} we have

\begin{equation}
	\begin{aligned}
\left\vert \B(\hat b -b)\right\vert_\infty
 &= \left\vert \B[(\B^T\B)^{-1}\B^T (\B b + \bm{\eta}) - b]\right\vert_\infty \\
 &= \max_i \vert \B_i(\frac{1}{n}\B^T\B)^{-1} \frac{1}{n}\B^T\bm{\eta}) \vert \\
 & \leq  \max_i \vert \B_i(\frac{1}{n}\B^T\B)^{-1}\vert_2 \vert\frac{1}{n}\B^T\bm{\eta} \vert_2 \\
 & \lesssim \lambda_{\min,Z}^{-1}\sqrt{pK}\log^2(n) \vert\frac{1}{n}\B^T\bm{\eta} \vert_2 \\
	\end{aligned}
\end{equation}
By definition, $\E(Z_{v,i,k}\eta_i) = 0$ and by assumption $\Vert Z_{v,i,k}\eta_i \Vert_{\Psi_{1/2}} < M^2$. Thus, with probability $1-o(1)$
\begin{equation}
\max_{k} \vert \frac{1}{n}\sum_i Z_{v,i,k}\eta_i\vert < \log(pK)/\sqrt{n}
\end{equation}
so $\vert\frac{1}{n}\B^T\bm{\eta} \vert_2 < \sqrt{pK} \log(pK)/\sqrt{n}$.
Putting everything together, we have
\begin{equation}\label{eq:zB_bound}
    \vert \B(\hat b -b)\vert_\infty \lesssim \lambda_{\min,Z}^{-1}pK\log^3(n)/\sqrt{n} < \log^3(n),
\end{equation}
and
\[\vert \hat \eta \vert_\infty \lesssim \lambda_{\min,Z}^{-1}\log^3(n). \]
Note that conditional on $\vecY$, $\zeta_i$ is fixed and the only randomness in $\tau_j(\tildevecv, u, \pr(v); \vecY)$ comes in the resampling of $\tilde \varepsilon_i$ from the empirical distribution of $\mathbf{\hat \eta}$. Thus,
\begin{equation}
	\tau_j(\tildevecv, u, \pr(v); Y) = \frac{1}{\sqrt{n}}\sum_i \zeta_{i,j}^{(v)}\tilde \varepsilon_i
\end{equation}
is the sum of independent random variables with magnitude bounded by $\lambda_{\min,Z}^{-1}h_{\max,2}\log^3(n)$. Thus,  $\tau_j(\vecv, u, \pr(v); Y)$ is sub-Gaussian with variance proxy at most $\lambda_{\min,Z}^{-2}h_{\max,2}^2 \log^6(n)$. By well-known results on the expectation of the maximum of sub-Gaussian random variables (see, e.g., Section 2.5 of \citet{boucheron2013Concentration}) we have
\begin{equation}
\E(\tilde T_\infty \mid \vecY) \lesssim \lambda_{\min,Z}^{-1}h_{\max,2}\log^3(n)\log^{1/2}(pJ) < \lambda_{\min,Z}^{-1} h_{\max,2}\log^4(n).
\end{equation}
\end{proof}

\begin{lemma}\label{lem:power_testStat}
	Suppose Assumptions~\ref{asm:designCondition} and \ref{asm:altResid} hold and $pK/n\rightarrow 0$. Then with probability $1-o(1)$,
	\begin{equation}\label{eq:power_TminusTau}
	\max_{u,j} \vert \frac{1}{\sqrt{n}}\tau_j(\vecY, u, \pr(v), j) - \E(h_{j,u,i}\eta_i)\vert \lesssim \lambda_{\min,Z}^{-1}h_{\max,3} pK\log^3(n)/\sqrt{n}.
	\end{equation}
\end{lemma}
\begin{proof}
For notational convenience, let $h_{j,u,i} = h_j(Y_{u,i})$. Furthermore, since $v$ and $\theta$ are fixed, we will let $\eta_i = \eta_{v \setminus \pr(v),i}$ and $\hat \eta_i = \hat \eta_{v \setminus \pr(v),i}$.

Note that $\frac{1}{\sqrt{n}}\tau_j(\vecY, u, \pr(v), j) = \frac{1}{n}\sum_i h_{j,u,i}\hat \eta_i$ which we denote as $\overline{h_{j,u}\hat \eta_i}$. Similarly, let $\overline{h_{j,u}\eta} =  \frac{1}{n}\sum_i h_{j,u,i}\eta_i$.  
	We note that
\begin{equation}
	\begin{aligned}
		\left \vert \overline{h_{j,u}\hat \eta} - \E(h_{j,u,i}\eta_i) \right \vert \leq &\left \vert  \overline{h_{j,u}\eta} - \E(h_{j,u,i}\eta_i)  \right \vert +  \left \vert  \frac{1}{n}\sum_i (\hat \eta_i - \eta_i)h_{j,u,i}  \right \vert\\
	\end{aligned}
\end{equation}
The first term is simply a sample average of i.i.d terms where $\Vert h_{j,u,i}\eta_i\Vert_{\Psi_{1/2}} < M^2$ so \[\max_{j,u}\left \vert  \overline{h_{j,u}\eta} - \E(h_{j,u,i}\eta_i)  \right \vert \lesssim \log(pJ)/\sqrt{n}\]
 with probability $1-o(1)$. 	
 As shown in~\eqref{eq:zB_bound}, $\left\vert \B(b - \hat b) \right\vert_\infty <\lambda_{\min,Z}^{-1}pK\log^3(n)/\sqrt{n} $ with probability $1- o(1)$. Thus, to bound the second term,
\begin{equation}
	\begin{aligned}
	\left \vert  \frac{1}{n}\sum_i (\hat \eta_i - \eta_i)h_{j,u,i}  \right \vert &= 	\left \vert  \frac{1}{n}\sum_i (b - \hat b)^T Z_{v,i} h_{j,u,i}  \right \vert\\
	&\leq \left\vert \B(b - \hat b) \right\vert_\infty \frac{1}{n}\left\vert h_{j,u,i}\right\vert_1\\
	&\leq \lambda_{\min,Z}^{-1}h_{\max,3} pK\log^3(n)/\sqrt{n}. 
	\end{aligned}
\end{equation}
Since this is the dominating term, we have Eq.~\eqref{eq:power_TminusTau}.
\end{proof}

\clearpage

\section{Confidence intervals for direct effects}
\label{sec:ci-direct}
Alg.~\ref{alg:DirectEffect} describes a procedure to calculate confidence intervals for direct effects which also incorporate model uncertainty. When estimating the direct effect of $v$ onto $u$, a valid adjustment set is $\an(u)$. Thus, this procedure is exactly the same as Alg.~\ref{alg:TotalEffect} except for the definition of $\mathcal{S}$.

\begin{algorithm}[htb]
		\caption{\label{alg:DirectEffect}Get $1-\alpha$ CI for the direct effect of $v$ onto $u$}
		\begin{algorithmic}[1]
            \For{$S \in \mathcal{S} = \{S \, : \, S  = \pr(u) \text{ for some } \theta \in \hat \Theta(\mathbf{Y},\alpha / 2) \text{ such that } \theta(v) < \theta(u) \}$}
                        \State{Calculate $C(S)$, the $1-\alpha/2$ confidence interval for the coefficient of $Y_v$ when regressing $Y_u$ onto $Y_{S \cup \{v\}}$}
                \EndFor
			\If{$\theta(u) > \theta(v)$ for any $\theta \in \hat \Theta(\mathbf{Y},\alpha / 2)$}
			    \State{\textbf{Return}: $\hat C_\alpha = \{0 \} \cup \{\bigcup_{S \in \mathcal{S}} C(S)\}$}
			 \Else
			    \State{\textbf{Return}: $\hat C_\alpha = \bigcup_{S \in \mathcal{S}} C(S)$}
			 \EndIf
		\end{algorithmic}
\end{algorithm}

\FloatBarrier
\newpage

\section{Derivation of Eq.~\eqref{eq:oracleStat} and Eq.~\eqref{eq:altHyp}}
Recall that $\bm{Y_{U.1}}$ denotes the $n \times(|U| + 1)$ matrix containing $\bm{Y_U}$ as well as a column of $1$s for an intercept. Under the null hypothesis, we have that $\bm{Y_v} = \bm{Y_{U.1}}\beta_{v,U.1} + \bm{\varepsilon_v}$ where $\beta_{v,u} = 0$ for $u \in U \setminus \pa(v)$. 
Since $\hat \beta_{v,U}$ is the least squares estimator, we have that:
\begin{equation}
\begin{aligned}
\hat \beta_{v,U} &= ( \mathbf{Y}_{U.1}^T \mathbf{Y}_{U.1})^{-1} \mathbf{Y}_{U.1} \bm{Y_v}\\
&= ( \mathbf{Y}_{U.1}^T \mathbf{Y}_{U.1})^{-1} \mathbf{Y}_{U.1}^T (\bm{Y_{U.1}}\beta_{v,U.1} + \bm{\varepsilon_v})\\
&= \beta_{v,U.1} + ( \mathbf{Y}_{U.1}^T \mathbf{Y}_{U.1})^{-1} \mathbf{Y}_{U.1}^T \bm{\varepsilon_v}\\
\end{aligned}
\end{equation}
The residuals are then:
\begin{equation}
\begin{aligned}
    \bm{\hat  \eta}_{v \setminus U} &= \bm{Y_v} - \bm{Y_U} \hat \beta_{v,U} \\
    &= \bm{\varepsilon_v} + \bm{Y_U}(\beta_{v,U} - \hat \beta_{v,U})\\
    &= \bm{\varepsilon_v} - \bm{Y_U}(\mathbf{Y}_{U.1}^T \mathbf{Y}_{U.1})^{-1} \mathbf{Y}_{U.1}^T \bm{\varepsilon_v}\\
    &= [I - \bm{Y_U}(\mathbf{Y}_{U.1}^T \mathbf{Y}_{U.1})^{-1} \mathbf{Y}_{U.1}^T] \bm{\varepsilon_v}.
\end{aligned}
\end{equation}
Finally, calculating $\tau_j$, we have
\begin{equation}\small
\begin{aligned}
\tau_j(\vecv, u, U; \vecY) = \frac{1}{\sqrt{n}}h_j(\vecu)^T \bm{\hat  \eta}_{v \setminus U}  
&= \frac{1}{\sqrt{n}}h_j(\vecu)^T[I - \mathbf{Y}_{U.1}(\mathbf{Y}_{U.1}^T\mathbf{Y}_{U.1})^{-1}\mathbf{Y}_{U.1}^T]\vecepsv.
\end{aligned}
\end{equation}

Recall that $b_{v,U} = \min_b \E(Y_v - Y_U^Tb)$ is the population regression coefficients. Under the alternative when the hypothesis in~\eqref{eq:nullHypSingle} does not hold, the population regression coefficients are generally not equal to the causal coefficients; i.e., $b_{v,U} \neq \beta_{v,U}$. Letting $U' = \pa(v)  \cup  U$, with a slight abuse of notation, we define $ b_{v,  U'} = (b_{v,U})_u$ if $u \in U$ and $0$ otherwise, and similarly let $\beta_{v,  U'} = (\beta_{v,U})_u$ if $u \in \pa(v)$ and $0$ otherwise. Then, 
\[\bm{\eta}_{v \setminus U} = \bm{Y_v} - \bm{Y_{U'.1}}b_{v,U'}  =\bm{\varepsilon_v} + \bm{Y_{U'.1}}(\beta_{v,U'} - b_{v, U'}) \]
and similarly $\bm{\hat \eta}_{v \setminus U} = \bm{\varepsilon_v} + \bm{Y_{U'.1}}(\beta_{v,U'} - \hat b_{v, U'}) $ so that 
\[\bm{\hat \eta}_{v \setminus U} - \bm{\eta}_{v \setminus U} = \bm{Y_{U'}}(b_{v,U'} - \hat b_{v, U'}) .\]
Then, we have: 
\begin{equation}\small
\begin{aligned}
\tau_j(\vecv, u, U; \vecY) &= \frac{1}{\sqrt{n}}h_j(\vecu)^T \bm{\hat  \eta}_{v \setminus U}\\
&=
\frac{1}{\sqrt{n}} h_j(\vecY_u)^T\left(\bm{\hat \eta}_{v \setminus U} - \bm{\eta}_{v \setminus U} + \bm{\eta}_{v \setminus U}\right) + \sqrt{n}\E( h_j(Y_u)\eta_{v \setminus U}) - \sqrt{n}\E( h_j(Y_u)\eta_{v \setminus U})\\
&=
\frac{1}{\sqrt{n}} h_j(\vecY_u)^T\bm{\eta}_{v \setminus U} + 
\frac{1}{\sqrt{n}} h_j(\vecY_u)^T\left(\bm{\hat \eta}_{v \setminus U} - \bm{\eta}_{v \setminus U}\right) + \sqrt{n}\E( h_j(Y_u)\eta_{v \setminus U}) - \sqrt{n}\E( h_j(Y_u)\eta_{v \setminus U})\\
&=\left(\frac{1}{\sqrt{n}}h_j(\vecY_u)^T\bm{\eta}_{v \setminus U} - \sqrt{n}\E( h_j(Y_u)\eta_{v \setminus U})\right) + \frac{1}{\sqrt{n}}h_j(\vecY_u)^T\bm{Y}_{ U'.1}[b_{v,U'} - \hat{ b}_{v,U'}]\\
&\quad + \sqrt{n} \E( h_j(Y_u)\eta_{v \setminus U}).
\end{aligned}
\end{equation}

\newpage

\end{document}